\documentclass[a4paper, 11pt]{article}
\usepackage[margin=3cm]{geometry}

\usepackage{tikz}
\usetikzlibrary{shapes, arrows}
\usepackage{graphicx}
\usepackage{multicol}
\usepackage{multirow}
\usepackage{tabularx}
\usepackage{booktabs}
\usepackage{enumerate}

\usetikzlibrary{chains,
                positioning}
\usepackage{enumitem}
\usepackage{float}
\usepackage{amsmath, amsthm, amssymb, bbm}
\usepackage{bbm}
\usepackage[round]{natbib}
\usepackage{mathtools, nccmath}
\usepackage[linesnumbered, ruled]{algorithm2e}
\usepackage{algpseudocode}
\usepackage[T1]{fontenc}
\usepackage[utf8]{inputenc}
\usepackage{babel}
\usepackage{lipsum}
 \usepackage[justification=centering]{caption}

\theoremstyle{plain}
\newtheorem{theorem}{Theorem}[section]
\newtheorem{corollary}{Corollary}[theorem]
\newtheorem{lemma}[theorem]{Lemma}
\newtheorem{proposition}[theorem]{Proposition}

\theoremstyle{remark}
\newtheorem{definition}{Definition}[section]
\newtheorem{remark}{Remark}
\theoremstyle{definition}
\newtheorem{example}{Example}[section]

\newenvironment{lema}[1]{\par\noindent{\bf Lemma #1\ }\em}{\em}
\newenvironment{cora}[1]{\par\noindent{\textbf{Corollary #1} }\em}{\em}
\newenvironment{proa}[1]{\par\noindent{\textbf{Proposition #1} }\em}{\em}

\DeclareMathOperator{\head}{head}
\DeclareMathOperator{\tail}{tail}
\DeclareMathOperator{\pa}{pa}
\DeclareMathOperator{\sib}{sib}
\DeclareMathOperator{\an}{an}
\DeclareMathOperator{\de}{de}
\DeclareMathOperator{\dis}{dis}
\DeclareMathOperator{\barren}{barren}
\DeclareMathOperator{\ceil}{ceil}

\DeclareMathOperator{\mb}{mb}
\DeclareMathOperator{\ham}{ham}
\DeclareMathOperator{\nb}{nb}
\DeclareMathOperator{\ch}{ch}

\DeclarePairedDelimiter{\abs}{\lvert}{\rvert}
\newcommand{\lqarrow}{\mathbin{\leftarrow\!\!\medmath{?}}}
\newcommand{\rqarrow}{\mathbin{\medmath{?}\!\!\rightarrow}}

\newcommand*\phantomrel[1]{\mathrel{\phantom{#1}}}
\newcommand{\ei}[1]{u_{\langle #1 \rangle}}
\DeclarePairedDelimiterX{\infdivx}[2]{(}{)}{%
  #1\;\delimsize\|\;#2%
}
\newcommand{\infdiv}{KL\infdivx}

\newcommand\indep{\protect\mathpalette{\protect\independenT}{\perp}}
\def\independenT#1#2{\mathrel{\rlap{$#1#2$}\mkern2mu{#1#2}}}

\DeclareMathOperator{\BIC}{BIC}
\DeclareMathOperator{\MF}{MF}

\newcommand{\Sset}{\cal S}
\newcommand{\I}{{\cal I}}

\newcommand{\G}{{\cal G}}

\newcommand{\cmid}{\,|\,}

\tikzset{deg/.style={->, very thick, color=blue}}
\tikzset{degl/.style={->, very thick, color=red}}
\tikzset{beg/.style={<->, very thick, color=red}}

\title{Towards standard imsets for maximal ancestral graphs}

\author{{\bf Zhongyi Hu}\\Department of Statistics\\University of Oxford\\zhongyi.hu@keble.ox.ac.uk \and 
{\bf Robin J.~Evans}\\Department of Statistics\\University of Oxford\\evans@stats.ox.ac.uk}

\begin{document}

\allowdisplaybreaks

\maketitle

\begin{abstract}
The imsets of \citet{studeny2006probabilistic} are an algebraic method for representing conditional independence models.  They have many attractive properties when applied to such models, and they are particularly nice for working with directed acyclic graph (DAG) models.  In particular, the `standard' imset for a DAG is in one-to-one correspondence with the independences it induces, and hence is a label for its Markov equivalence class. We first present a proposed extension to standard imsets for maximal ancestral graph (MAG) models, using the parameterizing set representation of \citet{hu2020faster}.  
We show that for many such graphs our proposed imset is \emph{perfectly Markovian} with respect to the graph, including a class of graphs we refer to as \emph{simple} MAGs, which includes DAGs as a special case.
In these cases the imset provides a scoring criteria by measuring the discrepancy for a list of independences that define the model; this gives an alternative to the usual BIC score that is also consistent, and much easier to compute. We also show that, of independence models that do represent the MAG, the imset we give is minimal.
Unfortunately, for some graphs the representation does not represent all the independences in the model, and in certain cases does not represent any at all. For these general MAGs, we refine the reduced ordered local Markov property \citep{richardlocalmarkov} by a novel graphical tool called \emph{power DAGs}, and this results in an imset that induces the correct model and which, under a mild condition, can be constructed in polynomial time.
\end{abstract}
\section{Introduction}

\emph{Maximal ancestral graphs} (MAGs) \citep{richardson2002} are used to model distributions via \emph{conditional independence} (CI) relations. They are an extension of models based on \emph{directed acyclic graphs} (DAGs), as MAG models remove the assumption of no unobserved common parents and allow data arising from distributions with a more general independence structure. These graphs have been proven to be useful in various scenarios, for example, to infer causal effects from observational data. Therefore learning the best graph associated with a given dataset is a crucially important task.

There are three classical types of learning algorithm for graphical models: as well as constraint-based and score-based methods, there are hybrid methods that combine the first two approaches. The canonical constraint-based method for learning DAGs/MAGs is the PC/FCI algorithm \citep{spirtes2000causation}. Briefly speaking, this type of method tests for conditional independences in the empirical distribution, and uses the results to reconstruct the graph. The problem of constraint-based methods is that when the group of variables is large, it is likely that a mistake in testing a conditional independence will be made; this error will be propagated through the algorithm, and the resulting graph will not reflect the true independence structure generating the data \citep[see, e.g.][]{ramsey06adjacency, evans20model}.

Score-based methods tend to be more accurate. The main idea is that they go through many graphs and select the graph with the highest score, which asymptotically gives the correct graph. There now exist several score-based methods for MAGs  \citep{triantafillou2016score, rantanen2021maximal,chen2021integer,claassen2022greedy}. There are two main problems with these approaches.
First, the score used for DAGs is generally the Bayesian information criterion  \citep[BIC,][]{chickering2002optimal,jaakkola2010learning}, which is also used by the above methods for MAGs. This scoring criteria is known to be \emph{consistent}, that is, in the limit of infinite samples, models with lower complexity that contain the true model always have better BIC compared to higher complexity models. However, although methods for fitting Gaussian or discrete MAG models via maximum likelihood have been given by \citet{drton2009computing} and \citet{evans2010maximum,Evans2014} (so the corresponding BIC score can be obtained), MLEs are not available in closed-form, and therefore they need to be computed iteratively using a numerical method.  The functions being maximized are also not generally convex if the model is not a DAG, and so these iterative algorithms may converge to a point which is not globally optimal. Moreover, the factorisation of distributions in MAG models is complicated \citep{richardson2009factorization}, and the scores are only decomposable with respected to the components connected by bidirected paths, also known as districts or c-components.  In contrast, for DAGs the MLE is available in closed-form and the BIC score can be decomposed in terms of individual variables and their parent sets.

Another problem is how to reduce the number of graphs visited. Logically we could score every graph, but this is obviously hopelessly inefficient. Graphs that represent the same CI relations are said to be in the same \emph{Markov equivalence class} (MEC). Graphs in the same MEC generally have the same score, and so it is a waste of time to score other graphs in the same class. Among the previous works mentioned, only \citet{claassen2022greedy} explore graphs in the space of MECs. However their method uses the BIC score, which is computationally expensive. We will see that the framework of imsets \citep{studeny2006probabilistic} combined with another representation of MECs \citep{hu2020faster}, provides a solution to address the above two issues at the same time.

\subsection{Imsets}

Imsets are `integer valued vectors indexed by subsets of the variables' that can be used to encode arbitrary CI models. Imsets are very useful in the context of graphical models, because DAG models fit into this representation very elegantly. The imset used to represent a DAG model is usually the `standard imset', meaning that it is the simplest imset that correctly and precisely represents the conditional independences implied by the corresponding graph. For DAGs, there are several advantages to using standard imsets. Firstly, when two DAGs represent the same CI model, in other words they are Markov equivalent, their standard imsets also agree. Moreover, the BIC score of a DAG (defined in Appendix \ref{sec:scoring_main}) is an affine function of the standard imset. So we can see that for DAGs, imsets not only provide a representation of the MEC but also lead to a consistent scoring criterion. 

In this paper, we propose a standard imset formula for MAGs which, though it does not always define the same model as the graph, and sometimes does not define a model at all, works for large sub-classes of MAGs. 
%
We use the \emph{parametrizing set} that arises in the discrete parametrization and a factorization theorem of MAG models \citep{richardson2009factorization, hu2020faster}. One of the motivations is that for DAGs, the \emph{characteristic} imset $c_{\G}$ introduced by \citet{studeny2010characteristic}, agrees with the \emph{parametrizing set}; since there is a bijection between the characteristic and standard imsets, we use the parameterizing set to define the characteristic imset for MAGs, and then work backwards to deduce an expression for the standard imset. Another motivation is that \citet{hu2020faster} show that two MAGs are Markov equivalent if and only if they have the same \emph{parametrizing set}. Thus by construction two MAGs agree on their standard imsets if and only if they are Markov equivalent. 

For general MAGs, we provide an imset from a new Markov property. The imset always defines the same model as the graph it is derived from, and can be constructed in polynomial time under mild assumptions on the graph structure. 

The paper is organised as follows. The definitions for basic terminology in graphical models and imsets are given in Section \ref{sec:dfn}. In Section \ref{sec:main_res}, we present our main results on:
\begin{itemize}
    \item the formula of the `standard' imset $u_{\G}$ and its properties;
    \item when the MAGs are \emph{simple} (no head of size three or more), the imset does define the right model;
    \item  a proof that $\I_{u_\G} \subseteq \I_\G$ for all MAGs $\G$ (provided that $\I_{u_\G}$ is well-defined), where $\I_{u_\G} $ and $ \I_\G$ are lists of conditional independences implied by the imset $u_{\G}$ and graph $\G$ respectively.
\end{itemize}
If $\I_{u_{\G}} \subsetneq \I_{\G}$, this means that the imset we define does not include all the conditional independences implied by the graph. Full definitions will be provided in Section \ref{sec:dfn}. 

In Section \ref{power DAGs}, we introduce the concept of \emph{power DAGs}, which are DAGs on certain subsets of vertices. This is inspired by the decomposition of $u_\G$ given in Section \ref{sec:decomp of u_G} and it helps us to give a list of independences, which is equivalent to, but shorter than, the list of independences given by the ordered local Markov property. This results in an imset for general MAGs.

In Section \ref{sec:exp}, we report our experimental results on which graphs are perfectly Markovian ($\I_{\G} = \I_{u_{\G}}$) when there are at most seven nodes. Then we discuss related and future work in Section \ref{sec:Discussion}.

We also study our imsets of bidirected graphs and show that for a large class of bidirected graphs, the imset defines the right model but due to limit of space, we present our results in Appendix \ref{sec:bidirected graphs}. The condition we give is proved to be sufficient and empirically we checked that it is also necessary for graphs with at most seven vertices. The condition, however, is complicated, thus we conjecture that it is combinatorically difficult to obtain a minimal list of independences that define the model. We also give a list of forbidden induced subgraphs; that is, 
a motif that graphs cannot contain if $\I_{u_{\G}}$ is to be valid and be perfectly Markovian with respect to $\G$.  

Lastly, the MAGs considered in this paper are all \emph{directed} and contain no undirected edges. Such MAGs are sometimes referred as \emph{directed MAGs} by other authors. To obtain imsets for MAGs an with undirected component, one can first obtain an imset for the directed part as though the undirected component were complete, then obtain another imset for the undirected component \citep{kashimura2015standard}, and add two imsets together. If the directed part alone has an imset that defines the model for the subgraph, then this combination will also define the model for the whole graph.

\subsection{Related work}\label{sec: related work}

Similar work has been done by \citet{andrews:phd, andrews22}, who call the parametrizing set of MAGs the \emph{m-connecting} set. In particular, we have the same initial motivation as them, which is the similarity between the 0-1 characteristic imset \citep{studeny2010characteristic} and parametrizing sets \citep{hu2020faster} of DAGs. 
Our work contains results distinct from those works; we give  a more detailed comparison in Section \ref{sec:Discussion}.

\section{Definitions} \label{sec:dfn}

A \emph{graph} $\G$ consists of a vertex set $\mathcal{V}$ and an edge set $\mathcal{E}$ of pairs of distinct vertices. We consider mixed graphs with two types of edge: \emph{directed} ($\to$) and \emph{bidirected} ($\leftrightarrow$). For an edge in $\mathcal{E}$ connecting vertices $a$ and $b$, we say these two vertices are the \emph{endpoints} of the edge and the two vertices are \emph{adjacent} (if there is no edge between $a$ and $b$, they are \emph{nonadjacent}). 

A \emph{path} of length $k$ is an alternating sequence of $k+1$ distinct vertices $v_{i}$ and edges connecting $v_i$ and $v_{i+1}$ for $i = 0,\dots, k-1$. Note that if there is more than one edge between a pair of vertices, it is ambiguous which edge is used in a path; however, we will only ever use paths in graphs which are \emph{simple} (i.e.\ there is at most one edge between any pair of vertices). A path is \emph{directed} if its edges are all directed and point from $v_i$ to $v_{i+1}$. A \emph{directed cycle} is a directed path of length $k \geq 1   $ plus the edge $v_k \rightarrow v_0$, and a graph $\G$ is \emph{acyclic} if it has no directed cycle. A \emph{graph} $\G$ is called an \emph{acyclic directed mixed graph} (ADMG) if it is \emph{acyclic} and contains only directed and bidirected edges. 

For a vertex $v$ in an ADMG $\G$, we define the following sets:
\begin{align*}
\pa_{\G}(v) &= \{w: w \rightarrow v \text{ in } \G\} & & \sib_{\G}(v) = \{w:w \leftrightarrow v \text{ in } \G\}\\
\an_{\G}(v) &= \{w: w \rightarrow \cdots \rightarrow v \text{ in } \G \text{ or } w=v\} && \de_{\G}(v) = \{w: v \rightarrow \cdots \rightarrow w \text{ in } \G \text{ or } w=v\}\\
\dis_{\G}(v) &= \{w: w \leftrightarrow \cdots \leftrightarrow v \text{ in } \G \text{ or } w=v\}.
\end{align*}
They are known as the \emph{parents}, \emph{siblings}, \emph{ancestors}, \emph{descendants} and \emph{district} of $v$, respectively. These operators are also defined disjunctively for a set of vertices $W \subseteq \mathcal{V}$. For example $\pa_{\G}(W) = \bigcup_{w \in W} \pa_{\G}(w)$. Vertices in the same district are connected by a bidirected path and this is an equivalence relation, so we can partition $\mathcal{V}$ and denote the \emph{districts} of a \emph{graph} $\G$ by $\mathcal{D}(\G)$. We sometimes ignore the subscript if the graph we refer to is clear, for example $\an (v)$ instead of $\an_{\G}(v)$.

A \emph{topological ordering} is an ordering on the vertices such that if $w \in \an_\G(v)$ then $w$ precedes $v$ in the ordering. There might be several topological orderings for any single graph. For an ADMG $\G$, given a subset $W \subseteq \mathcal{V}$, the \emph{induced subgraph} $\G_{W}$ is defined as the graph with vertex set $W$ and edges in $\G$ whose endpoints are both in $W$. Also for operators on an induced subgraph $\G_{W}$, we may denote them by the shorthand (e.g.) $\dis_{W}(v) := \dis_{\G_W}(v)$ if $\G$ is clear. 

For the set of vertices $\mathcal{V}$ in an ADMG $\G$, we associate them with random
variables $(X_v)_{v \in \mathcal{V}}$ which are defined in real finite dimensional  
probability spaces. Let $X_A = (X_v)_{v \in A}$; 
we will sometimes abuse notation by using $A$ to denote both a set of vertices and 
the corresponding collection of variables $X_A$. 

Graphs are associated with conditional independence relations via a separation criterion; in the case of mixed graphs we use \emph{m-separation}, which is defined in Appendix \ref{sec:msep}. Distributions are associated with graphs via different but equivalent Markov properties (lists of independences). We will use the (reduced) \emph{ordered local Markov property}. 

\begin{definition}\label{ancestral set}
If a set $A$ satisfies $\an_\G(A)=A$, then it is an \emph{ancestral set}.

If $A$ is an ancestral set in an ADMG $\G$, and $v$ is a vertex in $A$ such that $\ch_\G(v) \cap A=\emptyset$, then the \emph{Markov blanket} of $v$ with respect to $A$ is defined as:$$\mb_\G(v,A)=\pa_{A}(\dis_{{A}}(v)) \cup \dis_{{A}}(v) \setminus \{v\}.$$
\end{definition}

\begin{definition}\label{local Markov property}
A distribution $P(X_{\mathcal{V}})$ is said to satisfy the \emph{ordered local Markov property} with respect to an ADMG $\G$ if for any ancestral set $A$ and a childless vertex $v$ in $A$, $$X_v \indep X_{A \setminus (\mb(v,A) \cup \{v\})} \mid X_{\mb(v,A)} \quad [P].$$
\end{definition}

It is possible to reduce this ordered local Markov property to a smaller collection 
of ancestral sets, as first shown by \citet{richardlocalmarkov}.

\begin{definition}
For an ancestral set $A$ and a childless vertex $v \in A$, we say that $A$ is \emph{maximal} with respect to $\mb_\G(v,A)$ if, for every ancestral set $B$ such that $A \subseteq B$ and $\mb_\G(v,A) = \mb_\G(v,B)$, we have $A=B$. 

The \emph{reduced ordered local Markov property} is then the ordered local Markov property but restricted to only maximal ancestral sets.
\end{definition}
\subsection{MAGs}

\begin{definition}\label{MAGs}
An ADMG $\G$ is called a \emph{maximal ancestral graph} (MAG), if:
\begin{itemize}
    \item[(i)] for every pair of \emph{nonadjacent} vertices $a$ and $b$, there exists some set $C$ such that $a,b$ are m-separated given $C$ in $\G$ (\emph{maximality});
    \item[(ii)] for every $v \in \mathcal{V}$, $\sib_{\G}(v) \cap \an_{\G}(v) = \emptyset$ (\emph{ancestrality}).
\end{itemize}
\end{definition}

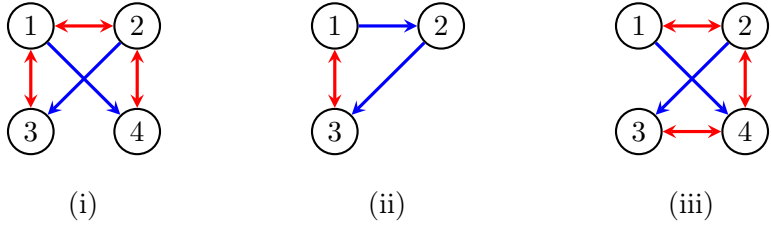
\begin{figure}
\centering
  \begin{tikzpicture}
  [rv/.style={circle, draw, thick, minimum size=6mm, inner sep=0.8mm}, node distance=14mm, >=stealth]
  \pgfsetarrows{latex-latex}
\begin{scope}
  \node[rv]  (1)            {$1$};
  \node[rv, right of=1] (2) {$2$};
  \node[rv, below of=1] (3) {$3$};
  \node[rv, right of=3] (4) {$4$};
  \draw[<->, very thick, red] (1) -- (3);
  \draw[<->, very thick, red] (2) -- (4);
  \draw[<->, very thick, red] (1) -- (2);
  \draw[<-, very thick, color=blue] (3) -- (2);
  \draw[<-,very thick, blue] (4) -- (1);
  \node[below right of=3,xshift=-0.3cm] {(i)};
  \end{scope}
\begin{scope}[xshift = 4cm]
   \node[rv]  (1)           {$1$};
  \node[rv, right of=1] (2) {$2$};
  \node[rv, below of=1] (3) {$3$};
  \draw[<->, very thick, red] (1) -- (3);
  \draw[->, very thick, blue] (1) -- (2);
  \draw[->, very thick, blue] (2) -- (3);
  \node[below right of=3,xshift=-0.3cm] {(ii)};
  \end{scope}
 \begin{scope}[xshift = 8cm]
   \node[rv]  (1)           {$1$};
  \node[rv, right of=1] (2) {$2$};
  \node[rv, below of=1] (3) {$3$};
  \node[rv, right of=3] (4) {$4$};
  \draw[<->, very thick, red] (2) -- (4);
  \draw[<->, very thick, red] (3) -- (4);
  \draw[<-, very thick, color=blue] (3) -- (2);
  \draw[<-,very thick, blue] (4) -- (1);
  \draw[<->, very thick,red] (1) -- (2);
  \node[below right of=3,xshift=-0.3cm] {(iii)};
  \end{scope}
\end{tikzpicture}
\caption{(i) An ancestral graph that is not maximal. (ii) A maximal graph that is not ancestral. (iii) A maximal ancestral graph. }
\label{fig:MAGs}
\end{figure}

For example, the graph in Figure \ref{fig:MAGs}(i) is not maximal because 1 and 2 are not adjacent, but no subset of $\{3,4\}$ will m-separate them; (ii) is not ancestral as 1 is a sibling of 3, which is also one of its descendants; and (iii) is a MAG, in which the only conditional independence is $X_1 \indep X_3 \mid X_2$.

\begin{definition}
Two graphs $\G_{1}$ and $\G_{2}$ with the same vertex sets, are said to be \emph{Markov equivalent} if any m-separation holds in $\G_{1}$ if and only if it holds in $\G_{2}$.
\end{definition}

For every ADMG $\G$, we can project it to a MAG $\G^m$ such that $\G$ is Markov equivalent to $\G^m$, and $\G^m$ preserves the ancestral relations in $\G$ \citep{richardson2002}. Moreover, \citet{hu2020faster} show that the heads and tails defined below (and so the parametrizing set) are preserved through the projection. Hence in this paper, we will only consider MAGs.

\subsection{Heads and tails}\label{sec: heads and tails}

A head is a subset of vertices with a corresponding tail. The concept of heads and tails originates from \citet{richardson2009factorization}, who provides a factorization theorem for ADMGs. 
The heads and tails are analogous to vertices and parents set in DAGs, but are generalized to set of vertices; in particular, the tail of a head $H$ is the Markov blanket of $H$ within the ancestral set $\an_\G(H)$. 
\begin{definition}\label{def: barren}
For a vertex set $W \subseteq \mathcal{V}$, we define the \emph{barren subset} of $W$ as:
$$
\barren_{\G}(W) = \{w \in W:\de_{\G}(w) \cap W = \{w\}\},
$$
which is the subset of $W$ such that each node has no non-trivial descendants in $W$.
\end{definition}

\begin{definition}\label{def: heads and tails}
A vertex set $H$ is called a \emph{head} if:
\begin{itemize}
    \item[(i)] $\barren_{\G}(H) = H$;
    \item[(ii)] $H$ is contained in a single district in $\G_{\an (H)}$.
\end{itemize}
For an ADMG $\G$, we denote the set of all heads in $\G$ by $\mathcal{H}(\G)$. The \emph{tail} of a $\head$ is defined as:
$$\tail_\G(H) = (\dis_{\an(H)}(H) \setminus H) \cup \pa_{\G}(\dis_{\an(H)}(H)).$$
\end{definition}

\begin{remark}
Note that for a head $H$ and each $h \in H$, the set $(H \setminus \{h\}) \cup \tail_\G(H)$ is the Markov blanket of $h$ within $\an_\G(H)$.
\end{remark}

\begin{definition}\label{def: parametrizing sets}
The \emph{parametrizing set} of $\G$, denoted by $\mathcal{S}(\G)$ is defined as:
$$\mathcal{S}(\G) = \{H\cup A:H \in \mathcal{H}(\G)\text{ and } \emptyset \subseteq A \subseteq \tail_\G(H)\}.$$
\end{definition}
In Figure \ref{example: heads and tails}, we give an example for heads and tails, as well as the sets not contained in the parametrizing set.
\begin{figure}

\centering
\begin{tikzpicture}[baseline=-1.5\baselineskip, 
node distance = 7mm,
  start chain = going right,
  rv/.style={circle, draw, thick, minimum size=6mm, inner sep=0.8mm}, node distance=14mm, >=stealth
]

  \pgfsetarrows{latex-latex}
\begin{scope}
  \node[rv]  (1)            {$1$};
  \node[rv, below of=1, xshift=-6mm,yshift=2mm] (2) {$2$};
  \node[rv, below of=1, xshift=6mm,yshift=2mm] (3) {$3$};
  \node[rv, above of=1,yshift=-2mm] (4) {$4$};
  \node[rv, below of=2, xshift=-10mm,yshift=5mm](5){$5$};
  \node[rv, below of=3, xshift=10mm,yshift=5mm](6){$6$};
  \draw[->, very thick, blue] (1) to [bend right=30] (5);
  \draw[->, very thick, blue] (2)  to [bend right=30] (6);
  \draw[->, very thick, blue] (3)  to [bend right=30] (4);
  \draw[beg] (1) -- (3);
  \draw[beg] (2) -- (3);
  \draw[beg] (1) -- (2);
  \draw[beg] (2) -- (5);
  \draw[beg] (3) -- (6);
  \draw[beg] (1) -- (4);
  \end{scope}
\end{tikzpicture}
\hfil
  \small
\begin{tabular}{c|c|c} 
  \toprule
  heads & tails & sets not in $\mathcal{S}(\G) $\\ [0.5ex] 
  \midrule
$\{1\},\{2\},\{3\}$ & \multirow{3}{*}{$\emptyset$} & $\{2,4\},\{3,5\}$\\
 $\{1,2\},\{1,3\},\{2,3\}$ &  & $\{4,5\},\{1,6\}$ \\
 $\{1,2,3\}$   &  &$\{4,6\},\{5,6\}$ \\
 \cmidrule{1-2}
  $\{5\},\{2,5\},\{2,3,5\}$ & $\{1\}$ &$\{1,2,6\},\{1,3,5\}$\\

  $\{6\},\{3,6\},\{1,3,6\}$ & $\{2\}$&$\{1,4,5\},\{1,5,6\}$\\

  $\{4\},\{1,4\},\{1,2,4\}$ & $\{3\}$&$\{2,3,4\},\{2,4,6\}$\\

  $\{3,5,6\}$ & $\{1,2\}$ & $\{2,5,6\},\{3,4,5\}$\\ 

  $\{2,4,5\}$ & $\{1,3\}$ & $\{3,4,6\},\{1,3,4,5\}$\\ 

  $\{1,4,6\}$ & $\{2,3\}$ & $\{2,3,4,5\},\{1,2,5,6\}$\\ 

  $\{4,5,6\}$& $\{1,2,3\}$ & \\ [1ex] 
 \bottomrule
\end{tabular}
\caption{A MAG $\G$ and a table for its heads and tails.}
\label{example: heads and tails}
\end{figure}

For a conditional independence $I = \langle A \indep B \cmid C \rangle$, let 
$$
\bar{\Sset}(I) = \{A' \cup B' \cup C': \emptyset \subset A' \subseteq A , \; \emptyset \subset B' \subseteq B, \; \emptyset \subseteq C' \subseteq C \}
$$ 
be the (\emph{constrained}) sets associated with $I$. These can be thought of as the set of interactions that are constrained to be zero by $I$, within the margin over $A \cup B \cup C$. Proposition \ref{prop:parametrizing set and independence} shows that a set $S$ is \emph{not} in $\Sset(\G)$ if and only if $S \in \bar{\Sset}(I)$ for some conditional independence $I$ entailed by the graph. 

\citet{hu2020faster} has a detailed introduction to the concept of heads and tails, and we recommend it for background reading. They prove that Markov equivalent MAGs have the same parametrizing set.  We will explain this in more detail as we will use this result.

\subsection{Imsets}

We now introduce the framework of imsets and conditional independence of \citet{studeny2006probabilistic}. 
Let $\mathcal{P}(\mathcal{V})$ be the power set of a finite set of variables $\mathcal{V}$. For any three disjoint sets, $A,B,C \subseteq \mathcal{V}$,  we write the triple as $\langle A, B \cmid C \rangle$ and denote the set of all such triples by $\mathcal{T}(\mathcal{V})$.
The set of conditional independence statements under $P$ is denoted as: 
$$
\I_P = \{\langle A,B \cmid C \rangle \in \mathcal{T}(\mathcal{V}) : A \indep B \mid C \, [P] \},
$$ 
and\footnote{Here for simplicity we use $A \indep B \mid C \, [P]$ to represent $X_A \indep X_B \mid  X_C$ under $P$.} this is called the \emph{conditional independence model} induced by $P$.

\begin{definition} \label{semi-elementary imset}
An \emph{imset} is an integer-valued function $u$: $\mathcal{P}(\mathcal{V}) \rightarrow \mathbb{Z}$. The \emph{identifier} function $\delta_A$ of a set $A \subseteq \mathcal{V}$ is an imset, defined as $\delta_{A}(B) = 1$ if $B=A$ and otherwise $\delta_{A}(B) = 0$.

A \emph{semi-elementary imset} $u_{\langle A,B | C\rangle}$ associated with any triple $\langle A,B \cmid C\rangle \in \mathcal{T}(\mathcal{V})$ is defined as: $u_{\langle A,B | C\rangle} = \delta_{A \cup B \cup C}-\delta_{A \cup C}-\delta_{B \cup C}+\delta_{C}$. An \emph{elementary} imset corresponds to the case when both $A$ and $B$ are singletons.

An imset $u$ is \emph{combinatorial} if it can be written as a non-negative integer combination of elementary imsets.
We call an imset $u$ \emph{structural} if there exists $k \in \mathbb{N}$ such that $k \cdot u$ is combinatorial.

We also define the \emph{degree} of a structural imset as the number of elementary 
imsets in the sum for the minimum $n$ that makes $n\cdot u$ combinatorial.
\end{definition}

Note that in the case of $\abs{\mathcal{V}} \leq 4$, every structural imset is also combinatorial, but for $\abs{\mathcal{V}} \geq 5$ this is not true \citep{hemmecke2008three}. We will also give an example in Section \ref{sec: non Markovian example}.

A triple $\langle A, B \cmid C \rangle$ is represented in a structural imset $u$ over $\mathcal{V}$, written as $A \indep B \mid C \, [u]$ if there exists $k \in \mathbb{N}$ such that $k \cdot u - u_{\langle A, B \mid C \rangle}$ is a combinatorial imset over $\mathcal{V}$. The \emph{model induced by} $u$ then is defined as: $$\I_u = \{\langle A, B \mid C\rangle \in \mathcal{T}(\mathcal{V}): A \indep B \mid C \, [u]\}.$$
We say that an imset $u$ is \emph{Markovian} with respect to an independence 
model $\I$ if for every $\langle A,B \mid C \rangle \in \I$,
we have that $\langle A,B \mid C \rangle \in \I_u$.  If the converse holds, we say it is \emph{faithful}, and if it is both Markovian and faithful we say it is \emph{perfectly Markovian} with respect to $\I$.

For a MAG $\G$, we write $A \perp_m B \mid C \, [\G]$ if $A$ and $B$ are m-separated by $C$ in $\G$. The \emph{model induced by} $\G$ is then defined as: $$\I_{\G} = \{\langle A, B \mid C\rangle \in \mathcal{T}(\mathcal{V}): A \perp_m B \mid C \, [\G] \}.$$

\begin{example}
Consider the DAG in Figure \ref{imsetexample}; a combinatorial imset $u$ such that $\I_u = \I_{\G}$ is: $u = u_{\langle 1,2 \rangle}+u_{\langle 4,12 \mid 3 \rangle}$, which has entries:
\begin{center}
\begingroup
\renewcommand{\arraystretch}{1.5} 
\begin{tabular}{c|cccc cccc}
$C$ & $\emptyset$  & $\{1\}$ & $\{2\}$ & $\{1,2\}$ & $\{3\}$ & $\{3,4\}$ & $\{1,2,3\}$ & $\{1,2,3,4\}$\\
\hline
$u(C)$ & $+1$ & $-1$ & $-1$ & $+1$ & $+1$ & $-1$ & $-1$ & $+1$ \\ 
\end{tabular}.
\endgroup
\end{center}

\end{example}

\begin{remark}
Independence models (whether represented by imsets or not) always obey the semi-graphoid axioms, which are listed in Appendix \ref{sec:graphoids}. In the previous example, the 
conditional independences $4 \indep 2 \mid 3$ and $4 \indep 1 \mid 2,3$ can both be 
deduced from $4 \indep 1,2 \mid 3$, and indeed $\I_u$ represents these constraints
as well.

Appendix \ref{sec:graphoids} also lists some other rules, but these only apply to probabilistic independence models under some additional assumptions. Note that we 
restrict to using the semi-graphoids, these additional rules are for discussion 
with related work in Section \ref{sec: related work}.
\end{remark}


\section{Main results} \label{sec:main_res}

In this section we will define the standard imset of MAGs. We begin by reviewing existing results on imsets for DAGs, and then draw a link between them and the parametrizing set of MAGs. This provides the motivation for our definition for the standard imsets of MAGs.

\subsection{Previous work}

For a DAG $\G$ over $\mathcal{V}$, its standard imset is defined as:
\begin{align}
u_{\G}=\delta_{\mathcal{V}}-\delta_{\emptyset}-\sum_{i \in \mathcal{V}} \left\{ \delta_{\{i\} \cup \pa_{\G}(i)}-\delta_{\pa_{\G}(i)}\right\}. \label{DAG imset}    
\end{align}

\citet{studeny2006probabilistic} shows that for a DAG $\G$, $\I_{\G} = \I_{u_{\G}}$.  That is, $u_\G$ is perfectly Markov with respect to $\I_\G$.  We will often just say that $u_\G$ is perfectly Markovian with respect to $\G$, rather than explicitly invoking the list of independences.

\begin{figure}
    \centering
     \begin{tikzpicture}
  [rv/.style={circle, draw, thick, minimum size=6mm, inner sep=0.8mm}, node distance=14mm, >=stealth]
  \pgfsetarrows{latex-latex}
\begin{scope}
  \node[rv,yshift=-1.25cm]  (1)            {$1$};
  \node[rv, above of=1] (2) {$2$};
  \node[rv, right of=1, yshift=0.70cm] (3) {$3$};
  \node[rv, right of=3] (4) {$4$};

  \draw[deg] (1) -- (3);
  \draw[deg] (2) -- (3);
  \draw[deg] (3) -- (4);
  \end{scope}
\begin{scope}[xshift=6cm]
  \node[rv]  (1)            {$1$};
  \node[rv, right of=1] (2) {$2$};
  \node[rv, below of=2] (3) {$3$};
  \node[rv, left of=3] (4) {$4$};
  \draw[<->, very thick, red] (1) -- (2);
  \draw[<->, very thick, red] (4) -- (3);
  \draw[->, very thick, blue] (1) -- (4);
  \draw[deg] (3) -- (2);
  \end{scope}

\end{tikzpicture}
    \caption{(i) A DAG with 4 nodes; (ii) a MAG $\G$ in which there is no topological ordering such that the tail of a head precedes any
vertex in the head.}
    \label{imsetexample}
\end{figure}
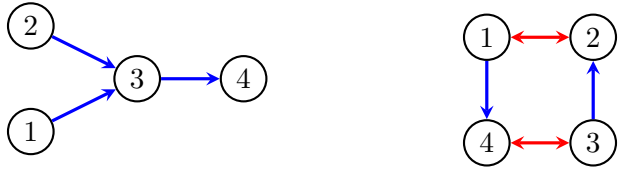

\begin{example}
Consider the DAG $\G$ in Figure \ref{imsetexample}(i), by definition, its standard imset is:
\begin{align*}
    u_{\G} &= \delta_{1234}-\delta_{\emptyset}-(\delta_{34}-\delta_{3})
    -(\delta_{123}-\delta_{12})-(\delta_{2}-\delta_{\emptyset})-(\delta_{1}-\delta_{\emptyset})\\
    &= \left(\delta_{1234}-\delta_{123} - \delta_{34} + \delta_3\right) +
    \left(\delta_{12}-\delta_1-\delta_2 + \delta_\emptyset\right)\\
    &= u_{\langle 4,12 \cmid 3\rangle}+u_{\langle 1,2 \rangle}.
\end{align*}
In the last line, the conditional independences of the semi-elementary imsets are the independences required for the ordered local Markov property of $\G$, that is, given the numerical (and also topological) ordering $$i \indep [i-1] \setminus \pa_\G(i) \mid \pa_\G(i),$$ where $[i] = \{1,2,\ldots,i\}$.
\end{example}

For a DAG $\G$, \citet{studeny2010characteristic} introduce the \emph{characteristic imset} $c_{\G}$, which can be obtained from the standard imset $u_{\G}$ by a one-to-one linear transformation called the M\"obius transform \citep{lauritzen1996graphical}:
\begin{align*}
c_{\G}(S) &= 1- \sum_{T: S \subseteq T \subseteq \mathcal{V}} u_{\G}(T) & \text{and} &&
u_{\G}(T) &= \sum_{S: T \subseteq S \subseteq \mathcal{V}} (-1)^{|V \setminus S|} \left\{1 - c_{\G}(S)\right\}. &
\end{align*}
These transforms are just inclusion and exclusion formulae.
This gives another representation of the equivalence classes of DAGs, and \citep{studeny2010characteristic},
%
provide the following theorem.

\begin{theorem}\label{thm: charac imset for MAGs}
For a DAG $\G$, we have:
\begin{itemize}
    \item[(i)] $c_{\G}(S) \in \{0,1\}$ for each $S \subseteq \mathcal{V}$, and
    \item[(ii)] $c_{\G}(S)=1$ if and only if either $S = \emptyset$ or there exists some $i \in S$ with $S \setminus \{i\} \subseteq \pa_{\G}(i)$.
\end{itemize}
\end{theorem}
It can be shown that for a DAG $\G$, a nonempty set $S$ is in the \emph{parametrizing set} $\mathcal{S}(\G)$ if and only if $c_{\G}(S) = 1$. We also have the following result about parameterizing sets \citep{hu2020faster}.

\begin{theorem}\label{thm: ME of MAGs}
Let $\G_{1}$ and $\G_{2}$ be two MAGs. Then $\G_{1}$ and $\G_{2}$ are Markov equivalent if and only if $\mathcal{S}(\G_{1})=\mathcal{S}(\G_{2})$.
\end{theorem}

Inspired by the relation in DAGs between standard imsets and characteristic imsets, and the consistency between characteristic imsets and parametrizing sets, we extend the definition of the characteristic imset to MAGs using the parametrizing set. The previous M\"obius transform used on $u_{\G}$ to obtain $c_{\G}$ is linear and thus has an inverse form (we will see later), we then apply this inverse formulae on $1 - c_\G$, to obtain the `standard' imsets, $u_\G$.  
We will show that for many MAGs, these `standard' imsets are combinatorial and perfectly Markovian with respect to the original graph $\G$.

\begin{example}\label{standard imset motivation}
Consider the graph $\G$ in Figure \ref{imsetexample}(ii), which is an example from \citet{richardson2009factorization}. 
There are two absent edges, each of them corresponding to a conditional independence; specifically $1 \indep 3$ and $2 \indep 4 \cmid 1,3$.  Clearly we want a standard imset $u_{\G}$ such that it is the sum of the elementary imsets for these two independences. So we want 
\begin{align*}
 u_{\G} &=u_{\langle 1,3 \rangle} + u_{\langle 2,4 | 13\rangle}\\
 &= \delta_{1234}-\delta_{134}-\delta_{123}+2\delta_{13}-\delta_{3}-\delta_{1}+\delta_{\emptyset}.
\end{align*}

The fourth term $\delta_{13}$ with coefficient 2 is quite interesting. Graphically these vertices are nonadjacent;
however, $\{1\}$ is the tail for the head $\{3,4\}$ that contains 3, and conversely $\{3\}$ is the tail for the head $\{1,2\}$ that contains 1. 
\end{example}

\subsection{Standard imsets of MAGs}

The fact that the parametrizing set and the characteristic imset of DAGs agree motivates our definition for the characteristic imset of MAGs, then we work backwards to deduce the form of the standard imset of MAGs.

\begin{definition}\label{def: characteristic imset}
For a MAG $\G$, define its characteristic imset $c_\G$ as $c_\G(S) = 1$ if $S \in \Sset(\G)$ and $c_\G(S) = 0$ otherwise. Moreover, we define $c_\G(\emptyset) = 1$ by convention and in order to preserve the characteristic imset for DAG models.
\end{definition}

\begin{theorem}\label{thm: standard imset}
For a MAG $\G$ with characteristic imset $c_\G$, let 
$$
u_{\G}(T)= \sum_{S: T \subseteq S \subseteq \mathcal{V}} (-1)^{\abs{S \setminus T}} (1-c_{\G}(S)).
$$
Then 
$$
u_{\G}=  \delta_{\mathcal{V}}-\delta_{\emptyset}-\sum_{H \in \mathcal{H}(\G)} \sum_{W \subseteq H} (-1)^{\abs{H \setminus W}} \delta_{W \cup \tail(H)}. 
$$
\end{theorem}

We will refer to $u_\G$ as the `standard' imset of the MAG $\G$, the quotes being in recognition of the fact that it does not always define the model (see Example \ref{example:5-cycle}).
 
\begin{proof}
We can obtain $c_\G$ via the following transformation from $u_\G$:
$$c_{\G}(S) = 1-\sum_{T: S \subseteq T \subseteq \mathcal{V}} u_{\G}(T).$$
Then we can prove the theorem by showing that, after substituting
$$
u_{\G} = \delta_{\mathcal{V}}-\delta_{\emptyset}-\sum_{H \in \mathcal{H}(\G)} \sum_{W \subseteq H} (-1)^{\abs{H \setminus W}} \delta_{W \cup \tail(H)}
$$ 
to the RHS of the previous equality, the result (say $c_{\G}^*$) is the same as 
$c_{\G}$.

Note that $c_{\G}^*(\emptyset) = 1$, as the sum of coefficients in $u_{\G}$ is 0. Suppose $S \neq \emptyset$ and let $\mathbbm{1}_P$ denote an indicator function, taking value 1 if $P$ is true and 0 otherwise. Then
\begin{align*}
c_{\G}^*(S) &= 1-\sum_{T: S \subseteq T \subseteq \mathcal{V}} \left\{\delta_{\mathcal{V}}(T)-\delta_{\emptyset}(T)-\sum_{H \in \mathcal{H}(\G)} \sum_{W \subseteq H} (-1)^{\abs{H \setminus W}} \delta_{W \cup \tail(H)}(T)  \right\}\\
&=\sum_{T: S \subseteq T \subseteq \mathcal{V}} \sum_{H \in \mathcal{H}(\G)} \sum_{W \subseteq H} (-1)^{\abs{H \setminus W}} \delta_{W \cup \tail(H)}(T) \\
&= \sum_{H \in \mathcal{H}(\G)} \sum_{T: S \subseteq T \subseteq \mathcal{V}} \sum_{W \subseteq H} (-1)^{\abs{H \setminus W}} \delta_{W \cup \tail(H)}(T) \\
&= \sum_{H \in \mathcal{H}(\G)} \sum_{T: S \subseteq T \subseteq \mathcal{V}} (-1)^{\abs{H \setminus T}} \mathbbm{1}_{\tail(H) \subseteq T \subseteq H \cup \tail(H)} \\
&= \sum_{H \in \mathcal{H}(\G)} \mathbbm{1}_{S \subseteq H \cup \tail(H)} \sum_{K \subseteq H \setminus S} (-1)^{\abs{K}+\abs{H \setminus S}} \\
&= \sum_{H \in \mathcal{H}(\G)} \mathbbm{1}_{H \subseteq S \subseteq H \cup \tail(H)}.
\end{align*}
The fifth equality can be seen in the following way: for each $S$ and $H$ we are counting the number of supersets $T \supseteq S$ such that $\tail_\G(H) \subseteq T \subseteq H \cup \tail_\G(H)$, multiplied by some constant $(-1)^{\abs{H \setminus T}}$. The indicator function comes from the fact that if $S$ is not a subset of $H \cup \tail_\G(H)$ then there is no such $T$. Now $S$ can be partitioned into two sets $S = S_{1} \dot\cup S_{2}$ which are subsets of $H$ and $\tail_\G(H)$, respectively. For any set $T \supseteq S$ such that $\tail_\G(H) \subseteq T \subseteq H \cup \tail_\G(H)$, It can be partitioned into three sets: $\tail_\G(H) \supseteq S_{2}$, $S_{1}$ and $K$. The last one $K$ is (any) subset of $H \setminus S$ and the first two sets are deterministic. Moreover, $(-1)^{\abs{H \setminus T}}$ is then equal to $(-1)^{\abs{(H \setminus S) \setminus K}}=(-1)^{\abs{K}+\abs{H \setminus S}}$.

Now the result follows from the proof of Lemma 4.3 in \citet{MLL} which shows that for each set $A$ there is at most one head $H$ such that $H \subseteq A \subseteq H \cup \tail_\G(H)$.
\end{proof}

If $\G$ is a DAG, then $u_\G$ in Theorem \ref{thm: standard imset} agrees with (\ref{DAG imset}).

\begin{remark}
Notice that the form of standard imsets in Theorem \ref{thm: standard imset} considers a tail with subsets of its head, in an opposite manner to the parametrizing set where we consider a head with subsets of its tail. One can check that this is how $2\delta_{13}$ is obtained in Example \ref{standard imset motivation}.
\end{remark}

\begin{corollary}
For two MAGs $\G$ and $\mathcal{H}$, they are Markov equivalent if and only if $u_{\G} = u_{\mathcal{H}}$.
\end{corollary}
\begin{proof}
This follows from Theorem \ref{thm: ME of MAGs} and the fact that the transformation between the standard imset and the characteristic imset is one-to-one.
\end{proof}

We also have the following result for induced
subgraphs, proven in Appendix \ref{sec:main_res_res}.

\begin{proposition} \label{prop:induced_subgraph}
Let $\G$ be a MAG, and suppose that for some ancestral subset $A \subset V$ we have that $u_{\G_A}$ is not Markovian with respect to $\G_A$.  Then the model $u_\G$ is not perfectly Markovian with respect to $\G$.
\end{proposition}

Before we proceed, we would like to clarify some of the terms and notations used in this paper. Originally, the standard imset of a DAG refers to the fact that it is the simplest imset that is perfectly Markovian w.r.t.~the graph \citep{studeny2006probabilistic}. 
However for MAGs, the `standard' imset we defined is not necessarily perfectly Markovian, and we use the quotes to refer to the imset obtained by applying the M\"obius inversion formula to the 0-1 imset defined by the parametrizing set. 

If we use $c$ with some subscript to denote an imset then it is in the characteristic form, i.e.~obtained by applying the M\"obius transform to some imset $u$. 
If $u$ is structural then it induces a model, and we associate the same model with both $c$ and $u$.

\subsection{Choice of the characteristic imset}\label{sec:choice of the characteristic imset}

Obviously, there are many imset representations for an independence model. The 
reason we choose to firstly work with a 0-1 characteristic imset is because if 
its corresponding `standard' imset \emph{is} perfectly Markovian w.r.t.\ the graph, 
then it is also the minimal representation. We proceed to 
prove this fact by studying the characteristic imset form $c$ of any structural imset $u$ that represent the model.

Example \ref{exm:6-cycle} gives a graph for which the `standard' imset is
not structural.  Fortunately, such graphs seem to be comparatively unusual.

We begin with a simple observation on the linear relationship between structural imset and its characteristic imset.

\begin{lemma} \label{lemma:linear relation between u and c}
For a structural imset $u$ and its characteristic imset $c$, we have that
$c(S) = 1$ if and only if $u$ contains no independence with $S$ as a constrained set, and otherwise $c(S) \leq 0$.
%
\end{lemma}

The following result from \citet{hu2020faster} will be useful.

\begin{proposition}\label{prop:parametrizing set and independence}
Let $\G$ be a MAG with vertex set $\mathcal{V}$. For
a set $S \subseteq \mathcal{V}$, we have $S \notin \Sset(\G)$ if and only if there are two
vertices $a, b$ in $S$ such that we can m-separate them by
a set $C$ such that $a, b \notin C$ with $S \subseteq C \cup \{a, b\}$.
\end{proposition}


\begin{corollary}\label{cor:imset must agree on S(G)}
Consider a MAG $\G$. For any structural imset $u$ such that $\I_{u}=\I_{\G}$, its characteristic imset $c$ must be an integer valued vector and satisfy the following:
\begin{itemize}
    \item[(i)] if $S \in \Sset(\G)$, then $c_\G(S) = 1$;
    \item[(ii)] if $S \notin \Sset(\G)$, then $c_\G(S) \leq 0$. 
\end{itemize}
\end{corollary}

We now argue for the choice of the 0-1 characteristic imset, which Corollary 
\ref{cor:imset must agree on S(G)} shows is unique for any given graph.  There is no way to totally order combinatorial imsets, though the degree does provide a partial ordering.  It is clear that, if our imset is combinatorial and perfectly Markovian, then it is also an \emph{imset of the smallest  degree} \citep{studeny2006probabilistic}, because there is exactly one independence for each missing edge.  We will see in Example \ref{example:5-cycle} that some MAG models cannot be represented in this way.



We can define a partial order on imsets defining the same model based on the coefficients of their corresponding characteristic imsets.  That is, given two structural imsets $u$ and $u'$ with corresponding characteristic imsets $c$ and $c'$, we say that $u$ \emph{is smaller than} $u'$ if $c(S) \geq c'(S)$ for every $S \subseteq \mathcal{P}(V)$.  Lemma \ref{lemma:linear relation between u and c} makes clear that this is a useful definition.
By Corollary \ref{cor:imset must agree on S(G)}, if the 0-1 imset's standard imset induces the same model as the graph, it is both an imset of smallest degree and the unique minimal such imset according to this partial order.  
%

%

\subsection{Standard imsets of simple MAGs}\label{sec: std imset of simple MAGs}

Unlike the standard imset of a DAG, it is often very hard to 
tell how to decompose this `standard' imset as a sum of semi-elementary imsets, 
because the size and number of heads seems arbitrary. However, if we 
restrict the size of heads to two, we can show that for any such MAG $\G$, its standard imset is 
always both combinatorial and perfectly Markovian with respect to $\G$.

\begin{definition}\label{simple MAG}
A MAG $\G$ is said to be \emph{simple} if $\G$ contains no head with size 
more than two.
\end{definition}

In Appendix \ref{sec: density of simple MAGs}, we consider how dense simple 
MAGs are as a subset of all MAGs by: (i) listing how many Markov equivalence 
classes contain at least one simple MAG; and (ii) simulating MAGs and counting 
proportions of them being Markov equivalent to some simple MAGs. Note that 
simple MAGs can have large districts and not be Markov equivalent 
to a graph with smaller districts; see Figure \ref{fig:simple_MAG_large_districts}.

It is not hard to show that for a MAG, $\{i,j\}$ is a head if and only if $i \leftrightarrow j$. 
Moreover, consider any  $j \leftrightarrow i \leftrightarrow k$.  If the MAG is simple, then there must be ancestral relations between $j$ and $k$ so that $\{i,j,k\}$ does not form a head (this is necessary and sufficient as shown by Lemma 7.3 of \citealp{MLL}). Hence for each vertex $i$, there is a total ordering on heads that contain $i$, and we now show that their tails are nested within one another. 

\begin{lemma}
Suppose $\G$ is a simple MAG with a given topological ordering. For every vertex $i$ and all heads of size two, $\{i,j_s\}$, with $j_1 < \dots < j_k < i$, we have $\pa_\G(i) = \tail_\G(i) \subseteq \tail_\G(i,j_1) \subseteq  \dots \subseteq \tail_\G(i,j_k)$.
\end{lemma}

\begin{example}\label{exm:non-trivial simple MAG}
Consider the simple MAG $\G$ in Figure \ref{A simple MAG}. The sets $\{7,8\}$ and $\{6,8\}$ are the heads of size two associated with the vertex 8. Their tails are $\{1,2,3,4,5,6\}$ and $\{1,3,4,5\}$ respectively. Note that $\tail_\G(\{6,8\}) \subset \tail_\G(\{7,8\})$. Also, we have $\pa_\G(8) = \{1,5\}$, which is a subset of both the other tails.

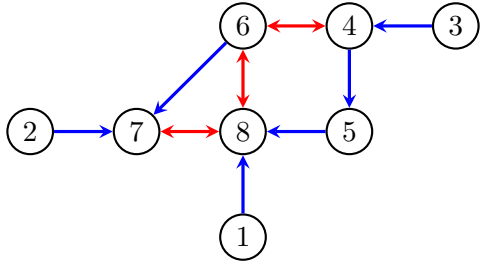
\begin{figure}
\centering
  \begin{tikzpicture}
  [rv/.style={circle, draw, thick, minimum size=6mm, inner sep=0.8mm}, node distance=14mm, >=stealth]
  \pgfsetarrows{latex-latex}
\begin{scope}
  \node[rv]  (1)            {$1$};
  \node[rv, above of=1] (8) {$8$};
  \node[rv, left of=8] (7) {$7$};
  \node[rv, left of=7] (2) {$2$};
  \node[rv, right of=8] (5) {$5$};
  \node[rv, above of=8] (6) {$6$};
  \node[rv, right of=6] (4) {$4$};
  \node[rv, right of=4] (3) {$3$};
  \draw[<->, very thick, red] (7) -- (8);
  \draw[<->, very thick, red] (6) -- (8);
  \draw[<->, very thick, red] (6) -- (4);
  \draw[->, very thick, blue] (1) -- (8);
  \draw[deg] (2) -- (7);
  \draw[deg] (5) -- (8);
  \draw[deg] (4) -- (5);
  \draw[deg] (3) -- (4);
  \draw[deg] (6) -- (7);
  \end{scope}

\end{tikzpicture}
\caption{A simple MAG used in Example \ref{exm:non-trivial simple MAG}}
\label{A simple MAG}
\end{figure}

One can check that its standard imset may be written as:
\begin{align*}
    u_{\G} &= u_{\langle 8,2 | 13456\rangle}+u_{\langle 8,34 | 15\rangle}+u_{\langle 7,1345 | 26\rangle}+u_{\langle 6,125 | 34\rangle}\\
    &\phantom{=}+u_{\langle 6,3\rangle}+u_{\langle 5,123 | 4\rangle}
    +u_{\langle 4,12 | 3\rangle}+u_{\langle 3,12\rangle}+u_{\langle 1,2\rangle}.
\end{align*}

Let us focus on the semi-elementary imsets that contain 8. The conditional independence between vertex 8 and vertex 2 can be seen in the following way: after marginalizing $\{7\}$ (the remaining seven vertices still form an ancestral set), $\{2\}$ would be outside of the Markov blanket of vertex 8 which is $\tail_\G(\{6,8\}) = \{1,3,4,5,6\}$, and this independence is implied by the ordered local Markov property. Similarly, if one further marginalizes $\{6\}$, then $\{3,4\}$ are not in the Markov blanket of vertex 8, which is just $\pa_\G(8) = \{1,5\}$, and this corresponds to $8 \indep 3,4 \mid 1,5$. 
\end{example}

Thus for simple MAGs with a given topological ordering, we can obtain conditional independences for each vertex $v$ by sequentially marginalizing the lesser member of each head of size two that contains $v$, and using the ordered local Markov property. See details in the proof of the following theorem.

\begin{theorem}\label{thm: standard imset with head size less 3}
For a MAG $\G$ and its standard imset $u_{\G}$, if $\G$ is simple then $\I_{u_{\G}} = \I_{\G}$. 
\end{theorem}
\begin{proof}
See Appendix \ref{Proof of simple mag Theorem}.
\end{proof}

\begin{corollary}\label{extendstandardimset}
For a MAG $\G$, if there exists a simple MAG $\mathcal{H}$ which is Markov equivalent to $\G$, then $\I_{u_{\G}} = \I_{\G}$.
\end{corollary}

\subsection{Examples where the imset is not perfectly Markovian}\label{sec: non Markovian example} 

In general---as the following examples show---the `standard' imset we have defined is not structural, and even if it is structural, $u_{\G}$ may not be perfectly Markovian with respect to the model induced by the graph.

\begin{figure}
    \centering
     \begin{tikzpicture}
  [rv/.style={circle, draw, thick, minimum size=6mm, inner sep=0.8mm}, node distance=14mm, >=stealth]
  \pgfsetarrows{latex-latex}
\begin{scope}
  \node (0) {};
    \node[rv]  (1) at (18:12.5mm)        {$1$};
  \node[rv]  (2) at (90:12.5mm)            {$2$};
\node[rv]  (3) at (162:12.5mm)            {$3$};
\node[rv]  (4) at (234:12.5mm)            {$4$};
\node[rv]  (5) at (306:12.5mm)            {$5$};

  \draw[beg] (1) -- (2);
  \draw[beg] (2) -- (3);
  \draw[beg] (3) -- (4);
  \draw[beg] (4) -- (5);
  \draw[beg] (1) -- (5);
  \node[below of=0, yshift=-5mm] {(i)};
  \end{scope}
\begin{scope}[xshift=4cm, yshift=0.5mm]
  \node (0) {};
  \node[rv]  (1) at (0:12.5mm)            {$1$};
  \node[rv]  (2) at (60:12.5mm)            {$2$};
\node[rv]  (3) at (120:12.5mm)            {$3$};
\node[rv]  (4) at (180:12.5mm)            {$4$};
\node[rv]  (5) at (240:12.5mm)            {$5$};
\node[rv]  (6) at (300:12.5mm)            {$6$};

  \draw[beg] (1) -- (2);
  \draw[beg] (2) -- (3);
  \draw[beg] (3) -- (4);
  \draw[beg] (4) -- (5);
  \draw[beg] (6) -- (5);
  \draw[beg] (6) -- (1);
  \node[below of=0, yshift=-5.5mm] {(ii)};
  \end{scope}

 \begin{scope}[xshift=7cm, yshift=6.6mm]
  \node[rv]  (1)            {$6$};
  \node[rv, right of=1] (2) {$1$};
  \node[rv, right of=2] (3) {$2$};
  \node[rv, below of=3](4){$3$};
  \node[rv, below of=2] (5) {$4$};
  \node[rv, below of=1] (6) {$5$}; 

  \draw[beg] (1) -- (2);
  \draw[beg] (2) -- (3);
  \draw[beg] (3) -- (4);
  \draw[beg] (2) -- (5);
  \draw[beg] (4) -- (5);
  \draw[beg] (6) -- (5);
  \draw[beg] (6) -- (1);
  \node[below of=5, yshift=2.5mm] {(iii)};
  \end{scope}

\end{tikzpicture}
    \caption{(i) a bidirected 5-cycle; (ii) a bidirected 6-cycle; (iii) a bidirected 6-cycle with an additional edge.}
    \label{fig:counterexample graph}
\end{figure}
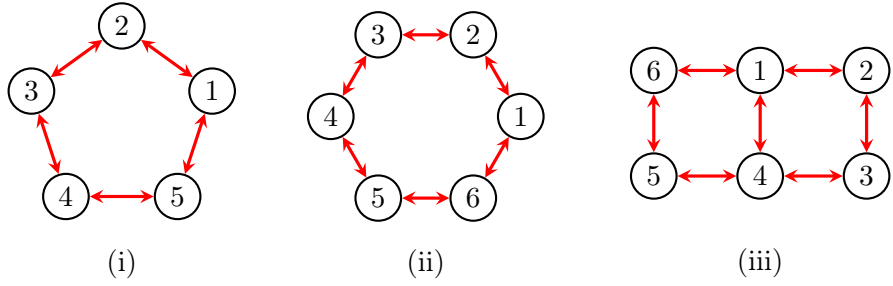

\begin{example}\label{example:5-cycle}

Consider the 5-cycle $\G$ with bidirected edges in Figure \ref{fig:counterexample graph}. Its `standard' imset is given by 
\begin{align*}
    u_{\G} &= u_{\langle 1,3 \mid 4\rangle}+u_{\langle 2,4 \mid 5 \rangle}+u_{\langle 3,5 \mid 1 \rangle}+u_{\langle 4,1 \mid 2 \rangle}+u_{\langle 5,2 \mid 3\rangle}\\ 
    &= u_{\langle 1,3 \mid 5\rangle}+u_{\langle 2,4 \mid 1\rangle}+u_{\langle 3,5 \mid 2 \rangle}+u_{\langle 4,1 \mid 3\rangle}+u_{\langle 5,2 \mid 4\rangle}.
\end{align*}
Any (strictly) conditional independence that holds in the graph is contained in $\I_{u_{\G}}$; however, marginal independences such as $1 \indep 3$ are not in the imset.

Note, however, that if we assume the underlying distribution is strictly positive  then we are able to deduce that any marginal independence in the graph is in $\I_{u_{\G}}$; by noting that (for example) $1 \indep 3 \mid 4$ and $1 \indep 4 \mid 3$ are both in $\I_{u_{\G}}$, and can be combined by using the intersection graphoid provided that the density is positive (See Appendix \ref{sec:more definition}), to obtain $1 \indep 3,4$.
It turns out that for MAGs on at most 5 vertices, this graph is the only one such that $u_{\G}$ is not perfectly Markovian with respect to $\G$.

As no marginal independences are represented in $\I_{u_\G}$ for the bidirected 5-cycle, this shows that for some
MAGs, if we want a structural imset $u_{\G}$ that is perfectly Markovian with respect 
to the graph, it is inevitable that there is some overlap between the sets associated 
with the independence decomposition of $u_{\G}$.  This also shows that the minimum 
degree of a model defining the 5-cycle is 6, even though there are only 5 pairs of 
vertices that are not adjacent; this in turn demonstrates the existence of MAG 
models that can only be defined by having a list of independences that involves 
a repeated pair of vertices.

Note also that this model does not have a unique minimal representation using the partial order from Section \ref{sec:choice of the characteristic imset}.  Indeed, there are at least five different such representations, obtained by adding $u_{\langle i,i+2 \rangle}$ to our `standard' imset for some $i \in \{1,\ldots,5\}$ (and using arithmetic modulo 5).
\end{example}

\begin{example} \label{exm:6-cycle}
Consider the bidirected 6-cycle $\G$ shown in Figure \ref{fig:counterexample graph}(ii), and $u_{\G}$, which is too long to list here.  Write $u_{\G}=\sum_{A \in \mathcal{P}(\mathcal{V})} \alpha_{A} \delta_{A}$, where $\alpha_A$ are coefficients for each identifier imset.  The largest disconnected sets are of size 4, and $\sum_{\abs{A}=4} \alpha_{A}=9$ and $\sum_{\abs{A}=3} \alpha_{A}=-22$; It is easy to show that when we add semi-elementary imset to an imset, if the semi-elementary imset contributes to the term with largest set size, we will subtract at most two to the coefficient of the terms with set size exactly one smaller than the largest one. But 22 is larger than $18 = 2 \cdot 9$. Hence this imset cannot be neither combinatorial or structural. 
Indeed, we have used a linear program to show that it is neither of these things.
\end{example}

\begin{example}

Consider Figure \ref{fig:counterexample graph} (iii). This graph has a `standard' imset that is not perfectly Markovian with respect to $\G$, and it is also not combinatorial.
However, it \emph{is} structural, because:
\begin{align*}
    2u_\G &= \ei{1,3} + \ei{1,3|56} +\ei{1,5} + \ei{1,5|23} +\ei{2,4} +\ei{2,4|56}+\\ 
    &\quad+ \ei{2,5|13} +\ei{2,5|46} +\ei{2,6} +\ei{2,6|35}+\ei{3,5} +\ei{3,5|26} + \\
    &\quad+\ei{3,6|24} +\ei{3,6|15} +\ei{4,6}+\ei{4,6|23}.
\end{align*}
This constitutes another example of an imset that is structural but not combinatorial, and it arises in a much more natural way that the one given by \citet{hemmecke2008three}.
\end{example}

\subsection{Relating to scoring criteria} \label{sec:scoring_main}

Here we explain how to use imsets to provide a consistent scoring criteria. For DAGs, the maximum likelihood part of the BIC (defined in Appendix \ref{sec:scoring_app}) can be shown to be the empirical entropy of $X_{\mathcal{V}}$, minus the inner product between the standard imset and the empirical entropy vector (we demonstrate this in Appendix \ref{sec:scoring_app} for discrete variables). This inner product is explained in the following. 

For random variables $X_{\mathcal{V}}$ with density function $p$, define the entropy ${\sf H} (X_{\mathcal{V}})$ as the expectation of $-\log p(X_{\mathcal{V}})$, i.e.~$\mathbb{E}[-\log p(X_{\mathcal{V}})]$. For three random variables $X_A$, $X_B$ and $X_C$, the inner product between the semi-elementary imset $u_{\langle A, B \mid C \rangle}$ and the entropy vector, whose entries correspond to the entropy $\sf{H}$ of every subset of $\mathcal{V}$, is the mutual information between $X_A, X_B$ given $X_C$; that is, ${\sf H}(X_{ABC})-{\sf H}(X_{AC})-{\sf H}(X_{BC})+{\sf H}(X_C)$. This quantity is always non-negative, and is zero if and only if the independence $A \indep B \mid C$ holds under $p$. 

Hence for DAGs, BIC scores can be interpreted as the discrepancy for a list of independences from the ordered local Markov property plus penalty terms for model complexity. It follows that we can do something similar for simple MAGs, since if $u_\G$ is perfectly Markovian w.r.t.\ the graph, this inner product (suitably penalized) provides a valid score. In fact, we prove in Appendix \ref{sec:scoring_app} that if $\I_{u_\G} = \I_\G$ then, this score is consistent as it approximates the BIC. However, for simple MAGs this score can be obtained much faster as the imsets can be constructed in quadratic time in the number of vertices; in contrast, there is no guarantee on computation time for BIC. Note that consistency of the BIC score only requires that the data are generated from one of the models being scored (and do not coincidentally lie 
on any other models with at most the same number of parameters). For this score, we also require that the graph that generates the data has a perfectly Markovian standard imset.

Moreover, as we have shown the decomposition (with positive integer coefficients) of simple MAGs, if the CI relations in a distribution can be described by a simple MAG, we always have that this inner product is zero. Empirically 
we observed that this is true even for general MAGs, in spite of previous examples where our imsets are not perfectly Markovian w.r.t.\ MAGs or not even structural. This suggests that the standard imsets we defined \emph{can} be expressed as the sum of semi-elementary imsets corresponding to conditional independences advertised by the graph, but some of the coefficients are negative, this will indeed turn out to be the case (Theorem \ref{thm:imset as sum of elementary imset for general MAGs}). 

\subsection{Motivations to simplify Markov property}

Next we discuss different choices of imsets. Obviously there are many imsets that represent the same model induced by graphs, but they are different in terms of computational complexity and statistical performance. Both the pairwise Markov property \citep{sadeghi2014markov} and the (reduced) ordered local Markov property \citep{richardlocalmarkov} can be used to construct imsets by summing 
semi-elementary imsets corresponding to list of independences\footnote{For the 
pairwise Markov property to be equivalent to the global Markov property it 
requires that the independence model is a `compositional graphoid' (see Appendix \ref{sec:graphoids}).  One can show the inner product with entropy is zero if and 
only if the distribution obeys the pairwise Markov property.}. 

However, the pairwise Markov property in general have more conditioning variables if graphs have complicated ancestral relations, and thus require to estimate entropy of more variables, which is hard problem. Even if it has the same degree as the standard imsets for simple MAGs, in practice we found it having worse performance in model search algorithms. For imsets corresponding to the (reduced) ordered local Markov property, since it in general contains redundant conditional independences (even for simple MAGs, see the example on next section), it has higher degree compared to the standard imsets, hence are less useful. Therefore for simple MAGs, our standard imsets are faster to compute (polynomial time) and more reliable.

Now for general MAGs, we have seen that the `standard' imsets obtained from the 0-1 characteristic imsets in general are invalid. To construct valid imsets one can of course use the imsets from the pairwise/reduced ordered local Markov property. However, we aim imsets that have fewer degree and fewer conditioning variables, hence simplifying the Markov property is essential to find the standard (simplest) imsets for general MAGs. In the next section we use a graphical tool called the Power DAGs to achieve this, and we show that our \emph{refined Markov property} is strictly simpler than the ordered local Markov property. The result is optimal for simple MAGs as it gives the list of independences that decomposes the standard imsets, but not optimal in general. 

Our refined Markov property results in an imset that is perfectly Markovian to a given MAG and have fewer degree/conditioning variables compared to pairwise/ordered local Markov property. We also show that fixing the maximal head size, the computational time for this imset is polynomial. 

\section{Power DAGs}\label{power DAGs}

We aim to give a simpler representation of the conditional independence model induced by MAGs. To describe these relations, we use \emph{power DAGs} of these graphs, which are DAGs over the set of heads. 

Our power DAG is completely analogous to the \emph{intrinsic power DAG} in \citep{richardson23nested}, as there is a one-to-one correspondence between the collections of head and \emph{intrinsic sets} used in that paper; indeed, for most MAGs, the two graphs are isomorphic.  The approach we give later to simplifying the power DAG is not suitable for intrinsic power DAGs for the \emph{nested Markov} model studied in that paper, because the order in which vertices are \emph{fixed} (marginalized in our case) is important for deriving nested constraints.

\subsection{Motivations and examples}
\begin{example}
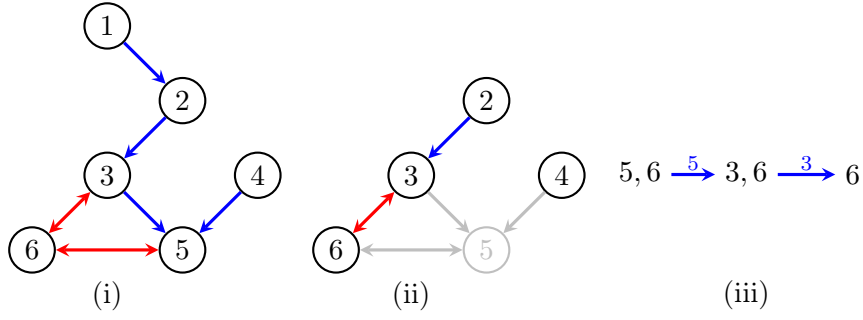
\begin{figure}
\centering
  \begin{tikzpicture}
  [rv/.style={circle, draw, thick, minimum size=6mm, inner sep=0.8mm}, node distance=14mm, >=stealth]
  \pgfsetarrows{latex-latex}
\begin{scope}
  \node[rv]  (6)            {$6$};
  \node[rv, above right of=6] (3) {$3$};
  \node[rv, below right of=3] (5) {$5$};
  \node[rv, above right of=5](4){$4$};
  \node[rv, above right of=3] (2) {$2$};
  \node[rv, above left of=2](1){$1$};
  \draw[->, very thick, blue] (1) to  (2);
  \draw[->, very thick, blue] (2) to  (3);
  \draw[->, very thick, blue] (3)  to (5);
  \draw[->, very thick, blue] (4)  to (5);
  \draw[beg] (5) -- (6);
  \draw[beg] (6) -- (3);
  \node[below of=3, xshift=0cm,yshift=-2mm] {(i)};
  \end{scope}
  \begin{scope}[xshift=4cm]
  \node[rv]  (6)            {$6$};
  \node[rv, above right of=6] (3) {$3$};
  \node[rv, lightgray, below right of=3] (5) {$5$};
  \node[rv, above right of=5](4){$4$};
  \node[rv, above right of=3] (2) {$2$};
  \draw[->, very thick, blue] (2) to  (3);
  \draw[->, very thick, lightgray] (3)  to (5);
  \draw[->, very thick, lightgray] (4)  to (5);
  \draw[<->,very thick, lightgray] (5) -- (6);
  \draw[beg] (6) -- (3);
  \node[below of=3, xshift=0cm,yshift=-2mm] {(ii)};
  \end{scope}
  \begin{scope}[xshift = 8cm, yshift=1cm]
  \node[]  (1)            {$5,6$};
  \node[, right of=1] (2) {$3,6$};
  \node[, right of=2] (3) {$6$};
  \path[->,draw,very thick,blue]
    (1) edge node [yshift=1.5mm]{\scriptsize{$5$}} (2)
    (2) edge node [yshift=1.5mm]{\scriptsize{$3$}} (3);

  \node[below of=2, yshift=-2mm] {(iii)};
  \end{scope}
\end{tikzpicture}
 \caption{(i) A simple MAG $\G$; (ii) the graph after removing $1$ and marginalizing $5$; (iii) a DAG on heads in $\G$ that contain the vertex 6.}
 \label{fig: local Markov is redundant}
\end{figure}
Consider the simple MAG in Figure \ref{fig: local Markov is redundant}(i), given the numerical ordering, the reduced ordered local Markov property for the vertex 6 would require: 
\begin{align*}
    & 6 \indep 1 \mid 2,3,4,5 
    & & 6 \indep 1,4 \mid 2,3
    & & 6 \indep 1,2,4.
\end{align*}
However, we only actually need:
\begin{align*}
    & 6 \indep 1 \mid 2,3,4,5
    &&6 \indep 4 \mid 2,3
    &&6 \indep 2.
\end{align*}
This list is the simplest as it comes from decomposition of standard imset $u_{\G}$ of the simple MAG $\G$.

We start with a topological ordering, in this case the numeric 
ordering of the vertices.  Then, for each vertex we consider the 
graph of its predecessors; for the vertex 6 this is just the 
graph itself.  
We see from this that 1 is not in the Markov blanket of 6, deduce that $6 \indep 1 \cmid 2,3,4$ holds, so we remove $1$ from the graph.  Then we must marginalize something other than 6 in the head $\{5,6\}$; for simple MAGs there is only ever one choice, 
so we obtain the graph in Figure \ref{fig: local Markov is redundant}(ii).  The maximal head in this graph is now $\{3,6\}$, 
and we now see that $\{4\}$ is no longer in the Markov blanket of 6; hence we obtain $6 \indep 4 \mid 2,3$ and remove it from further consideration. 
%
Finally we marginalize 3 and see that $6 \indep 2$. 

Alongside these operations, in Figure 
\ref{fig: local Markov is redundant}(iii) we construct the 
corresponding component of the power DAG for heads that 
contain 6.  We (potentially) associate a single independence
with the initial node (so $6 \indep 1 \mid 2,3,4,5$ in our case)
and another with each of the transitions in the graph ($6 \indep 4 \mid 2,3$ and $6 \indep 2$).

Specifically, in this example, we reach the head $\{3,6\}$ from the head $\{5,6\}$ by marginalizing $\{5\}$, as illustrated in 
. Now $\{4\}$ is no longer in the Markov blanket of 6, so we obtain $6 \indep 4 \mid 2,3$ for it and remove it from consideration. Next by marginalizing $\{3\}$, we reach the head $\{6\}$ and now $\{2\}$ is not in the Markov blanket of $\{6\}$ anymore, so we add the independence $6 \indep 2$. For better illustration, we can draw a DAG on heads if one head can reach another by marginalizing vertices in the head, see Figure \ref{fig: local Markov is redundant} (iii). Then the second and third independences in our list are represented by this DAG.
\end{example}

Now we give formal definitions of power DAGs and a more complicated example on non-simple MAGs will be given.
\subsection{Definition of power DAGs}

We assume throughout that the numerical ordering of vertices is also topological. 
To properly formulate the definition of  power DAGs, we need a few more notations and definitions.
For a vertex set $A$ and a vertex $i$, we write $A \leq i$ if $i$ is the \emph{maximal} vertex in $A$ w.r.t.\ a given topological ordering.  

\begin{definition}\label{one head to another head}
For a MAG $\G$ and two heads $H,H' \leq i$, we write $H \to^{K} H'$ ($\emptyset \subset K \subseteq H \setminus \{i\}$) if $\barren_{\G'}(\dis_{\G'}(i)) = H'$, where $\G' = \G_{\an(H)\setminus K}$. We will refer to $K$ as a marginalization set.
\end{definition}


Graphically, $H \to^{K} H'$ means that in the subgraph $\G_{\an(H)}$, the maximal head (i.e.~the barren subset of the district) that contains $i$ after marginalizing $K$ (subset of the barren subset) is $H'$.  Moreover, to save space we eschew set notation and union signs and write (e.g.) $k$ for $\{k\}$ and $HT$ for $H \cup T$. 
\begin{figure}
\centering
  \begin{tikzpicture}
  [rv/.style={circle, draw, thick, minimum size=6mm, inner sep=0.8mm}, node distance=14mm, >=stealth]
  \pgfsetarrows{latex-latex}
\begin{scope}[xshift=-4cm]
  \node[rv]  (1)            {$1$};
  \node[rv, below of=1, xshift=-6mm,yshift=2mm] (2) {$2$};
  \node[rv, below of=1, xshift=6mm,yshift=2mm] (3) {$3$};
  \node[rv, above of=1,yshift=-2mm] (4) {$4$};
  \node[rv, below of=2, xshift=-10mm,yshift=5mm](5){$5$};
  \node[rv, below of=3, xshift=10mm,yshift=5mm](6){$6$};
  \draw[->, very thick, blue] (1) to [bend right=30] (5);
  \draw[->, very thick, blue] (2)  to [bend right=30] (6);
  \draw[->, very thick, blue] (3)  to [bend right=30] (4);
  \draw[beg] (1) -- (3);
  \draw[beg] (2) -- (3);
  \draw[beg] (1) -- (2);
  \draw[beg] (2) -- (5);
  \draw[beg] (3) -- (6);
  \draw[beg] (1) -- (4);
  \node[below of=1, yshift=-2.1cm] {(i):$\G = \G_{\an(\{4,5,6\})}$};
  \end{scope}
\begin{scope}[xshift = 0.9cm]
  \node[rv]  (1)            {$1$};
  \node[rv, below of=1, xshift=-6mm,yshift=2mm] (2) {$2$};
  \node[rv, below of=1, xshift=6mm,yshift=2mm] (3) {$3$};
  \node[rv, below of=2, xshift=-10mm,yshift=5mm](5){$5$};
  \node[rv, below of=3, xshift=10mm,yshift=5mm](6){$6$};
  \draw[->, very thick, blue] (1) to [bend right=30] (5);
  \draw[->, very thick, blue] (2)  to [bend right=30] (6);
  \draw[beg] (1) -- (3);
  \draw[beg] (2) -- (3);
  \draw[beg] (1) -- (2);
  \draw[beg] (2) -- (5);
  \draw[beg] (3) -- (6);
  \node[below of=1, yshift=-2.1cm] {$\substack{\text{(ii)}:\text{ Marginalize 4} \\ \G' = \G_{\an(\{4,5,6\}) \setminus \{4\}} = \G_{\an(\{3,5,6\})}} $ };
  
  \end{scope}
\begin{scope}[xshift = 5.5cm]
  \node[rv]  (1)            {$1$};
  \node[rv, below of=1, xshift=-6mm,yshift=2mm] (2) {$2$};

  \node[rv, below of=2, xshift=-10mm,yshift=5mm](5){$5$};
  \node[rv, below of=2, xshift=2cm,yshift=5mm](6){$6$};
  \draw[->, very thick, blue] (1) to [bend right=30] (5);
  \draw[->, very thick, blue] (2)  to [bend right=30] (6);
  \draw[beg] (1) -- (2);
  \draw[beg] (2) -- (5);
  \node[below of=1, yshift=-2.1cm] {$\substack{\text{(iii)}:\text{ Marginalize 3} \\ \G'' = \G_{\an(\{3,5,6\}) \setminus \{3\}}}$};
  
  \end{scope}
\end{tikzpicture}
 \caption{An example for Definition \ref{one head to another head}.}
 \label{exp: marginalizing example}
\end{figure}
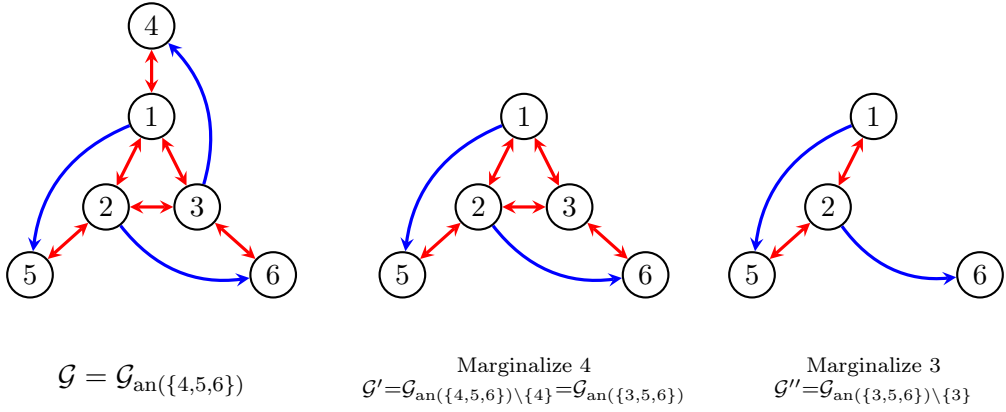

Take the graph in Figure \ref{exp: marginalizing example}(i) (see Figure \ref{example: heads and tails}), under the numerical topological ordering. Consider the final vertex $6$; other barren vertices in its Markov blanket are 4 and 5, which together with 6 form a head. By marginalizing $4$, we reach the head $\{3,5,6\}$, hence by definition, $\{4,5,6\} \to^{4} \{3,5,6\}$. Then if we marginalize $3$, we have $\{3,5,6\} \to^{3} \{6\}$. The above two marginalization steps are shown in Figure \ref{exp: marginalizing example} (ii) and (iii).

With each $H \to^{K} H'$ we associate a conditional independence. By Definition \ref{one head to another head}, if $H \to^{K} H' $ then $\an_\G(H)\setminus K = B$, where $B$ is an ancestral set, $T' = \tail_\G(H')$ and $\{i\} \cup \mb_\G(i,B) = H' \cup T'$. Hence by the ordered local Markov property and marginalizing vertices that lie outside of the Markov blanket of $i$ in $\G_{B}$, we have
$$
i \indep (H \cup T) \setminus (H' \cup T'\cup  K) \mid (H'\cup T') \setminus \{i\}.
$$

For instance, in Figure \ref{exp: marginalizing example}, the conditional independence associated with the edge $\{3,5,6\} \to^{3} \{6\}$ is obtained as the following: we have $H = \{3,5,6\}, T = \tail(H) = \{1,2\}, H' = \{6\}, T' = \tail(H') = \{2\}$  and $K = \{3\}$, hence $$6 \indep (\{3,5,6\} \cup \{1,2\}) \setminus (\{6\} \cup \{2\} \cup \{3\}) \mid \{6\} \cup \{2\} \setminus \{6\} = 6 \indep 1,5 \mid 2.$$

The following definitions will help us to characterize the marginalization sets.

\begin{definition}\label{ceiling}
For a MAG $\G$ and a set of vertices $W$, define the \emph{ceiling} of $W$ as 
$$
\ceil_\G(W) = \{w \in W: W \cap \an_\G(w) = w\}.
$$
Given a head $H$ we define its \emph{Hamlet}\footnote{This nomenclature makes sense on understanding that the \emph{Claudius} of $H$, within a set such that $H$ is barren, is the subset of vertices after strict siblings of $H$ and their descendants are removed.  Note that this set that has been removed is precisely the Hamlet of $H$.} as 
$$
\ham_\G(H) = \sib_\G(\dis_{\an(H)}(H))\setminus \dis_{\an(H)}(H).
$$
\end{definition}

\begin{lemma}\label{characterizing marginalization vertex}
For a MAG $\G$ with heads $H,H' \leq i$, if $H \to^{k} H'$, then $k \in \ceil_\G(\ham_\G(H'))$.
\end{lemma}

\begin{proof}
This is a direct consequence of Lemma \ref{prop: what is minimal marginalization set}.
\end{proof}

Intuitively, $\ham_\G(H)$ serves as the bidirected boundary of $H$ and so must be contained in the marginalization set to reach a graph in which $H$ is the maximal head. Also clearly the last marginalization vertex must be at ceiling of the hamlet, otherwise the barren subset of the district will contain some vertices not in $H$.



For every $i$, we can draw a DAG on all the heads $H \leq i$ with edges that are precisely of the form $H \to^k H'$.  The list of independences associated with the resulting graph, which we refer to as the \emph{complete} power DAG $\mathfrak{I}^\G$ is sufficient to deduce the ordered local Markov property (see Appendix \ref{sec:complete power DAG} and Theorem \ref{thm:head and tail Markov property}). However, this leads to many redundant independences, and we will show that it is sufficient to only include a single independence for each head.  Consequently we will call these simpler power DAGs \emph{refined}. An example of the power DAGs of the graph in Figure \ref{exp: marginalizing example} is given in Figure \ref{exp: complete power DAG and refined power DAG}.

Our definition for a refined power DAG will require a partial ordering on the heads.

\begin{definition}\label{partial ordering}
Define a partial ordering on heads by setting $H < H'$ if and only if $H$ and $H'$ share the same maximal vertex, and $\an_\G(H) \subseteq \an_\G(H')$. 
\end{definition}

\begin{lemma}\label{justify the partial ordering}
For any MAG $\G$ and topological ordering, the relation $<$ is a partial ordering on $\mathcal{H}(\G)$.
\end{lemma}

\begin{proof}
This is a direct consequence of Lemma 4.8 in \citet{Evans2014} as the ordering is a sub-order (in that heads are only comparable if they have the same maximal vertex) of the partial orders defined in that paper.
\end{proof}

\begin{definition}\label{def: refined power DAG}
 For a MAG $\G$ and a topological order $<$, the \emph{refined power DAG} 
 $\widetilde{\mathfrak{I}}^\G_<$
 for $\G,<$ consists of a component for each vertex $i$. Denote this by $\widetilde{\mathfrak{I}}^\G_i$; it has vertices $\{H : H\leq i\}$, and
 an edge $H' \to^{k} H''$ where
 \begin{align*}
        k &= \min \ceil_\G(\ham_\G(H')), \text{ and} \\
        H'' &= \max \{H' : H' \in \pa_{\mathfrak{I}_i (\G)}(H'') \text{ and } H' \to^{k} H''\}.
 \end{align*}
That is, for each head, we only take at most one edge and therefore at most one independence into it. 
\end{definition}
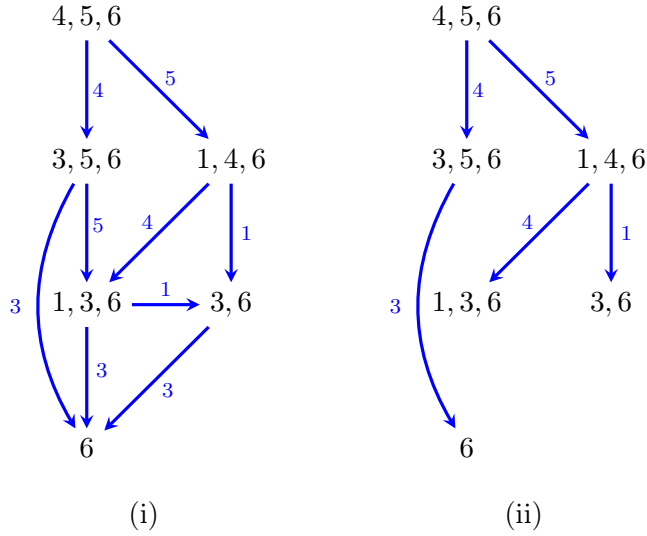
\begin{figure}
\centering
  \begin{tikzpicture}
  [rv/.style={circle, draw, thick, minimum size=6mm, inner sep=0.8mm}, node distance=14mm, >=stealth]
  \pgfsetarrows{latex-latex}
\begin{scope}[yshift=3cm]
  \node[]  (1)            {$4,5,6$};
  \node[, below of=1,yshift=-0.5cm] (2) {$3,5,6$};
  \node[, right of=2,xshift=0.5cm] (3) {$1,4,6$};
  \node[, below of=2,yshift=-0.5cm] (4) {$1,3,6$};
  \node[, below of=3,yshift=-0.5cm] (5) {$3,6$};
  \node[, below of=4,yshift=-0.5cm] (6) {$6$};
  \path[->,draw,very thick,blue]
    (1) edge node [xshift=1.5mm]{\scriptsize{$4$}} (2)
    (1) edge node [xshift=1.5mm,yshift=1.5mm]{\scriptsize{$5$}} (3)
    (2) edge node [xshift=1.5mm,yshift=1mm]{\scriptsize{$5$}} (4)
    (3) edge node [xshift=-1.5mm,yshift=1.5mm]{\scriptsize{$4$}} (4)
    (3) edge node [xshift=2mm]{\scriptsize{$1$}} (5)
    (4) edge node [yshift=2mm]{\scriptsize{$1$}} (5)
    (4) edge node [xshift=2mm,yshift=1mm]{\scriptsize{$3$}} (6)
    (5) edge node [xshift=1.5mm,yshift=-1.5mm]{\scriptsize{$3$}} (6)
    (2) edge [bend right=30] node [xshift=-3mm]{\scriptsize{$3$}} (6);

  \node[below of=6, xshift=7.5mm, yshift=0.5cm] {(i)};
  
  \end{scope}
\begin{scope}[xshift = 5cm,yshift=3cm]
  \node[]  (1)            {$4,5,6$};
  \node[, below of=1,yshift=-0.5cm] (2) {$3,5,6$};
  \node[, right of=2,xshift=0.5cm] (3) {$1,4,6$};
  \node[, below of=2,yshift=-0.5cm] (4) {$1,3,6$};
  \node[, below of=3,yshift=-0.5cm] (5) {$3,6$};
  \node[, below of=4,yshift=-0.5cm] (6) {$6$};
  \path[->,draw,very thick,blue]
    (1) edge node [xshift=1.5mm]{\scriptsize{$4$}} (2)
    (1) edge node [xshift=1.5mm,yshift=1.5mm]{\scriptsize{$5$}} (3)
    (3) edge node [xshift=-1.5mm,yshift=1.5mm]{\scriptsize{$4$}} (4)
    (3) edge node [xshift=2mm]{\scriptsize{$1$}} (5)
    (2) edge [bend right=30] node [xshift=-3mm]{\scriptsize{$3$}} (6);

  \node[below of=6, xshift=7.5mm, yshift=0.5cm] {(ii)};
  
  \end{scope}
\end{tikzpicture}
 \caption{(i) The component of the power DAG for the graph in Figure \ref{exp: marginalizing example}(i); and (ii) the refined version of the same component. Both are on the heads of $\G$ with maximal vertex 6.}
 \label{exp: complete power DAG and refined power DAG}
\end{figure}
Taking the maximal set among parent heads of $H$ requires some 
justification. First we note that it is well defined, since a direct consequence of 
Proposition \ref{prop:maximal independence from maximal parent head} is that if 
$H_1 \to^{k} H'$  and $ H_2 \to^{k} H'$, then $H_3 := \barren_\G(H_1 \cup H_2)$ satisfies $H_3 \to^{k} H'$. Thus there always exists a maximal parent head in any non-empty 
$\{H : H \in \pa_{\mathfrak{I}_i (\G)}(H') \text{ and } H \to^{k} H'\}$.

For simple MAGs, the complete and refined power DAGs are identical and are chains, in the sense that any node has at most one parent head and at most one child. An important motivation inspired by simple MAGs is that the list of conditional independences associated with every edge in the power DAG of a simple MAG decomposes  is $u_\G$, and $u_\G$ is the simplest imset that defines the same model as the graph. 

\subsection{The refined Markov property from the refined power DAGs}\label{sec:very local Markov property}

Next we define the list of independences associated with the refined power DAGs $\widetilde{\mathfrak{I}}^\G_i$. Recall that $[n]$ denotes the set $\{1,\dots, n\}$.

\begin{definition}\label{def: simplified head and tail Markov property}
 For a MAG $\G$ and each $i$, let $\widetilde{\mathbb{L}}^\G_i$ be a list of independences, such that:
 \begin{itemize}
     \item[($a$)]  it contains $i \indep [i-1] \setminus \mb_\G(i,[i]) \mid \mb_\G(i,[i])$, and 
     \item[($b$)] for every head $H'$ other than the maximal one, it contains the independence associated with the unique edge into it in $\widetilde{\mathfrak{I}}^\G_i$
\end{itemize}
We will refer to the collection $\widetilde{\mathbb{L}}^\G = \bigcup_i \widetilde{\mathbb{L}}^\G_i$ as the \emph{refined (ordered) Markov property}.
\end{definition}

\begin{proposition}\label{prop:equiv between refined and local}
For a MAG $\G$, the refined Markov property is equivalent to the ordered local Markov property. Further, it contains fewer and smaller independences compared to the reduced ordered local Markov property.
\end{proposition}

 By `smaller', we mean any independence in the refined Markov property can be deduced from the reduced ordered local Markov property by simply marginalizing some vertices.

Now, since the number of independences is bounded by the number of heads, this will greatly reduce the number of independences arising from the reduced ordered local Markov property.


\begin{remark}
For simple MAGs, the refined Markov property is the simplest possible description of the model and cannot be further reduced; however, for some graphs even less complicated descriptions exist. For example, for the bidirected 5-cycle, adding the semi-elementary imsets corresponding to these independences would give an imset of degree 7. However, in fact one can build an imset that represents the model of 5-cycle of only degree 6, simply by adding the elementary imset corresponding to any valid marginal independence to our `standard' imset. 

Another example is the 5-chain with bidirected edges (see Figure \ref{fig:5-chain and its dual}(i)); its standard imset \emph{is} perfectly Markovian for the graph, but the list of independences that defines the imset model is smaller than our refined Markov property (see Example \ref{exm:5chain}). In Appendix \ref{sec:redundant indep for refined power DAGs}, we give one more interesting example where there are still redundant independences arising in refined power DAGs.
\end{remark}

\begin{remark}
We use the concept of the `Hamlet' to characterize minimal marginalization sets and maximal parent heads. This concept originates in \citet{richardlocalmarkov}, though he did not use this term.  
He uses the Hamlet to reduce the ordered local Markov property by only visiting maximal ancestral sets. This is similar to our approach in that one visits each head only once (one can easily show there is a one-to-one correspondence between heads and maximal ancestral sets), however, \citeauthor{richardlocalmarkov} only reduces the number of ancestral sets visited, while we also marginalize some unnecessary vertices in the independence statements. In Example \ref{exm:non-trivial simple MAG}, for instance, $A = \{1,2,3,4,5,8\}$ is a maximal ancestral set with $\mb_\G(8;A) = \pa_\G(8) = \{1,5\}$ for the head $\{8\}$. The reduced ordered local Markov property gives $8 \indep 2,3,4 \mid 1,5$, but we only need $8 \indep 3,4 \mid 1,5$; this is because $8 \indep 2 \mid 1,3,4,5,6$ has already been obtained for the head $\{6,8\}$, and we already know that $6 \indep 2 \mid 1,3,4,5$ from
variables earlier in the graph. Reducing the number of variables in a independence statement is crucial in practice when trying to obtain a combinatorial imset by adding the semi-elementary imsets corresponding to the list of independences.
\end{remark}

\subsection{Computing refined power DAGs}

In this section, we show that  $\widetilde{\mathfrak{I}}^\G_i$ can be obtained without computing the complete power DAG.  We use Algorithm \ref{algo1} to achieve this and it contains two important ideas:
\begin{itemize}
    \item[(i)] only marginalizing vertices that are smaller (w.r.t.\ a topological ordering) than those already marginalized, and 
    \item[(ii)] for each head, only keep one parent of it, which is on the shortest path from the maximal head.
\end{itemize}

\begin{proposition}\label{prop: algorithm 1 give the refined power dag}
Given a MAG $\G$, Algorithm \ref{algo1} computes $\widetilde{\mathfrak{I}}^\G_i$ for each $i$.
\end{proposition}

\begin{proposition} \label{prop: complexity of algorithm 1}
Given a MAG $\G$ with $n$ vertices, $e$ edges, and maximal head size $k$, then the complexity of Algorithm \ref{algo1} and \ref{alg:algo3} are $O(kn^{k}(n+e))$ and $O(n^{k}(n+e))$, respectively.
\end{proposition}

\begin{algorithm} 
\SetAlgoLined
\SetKw{KwDef}{define}
\KwIn{A MAG $\mathcal{G}([n],\mathcal{E})$}
\KwResult{The refined power DAGs $\widetilde{\mathfrak{I}}^\G_i$ for each $i$}
\KwDef{$M(H)$ the minimum vertex marginalized to reach $H \in \mathcal{H}_i(\G)$}\;
  \KwDef{$SD(H)$ the shortest path to $H \in \mathcal{H}_i(\G)$}\;
\For{$i \in [n]$ \label{line: for i in [n]}}{
  \KwDef{$VS_i$ the set of visited heads for $H \in \mathcal{H}_i(\G)$}\;
  \KwDef{$UVS_i$ the set of not visited heads for $H \in \mathcal{H}_i(\G)$}\;
  Compute the maximal head $H^* = \barren(\mb(i,[i]))$ and set $VS_i = \{\{i\}\}, UVS_i = \{H^*\}$\;
  \If{$H^* = \{i\}$}{Next}
  Set $M(H^*) = i$, $SD(H^*) = 0, M(\{i\}) = SD(\{i\}) = \infty$\;
  Start with a graph $\widetilde{\mathfrak{I}}^\G_i$ with vertex $H^*$\;
  \While{$UVS_i$ is not empty \label{algl:poheads}}{
    Take any $H \in UVS_i$ and \textbf{move} $H$ to $VS_i$\;
    \For{$k \in H \setminus \{i\}$ and $k < M(H)$ \label{algl:smaller}}{ 
    Let $H': H \to^{k} H'$\ \label{line: check heads};\\
     \eIf{$H' \notin  VS_i \cup UVS_i$ \label{line: check if visited}}{ \label{algl:shorter}
     \textbf{add} $H'$ to $\widetilde{\mathfrak{I}}^\G_i$ and $H \to^{k} H'$; $UVS_i = UVS_i \cup \{H'\}$;
     $SD(H') = SD(H)+1$;
     $M(H') = k$;
     }{
     \If{$SD(H') > SD(H)+1$}{
       \textbf{delete} any edges into $H'$\;
     \textbf{add} $H \to^{k} H'$\;
     $SD(H') = SD(H)+1$\;
     $M(H') = k$\;
     }
     }
    }
  }
}

\Return{$( \widetilde{\mathfrak{I}}^\G_1, \widetilde{\mathfrak{I}}^\G_2, \ldots, \widetilde{\mathfrak{I}}^\G_n)$}
\caption{Obtain the refined power DAGs for a general MAG}
\label{algo1}
\end{algorithm}

\begin{algorithm}
    \SetAlgoLined
\SetKw{KwDef}{define}
\KwIn{The refined power DAG $\widetilde{\mathfrak{I}}^\G_i$ of a MAG $\G$ for each $i$}
\KwResult{The imset $u_\G^r$ that is perfectly Markovian w.r.t.~$\G$}
Let $u_\G^r$ be an empty imset;\\
\For{$i \in [n]$ }{
  \For{$H$ in $\widetilde{\mathfrak{I}}^\G_i$ along the partial ordering}{
  \textbf{Compute} $T = \tail(H)$;\\
  \eIf{$\pa(H) = \emptyset$}{
  $u_\G^r := u_\G^r+u_{\langle i, \mid [i]\setminus HT \mid HT \rangle}$
  }{
  Let $H'$ be the \emph{only} parent of $H$ and $T'=\tail(H')$;\\
  $u_\G^r := u_\G^r+u_{\langle i, H'T' \setminus HT \mid HT\rangle}$
  }
  }
}

\Return{$u_\G^r$}
\caption{Obtain the imset $u_\G^r$ from refined Markov property}
\label{alg:algo3}
\end{algorithm}

\begin{remark}
If the input MAG is given to be simple, then the refined power DAG, which is the same as the complete power DAG, can be obtained very fast as shown in Algorithm \ref{alg:algo2} in Appendix. In fact, the complexity of Algorithm \ref{alg:algo2} is linear in the number of edges and vertices, as the structure of the power DAG is inherited completely given the topological ordering.  
\end{remark}

\subsection{Decomposition of the `standard' imset for general MAGs}\label{sec:decomp of u_G}

In this section, we give a decomposition of $u_\G$, using all independences arising from $H \to^{K} H'$ for every possible pair of heads $H, H'$ and every possible marginalization set $K$. Note that this theorem does not have practical use but it proves why empirically the inner product between the standard imset and entropy is zero for any MAGs. Moreover, it shows that if the standard imset is structural then $\I_{u_\G} \subseteq \I_\G$. See discussion after the sketch proof.

\begin{theorem}\label{thm:imset as sum of elementary imset for general MAGs}
For a MAG $\G$, with vertices $[n]$ (topologically ordered), 
we have  
\begin{align*}
u_{\G} &= \sum_{i=1}^n \biggl\{
u_{\langle i,[i-1] \setminus \mb(i,[i]) \mid \mb(i,[i]) \rangle } \\
&\phantomrel{=}+ 
\sum_{\substack{H \in \mathcal{H}(\G) \setminus \{i\}: \\ H \leq i}} 
  \sum_{ \substack{\emptyset \subset K \subseteq H \setminus \{i\}:  \\ H \to^{K} H'}} (-1)^{\abs{K}+1} 
    u_{\langle i,H  T \setminus H' T' K  \mid H'T' \setminus i \rangle } \biggr\}.
\end{align*}
\end{theorem}

 We provide a sketch; the full proof is in Appendix \ref{full proof of decomp of u_G}.

\begin{proof}[Sketch Proof]
We can rewrite $u_{\G}$ as the following:
\begin{align*}
 u_{\G} &=  \delta_{\mathcal{V}}-\delta_{\emptyset}-\sum_{H \in \mathcal{H}(\G)} \sum_{W \subseteq H} (-1)^{\abs{H \setminus W}} \delta_{W \cup \tail(H)};\\
 &=\delta_{\mathcal{V}}-\delta_{\emptyset}-\sum_{H \in \mathcal{H}(\G)}\sum_{K \subseteq H} (-1)^{\abs{K}} \delta_{H T \setminus K};\\
 &=\delta_{\mathcal{V}}-\delta_{\emptyset}- \sum_{i=1}^n \sum_{\substack{H \in \mathcal{H}(\G):\\ H \leq i}}\sum_{K \subseteq H} (-1)^{\abs{K}} \delta_{H T \setminus K};\\
 &=\delta_{\mathcal{V}}-\delta_{\emptyset}- \sum_{i=1}^n \sum_{\substack{H \in \mathcal{H}(\G):\\ H \leq i}}\sum_{K \subseteq H \setminus i} (-1)^{\abs{K}}( \delta_{H T \setminus K}-\delta_{H T \setminus Ki});\\
 &=\sum_{i=1}^n \delta_{[i]}-\delta_{[i-1]}- \sum_{\substack{H \in \mathcal{H}(\G):\\ H \leq i}} \biggl\{\sum_{\emptyset \subset K \subseteq H \setminus i} (-1)^{\abs{K}}( \delta_{H T \setminus K}-\delta_{H T \setminus Ki})
 +\delta_{HT}-\delta_{HT\setminus i} \biggr\}.
\end{align*}
For any $i$ and $H \leq i$, a marginalization set can be any set $K$ such that $\emptyset \subset K \subseteq H \setminus i$. Suppose $H \to^K H'$ where $H,H' \leq i$. We know that $i \indep H  T \setminus H'  T' K  \mid H'T' \setminus i$, and this corresponds to the semi-elementary imset $u_{\langle i,H  T \setminus H'  T' K  \mid H'T' \setminus i \rangle} =\delta_{H T \setminus K}-\delta_{H T \setminus Ki}-\delta_{H' T'\setminus i}+\delta_{H'T'}.$

The last equality already has two terms in this semi-elementary imset. It remains to show that for each $i$, where $1 \leq i \leq n$ and each head $H' \leq i$, $$\sum_{\substack{H,K: \\ H \to^K H', H \leq i}} (-1)^{\abs{K}} = 1.$$ This corresponds to the extra terms $\delta_{HT}-\delta_{HT \setminus i}$. Informally this is saying that for each head, if we count the number of pairs of heads and marginalization sets $K$ that could reach it, and assign a $\pm 1$ to each such head depending on parity of $K$'s size, the sum of these is 1.

Note that given a vertex $i$, there is no head that will marginalize to the maximal head whose maximal vertex is $i$; that is $\barren_\G(\dis_{[i]}(i))$. Instead for this maximal head $H$,  $\delta_{HT}-\delta_{H T \setminus i}$ appears in $u_{\langle i,[i-1] \setminus \mb(i,[i]) \mid \mb(i,[i]) \rangle }$ which contains $\delta_{[i]}-\delta_{[i-1]}$.
%
\end{proof}

\begin{corollary}\label{cor: submodel}
For a MAG $\G$, if $u_{\G}$ is structural, then $\I_{u_{\G}} \subseteq \I_{\G}$.  
\end{corollary}

\begin{remark}
Note also that we can separate the independences out by district, by replacing the first sum in Theorem \ref{thm:imset as sum of elementary imset for general MAGs} by a sum over districts, and then pushing in the summations over vertices in that district.
\end{remark}

\section{Experimental results} \label{sec:exp}

We went through graphs with $\abs{\mathcal{V}} \leq 7$ nodes, checking their standard imset from Theorem \ref{thm: standard imset}.
For $\abs{\mathcal{V}} \leq 4$, all MAGs have combinatorial standard imsets that are 
perfectly Markovian with respect to the graph. We summarize the information in Table \ref{tab: number of graph works} for sizes of graph between 4 and 7. 
For $\abs{\mathcal{V}} = 5$, the only failure is the bidirected 5-cycle mentioned in Section \ref{sec: non Markovian example}.
For $\abs{\mathcal{V}} = 6$, all the imsets which are 
perfectly Markovian with respect to the graph are also combinatorial. For those models where the imset is not perfectly Markovian, only two of 
them are not combinatorial---one of these imsets \emph{is} structural, and both graphs are shown in Figures \ref{fig:counterexample graph}(ii--iii) in Section \ref{sec:main_res}.

\begin{table}[h]
    \centering
    \begin{tabular}{c|r|rrr}
    \toprule
     $|V|$ & 
     {equiv.~classes} & PM & SNPM & NS\\
    \midrule
4 & 19 & 19 & 0 & 0\\
5 & 285 & 284 & 1 & 0\\
6 & 13,303 & 13,248 & 54 & 1\\
7$^*$ & 1,161,461 & 1,146,501 & 14,562 & 8\\
\bottomrule
    \end{tabular}
    \caption{Number of equivalence classes of connected maximal ancestral graphs for various numbers of nodes (for 7 nodes we only include graphs having at most 13 or at least 18 edges.)  PM represents models that are Perfectly Markovian, SNPM those where the imset is Structural but represents a strict subset of the independences (so is Not Perfectly Markovian), and NS where the imset is Not Structural.}
    \label{tab: number of graph works}
\end{table}

\section{Discussion}\label{sec:Discussion}

\subsection{Relation to the work in Andrews}



The main result of \citet{andrews22} is an algorithm that computes an imset that is
guaranteed to define the model induced by an ADMG $\G$, but its complexity 
is exponential in the number of vertices even restricting the maximum head size.  In comparison, our imset $u_\G^r$ in Algorithm \ref{alg:algo3} can be constructed in polynomial time if the maximum head size is bounded.  In 
addition, the imset that \citeauthor{andrews22} gives has much higher degree than $u_\G^r$ for most graphs. 


A focus of both works is on consistent scoring and searching for an optimal model, 
and they are the first to introduce $\BIC_{\MF}$, the inner product between empirical entropy and the standard imset from the 0-1 characteristic imset, to approximate the actual BIC. They conduct a brute force search for Gaussian MAG models with five variables by scoring all MECs and selecting the best one. The result is better than FCI \citep{spirtes2000causation} and GFCI \citep{ogarrio2016hybrid}, and comparable to scoring by the true BIC \citep{drton2009computing}. 

Since, as we have seen, for any MAG with at most five vertices, $u_{\G}$ is perfectly Markovian w.r.t.~the graph, this application of $\BIC_{\MF}$ is valid. (Note that \citet{andrews22} finesses this by assuming the model is a graphoid, and that therefore the imset for the bidirected 5-cycle is also perfectly Markovian.) They also express this slightly differently by saying the imset is valid if every set in the parametrizing set has cardinality at most five (see their Theorem 3.1).  

For general MAGs, our algorithmic advantage in constructing $u_\G^r$ comes from the refined power DAGs, which consider only one parent of each head and marginalize one vertex at each time.  The approach of \citeauthor{andrews22} is to take intersections between all subsets of every Markov blanket which gives the exponential complexity.
As well as looking all graphs up to a certain size, we present theoretical proofs that for simple MAGs and a class of bidirected graphs, the standard imset $u_{\G}$ is valid. 


Some extra assumptions are made by \citeauthor{andrews:phd} in his Theorem 3.1. In particular he makes an assumption that one can use the \emph{intersection axiom}  (see Section \ref{sec:graphoids}), which does not generally hold for conditional independence.  This means that our `standard' imsets work for every MAG with size less than or equal to five. We emphasize that we make no such assumption, and use only properties that hold for all distributions, i.e.~the semi-graphoid axioms. 
Stronger rules can change the simplest representation (standard imset) of an independence models. 
For example, if one assumes positivity of the distribution, then we have already seen that the simplest imset representation for the 5-cycles has degree five, whereas the minimum degree needed without this condition is six.  
We have verified that using compositional graphoids is sufficient for 51 of the 54 graphs with five or six vertices such that $\I_{u_\G} \neq \I_\G$, to define the model.  However, there are three other graphs that require additional axioms.  Specifically, the `standard' imsets for these graphs will define the model if we can assume \emph{ordered downward stability} (see Section \ref{sec:graphoids}).  For further details, see Example \ref{exp: downward stability example} in Appendix \ref{sec:graphoids}.

\subsection{Future work}

This paper leaves open several interesting questions.  First, if we are to use an algorithm to search for the optimal MAG using the $\BIC_{\MF}$-score, what moves should it propose and make?  The 0-1 imset formulation
makes it easy to search, as---provided that Meek's conjecture \citep{meek1997graphical} for MAGs is true---we only need to search by first adding, and later deleting sets from the parametrizing set.


 Another still open problem is to provide a complete solution to obtaining standard imsets of MAGs. This is analogous to undirected graphs, where \citet{kashimura2015standard} provide a solution for the standard imset of non-decomposable undirected graphs; previously graphs with chordless cycles were not covered by the theory.  It seems clear from the example of the bidirected 5-cycle that we will have to choose between having the lowest possible degree and symmetry of the standard imset.  We hope that this can expand the class of graphs that can be scored further, and that it may lead to algorithms that are more accurate or more efficient than the current state-of-the-art.

\subsection*{Acknowledgements}

We thank the Associate Editor and two anonymous reviewers for their very helpful comments and suggestions, including an error in an earlier version of Proposition \ref{prop:induced_subgraph}.  We also thank Bryan Andrews, James Cussens, and Milan Studen\'y for helpful discussions, and Joseph Ramsey for supplying the Tetrad code used to query independences for Example \ref{exp: downward stability example}.


\bibliographystyle{abbrvnat} 
\bibliography{refs.bib}       

\newpage

\begin{appendix}



\section{More definitions}\label{sec:more definition}

\subsection{Separation criterion}\label{sec:msep}
For a path $\pi$ with vertices $v_{i}$, $0 \leq i \leq k$ we call $v_{0}$ and $v_{k}$ the \emph{endpoints} of $\pi$ and any other vertices the \emph{nonendpoints} of $\pi$. For a nonendpoint $w$ in $\pi$, it is a \emph{collider} if $\rqarrow$ $w$ $\lqarrow$ on $\pi$ and a \emph{noncollider} otherwise (an edge $\rqarrow$ is either $\rightarrow$ or $\leftrightarrow$). For two vertices $a,b$ and a disjoint set of vertices $C$ in $\G$ ($C$ might be empty), a path $\pi$ is \emph{m-connecting} $a,b$ given $C$ if (i) $a,b$ are endpoints of $\pi$, (ii) every noncollider is not in $C$ and (iii) every collider is in $\an_{\G}(C)$. A \emph{collider path} is a path where all the nonendpoints are colliders.

\begin{definition}
For three disjoint sets $A,B$ and $C$ ($A,B$ are non-empty), $A$ and $B$ are \emph{m-separated} by $C$ in $\G$ if there is no m-connecting path between any $a \in A$ and any $b \in B$ given $C$. We denote m-separation by $A \perp_m B \mid C$. 
\end{definition}

\begin{definition}\label{ordinary}
A distribution $P(X_{\mathcal{V}})$ is said to satisfy the \emph{global Markov property} with respect to an ADMG $\G$ if whenever $A \perp_m B\mid C$ in $\G$, we have $X_A \indep X_B\mid X_C$ under $P$.
\end{definition}

\subsection{Graphoids} \label{sec:graphoids}

Graphoids are rules for logical implications between conditional independences. There is no finite axiomatization of conditional independences \citep{studeny1992conditional}, but there are some rules that hold for any distribution, for example, the \emph{semi-graphoids}:
\begin{itemize}
    \item[(1)] \emph{symmetry}: $X \indep Y \mid Z \iff Y \indep X \mid Z$; 
    \item[(2)] \emph{decomposition}: $X \indep Y,W \mid Z \implies X \indep Y \mid Z$;
    \item[(3)] \emph{weak union}: $X \indep Y,W \mid Z \implies X \indep W \mid Y,Z$;
    \item[(4)] \emph{contraction}: $(X \indep Y \mid Z) \text{ and } (X \indep W \mid Y, Z) \implies X \indep Y,W \mid Z$.
\end{itemize}

Under the assumption that the distribution is positive, there is an additional property that is guaranteed to hold:
\begin{itemize}
    \item[(5)] \emph{intersection}: $(X \indep Y \mid W,Z) \text{ and } (X \indep W \mid Y,Z) \implies X \indep Y,W \mid Z.$ 
\end{itemize}

The semi-graphoids with (5) are called \emph{graphoids}. Also if one assumes that the distribution is Markov w.r.t.\ a graph, then the following axiom also holds:
\begin{itemize}
    \item[(6)]  \emph{composition}: $(X \indep Y \mid Z) \text{ and } (X \indep W \mid Z) \implies X \indep Y,W \mid Z.$
\end{itemize}

Graphoids with (6) are called \emph{compositional graphoids}, 

\citet{sadeghi2017faithfulness} defines three mode graphoid-like rules; two of these require a `pre-order' $\prec$, which for directed mixed graphs is essentially the partial order implied by the directed edges.
\begin{itemize}
    \item[(7)] \emph{singleton transitivity}: 
    $(X_i \indep X_j \mid Y) \text{ and } (X_i \indep X_j \mid Y, X_k) \implies (X_i \indep X_k \mid Y)\text{ or }(X_j \indep X_k \mid Y).$
    \item[(8)] \emph{ordered upward stability}: 
    $(X_i \indep X_j \mid Y)  \implies (X_i \indep X_j \mid Y, X_k)$ for any $k \prec i$ or $k \prec j$.
    \item[(9)]  \emph{ordered downward stability}: 
    $(X_i \indep X_j \mid Y, X_k)  \implies (X_i \indep X_j \mid Y)$ for any $k$ that is either larger than or incomparable to $i$, $j$ and every element of $Y$.
\end{itemize}
Here is an example where ordered downward stability is needed to deduce all the independences in the graph.

\begin{example}\label{exp: downward stability example}

\begin{figure}
    \centering
    \begin{tikzpicture}
     [rv/.style={circle, draw, thick, minimum size=6mm, inner sep=0.8mm}, node distance=14mm, >=stealth]
  \pgfsetarrows{latex-latex}
\begin{scope}
  \node (0) {};
  \node[rv]  (1) at (0:12.5mm)            {$1$};
  \node[rv]  (2) at (60:12.5mm)            {$2$};
\node[rv]  (3) at (120:12.5mm)            {$3$};
\node[rv]  (4) at (180:12.5mm)            {$4$};
\node[rv]  (5) at (240:12.5mm)            {$5$};
\node[rv]  (6) at (300:12.5mm)            {$6$};
  \draw[beg] (1) -- (2);
  \draw[beg] (2) -- (3);
  \draw[beg] (3) -- (4);
  \draw[beg] (4) -- (5);
  \draw[beg] (6) -- (5);
  \draw[beg] (6) -- (1);
  \draw[->, very thick, blue] (1) -- (3);
  \node[below of=0, yshift=-7mm] {(i)};
  \end{scope}
\begin{scope}[xshift=4.5cm]
  \node (0) {};
  \node[rv]  (1) at (0:12.5mm)            {$1$};
  \node[rv]  (2) at (60:12.5mm)            {$2$};
\node[rv]  (3) at (120:12.5mm)            {$3$};
\node[rv]  (4) at (180:12.5mm)            {$4$};
\node[rv]  (5) at (240:12.5mm)            {$5$};
\node[rv]  (6) at (300:12.5mm)            {$6$};
  \draw[beg] (1) -- (2);
  \draw[beg] (2) -- (3);
  \draw[beg] (3) -- (4);
  \draw[beg] (4) -- (5);
  \draw[beg] (6) -- (5);
  \draw[beg] (6) -- (1);
  \draw[->, very thick, blue] (1) -- (4);
  \node[below of=0, yshift=-7mm] {(ii)};
  \end{scope}
  \begin{scope}[xshift=9cm]
  \node (0) {};
  \node[rv]  (1) at (0:12.5mm)            {$1$};
  \node[rv]  (2) at (60:12.5mm)            {$2$};
\node[rv]  (3) at (120:12.5mm)            {$3$};
\node[rv]  (4) at (180:12.5mm)            {$4$};
\node[rv]  (5) at (240:12.5mm)            {$5$};
\node[rv]  (6) at (300:12.5mm)            {$6$};
  \draw[beg] (1) -- (2);
  \draw[beg] (2) -- (3);
  \draw[beg] (3) -- (4);
  \draw[beg] (4) -- (5);
  \draw[beg] (6) -- (5);
  \draw[beg] (6) -- (1);
  \draw[->, very thick, blue] (1) -- (3);
  \draw[->, very thick, blue] (1) -- (4);
  \node[below of=0, yshift=-7mm] {(iii)};
  \end{scope}
  \end{tikzpicture}
    \caption{Three graphs that do not define the model without assuming rules beyond composition and intersection.}
    \label{fig:ds_graphs}
\end{figure}
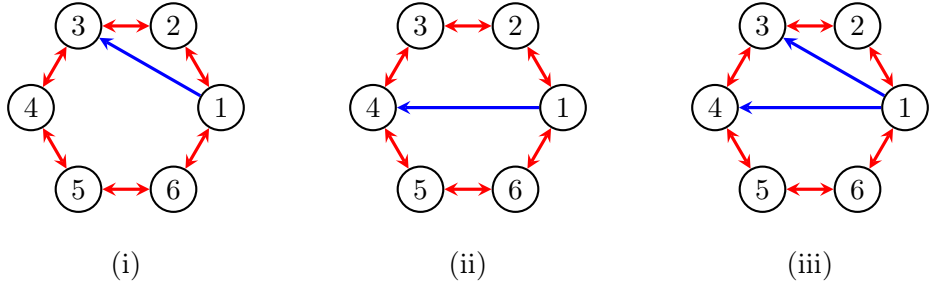
Consider the bidirected 6-cycle with one additional directed edge in Figure \ref{fig:ds_graphs}(i). The imset $u_\G$ of this graph $\G$ is structural and $\I_{u_\G}$ contains the following independences:
\begin{align*}
    6 &\indep 2 & 6 &\indep 4 \mid 1,3 &
    2 &\indep 4 \mid 1,6 & 2 &\indep 5 \mid 4,6\\
    1 &\indep 4 &
    1 &\indep 5 \mid 2,4 &
    6 &\indep 3 \mid 1,5 &
    3 &\indep 5 \mid 1,2.
\end{align*}
In particular, even with the intersection and compositional graphoids, we do not get any joint independences, when the graph says we should have (e.g.) $2 \indep 4, 5, 6$.  In the list of the local Markov property, it does not satisfy $4 \indep 1,2$ or $5 \indep 1,2$ or $4 \indep 2,6$. However, ordered downward stability does help.  The only pair that is ordered is $1 \to 3$, with all other pairs being incomparable.  Then we can deduce $2 \indep 5$ by removing 4 and 6 from $2 \indep 5 \mid 4,6$, and use composition to obtain $2 \indep 5,6$; we can then obtain $2 \indep  4,5,6$.  Similarly we can remove first 3 and then 1 from $6 \indep 4 \mid 1,3$ to obtain the marginal independence, and get $4 \indep 2,6$ using composition (or just contraction). Also, we can remove $2,4$ from $5 \indep 1 \mid 2,4$ to get $5 \indep 1$, then combined with $5 \indep 2$, we have $5 \indep 1,2$. On the other hand, we can obtain $4 \indep 2$ by removing 6 and 1 from $4 \indep 2 \mid 1,6$, then use composition to get $4 \indep 1,2$, as $4 \indep 1$ is already in $\I_{u_\G}$.
\end{example}

\section{Missing materials from section \ref{sec:main_res}} \label{sec:main_res_res}

A useful fact about MAGs is that any induced subgraph of 
a MAG is itself a MAG.

\begin{lemma} \label{lem:subMAGs_are_MAGs}
Let $\G$ be a graph that is maximal and ancestral.  Then for
any subset of the vertices $W \subseteq V$, so is the induced
subgraph $\G_W$.
\end{lemma}

\begin{proof}
The ancestrality follows from Proposition 3.5 of \citet{richardson2002}.  
For maximality, note that there are no more paths in an induced 
subgraph than in the original graph, so in particular there cannot be 
any more inducing paths.  Since all bidirected edges between vertices 
are preserved, this implies that the graph remains maximal.
\end{proof}

Another useful fact will concern conditional independences in induced subgraphs.

\begin{proposition} \label{prop:msep_subgraph}
Let $\G$ be a MAG and $\G_W$ an induced subgraph.  Then for $a,b \in W$, any m-separation $a \perp_m b \mid C$ holding in $\G$ implies that $a \perp_m b \mid E$ holds in $\G_W$, where $E = C \cap W$.
Additionally, if $W$ is an ancestral set, then $a \perp_m b \mid E$ in $\G_W$ if and only if it also holds in $\G$.
\end{proposition}

\begin{proof}
Removing vertices that are not on the path from the conditioning set can only block a path, so the removal of vertices in $C \setminus E$ will not affect the status of any of the paths through $W$.  Hence the result holds.  The result for ancestral subgraphs follows, for example, from the results of \citet{richardlocalmarkov}.
\end{proof}

The above result is also useful for considering the parametrizing set.  
Let $\mathcal{P}(W)$ denote the \emph{power set} of $W$, i.e.~the collection of subsets of $W$. 

\begin{corollary}
We have $\mathcal{S}(\G_W) \subseteq \mathcal{S}(\G) \cap \mathcal{P}(W)$.
\end{corollary}

\begin{proof}
This follows immediately from Propositions \ref{prop:parametrizing set and independence} and \ref{prop:msep_subgraph}.
\end{proof}

Note that, of course, the sets of size at most two are always identical in the original MAG and any induced subgraph, since these are just the adjacencies.

The next result shows that the certain graphs must not appear as induced subgraphs within a MAG, if we want the corresponding 
`standard' imset to be perfectly Markovian with respect to it.

\medskip

\begin{proa}{\ref{prop:induced_subgraph}}
Let $\G$ be a MAG, and suppose that for some ancestral subset $A \subset V$ we have that $u_{\G_A}$ is not Markovian with respect to $\G_A$.  Then the model $u_\G$ is not perfectly Markovian with respect to $\G$.
\end{proa}

\begin{proof}
Since the subgraph is ancestral, the independences it encodes 
are just those from the larger graph, with potentially smaller 
conditioning sets.  Now, suppose that it is not possible
to write the independences implied by the smaller graph in
such a way as to avoid repeating a set.  Then the same problem
will clearly arise in the larger graph, since (by Proposition 
\ref{prop:msep_subgraph}) we must specify isomorphic independences 
with potentially more restrictions on the conditioning set.
\end{proof}

\begin{remark}
We know from Theorem \ref{thm:imset as sum of elementary imset for general MAGs} that if we take an ancestral subset $W=A$, the structure of the imset is preserved.  Hence, if we marginalize the graph then the imset will just be the induced subimset over the entries that are subsets of $A$, and so by Proposition 9.3 of \citet{studeny2006probabilistic}, the imset will match the model for the ancestral subgraph.  Hence, in this case the result is clear.
\end{remark}

The following example shows that Proposition \ref{prop:induced_subgraph} does not work if a subset $W$ is not ancestral. This was pointed out by one of the anonymous reviewers.
\begin{example}
 \begin{figure}
    \centering
     \begin{tikzpicture}
  [rv/.style={circle, draw, thick, minimum size=6mm, inner sep=0.8mm}, node distance=14mm, >=stealth]
  \pgfsetarrows{latex-latex}
\begin{scope}
  \node (0) {};
    \node[rv]  (1) at (18:12.5mm)        {$1$};
  \node[rv]  (2) at (90:12.5mm)            {$2$};
\node[rv]  (3) at (162:12.5mm)            {$3$};
\node[rv]  (4) at (234:12.5mm)            {$4$};
\node[rv]  (5) at (306:12.5mm)            {$5$};
\node[rv] (6)        {$6$};
\draw[deg] (1) -- (6);
\draw[deg] (2) -- (6);
\draw[deg] (6) -- (4);
  \draw[beg] (1) -- (2);
  \draw[beg] (2) -- (3);
  \draw[beg] (3) -- (4);
  \draw[beg] (4) -- (5);
  \draw[beg] (1) -- (5);
  \end{scope}

\end{tikzpicture}
    \caption{A counter example for Prop \ref{prop:induced_subgraph} when $W$ is not ancestral}
    \label{fig: counterexample graph for prop:induced}
\end{figure}   

One can check that $u_{\langle 4,12\mid 6 \rangle}+u_{\langle 6,35 \mid 12 \rangle}+u_{\langle 1,3 \mid 5 \rangle}+u_{\langle 2,5 \mid 3\rangle}+u_{\langle 3,5 \rangle}$ is the standard imset of the MAG in Figure \ref{fig: counterexample graph for prop:induced} and is perfectly Markovian with respect to it. The subgraph induced by $\{1,2,3,4,5\}$ is the bidirected 5-cycle, however, and hence the `standard' imset does not work.

For the graph, if we were to interpret it using the nested property then, 
after fixing 6, we would have two additional constraints: $4 \indep 1 \mid 3$ and
$4 \indep 2 \mid 5$.  This model \emph{cannot} be rewritten in a manner that avoids 
overlap, and therefore if the imset represented the nested model, this graph would not
have a perfectly Markovian `standard' imset.  Of course, the imset \emph{does not} 
represent the nested model, and so this is merely an academic point. 
\end{example}

\begin{lema}{\ref{lemma:linear relation between u and c} }
For a structural imset $u$ and its characteristic imset $c$, we have that
$c(S) = 1$ if and only if $u$ contains no independence with $S$ as a constrained set, and otherwise $c(S) \leq 0$.
\end{lema}

\begin{proof}
As $u$ is structural, there exists an positive integer $k$ and some conditional independences $I$ such that
$$
k \cdot u = \sum_I u_I.
$$

The transformation from standard imsets to characteristic imsets is bijective and linear. 
Then 
\begin{align*}
    c(S) &= 1-\sum_{T: S \subseteq T \subseteq \mathcal{V}} u(T)\\
    &= 1-\sum_{T: S \subseteq T \subseteq \mathcal{V}} k^{-1} \sum_{I} u_{\langle a,b \cmid C\rangle}(T)\\
    &= 1 - k^{-1} \sum_{I} \sum_{C' \subseteq C} \delta_{abC'}(T).
\end{align*}
Hence we see that for any elementary imset that is represented in
$u$ the corresponding constrained sets are all strictly less than one, and hence at most zero; 
conversely any set not so represented will have value 1.
Notice that for an empty standard imset, its characteristic imset is just a vector of 1s. 
\end{proof}



\begin{cora}{\ref{cor:imset must agree on S(G)}}
Consider a MAG $\G$. For any structural imset $u$ such that $\I_{u}=\I_{\G}$, its characteristic imset $c$ must be an integer valued vector and satisfy the following:
\begin{itemize}
    \item[(i)] if $S \in \Sset(\G)$, then $c_\G(S) = 1$;
    \item[(ii)] if $S \notin \Sset(\G)$, then $c_\G(S) \leq 0$. 
\end{itemize}
\end{cora}

\begin{proof}
 This is a direct consequence of Lemma \ref{lemma:linear relation between u and c} and Proposition \ref{prop:parametrizing set and independence}.
 If $S \notin \Sset(\G)$ then there exists $k \in \mathbb{N}$ such 
 that $k \cdot u - u_{\langle a b |C \rangle}$ is combinatorial and 
 $\{a,b\} \subseteq S \subseteq \{a,b\} \cup C$, hence by the same 
 argument used in Lemma \ref{lemma:linear relation between u and c} we
 have $c(S) \leq 0$.

  If $S \in \Sset(\G)$ then there is no independence 
 represented by $\G$ that has $S$ as a constrained set. By Lemma \ref{lemma:linear relation between u and c}, $c(S) = 1$.
%
%
\end{proof}

\subsection{Simple MAGs}\label{sec: density of simple MAGs}

\begin{example} \label{exm:simple_MAG_large_district}
We first start by providing an example of a simple MAG with an arbitrarily 
large district.  This illustrates that a search algorithm over simple MAGs
is potentially very useful, since it includes many more causal models than
one would obtain by restricting the maximum district size to two or three.  
This is in contrast with other methods for score-based learning, which 
generally make this kind of restriction \citep[e.g.][]{chen2021integer}.  

\begin{figure}
\centering
  \begin{tikzpicture}
  [ev/.style={circle, minimum size=7mm, inner sep=0.5mm}, node distance=14mm, >=stealth]
  \pgfsetarrows{latex-latex}
\begin{scope}
  \node[ev]  (k)             {$v_k$};
  \node[ev, above right of=k] (kp)  {$p_{k}$};
  \node[ev, left of=k, yshift=6mm] (k1)  {$v_{k-1}$};
  \node[ev, above left of=k1] (k1p)  {$p_{k-1}$};
  \node[ev, right of=k1, yshift=6mm] (k2)  {$v_{k-2}$};
  \node[ev, above right of=k2] (k2p)  {$p_{k-3}$};
    \node[left of=k2, yshift=6mm] (k3)  {\phantom{$v_{k-3}$}};
    \node[xshift=7mm, yshift=6mm]  (dots) at (k3)             {$\vdots$};
    \node[xshift=7mm, yshift=2mm]  (dotsr) at (dots)             {$\vdots$};
    \node[xshift=-7mm, yshift=-2mm]  (dotsl) at (dots)             {$\vdots$};
      \node[ev, above of=k3, yshift=2mm]  (3)             {$v_3$};
        \node[right of=3, yshift=-6mm] (4)  {\phantom{$p_{k-1}$}};
  \node[ev, above left of=3] (3p)  {$p_{3}$};
  \node[ev, right of=3, yshift=6mm] (2)  {$v_{2}$};
  \node[ev, above right of=2] (2p)  {$p_{2}$};
  \node[ev, left of=2, yshift=6mm] (1)  {$v_{1}$};
  \node[ev, above left of=1] (1p)  {$p_{1}$};
  %
  \draw[deg] (kp) -- (k);
  \draw[beg] (k) -- (k1);
  \draw[deg] (k1p) -- (k1);
  \draw[beg] (k2) -- (k1);
  \draw[deg] (k2p) -- (k2);
  \draw[deg] (k2) -- (k);
  \draw[beg] (k2) -- (k3);
  \draw[deg] (k3) -- (k1);
    \draw[beg] (4) -- (3);
    \draw[deg] (2) -- (4);
    \draw[deg] (3p) -- (3);
  \draw[beg] (2) -- (3);
  \draw[deg] (1) -- (3);
  \draw[deg] (2p) -- (2);
  \draw[beg] (2) -- (1);
  \draw[deg] (1p) -- (1);
  \draw[deg] (3) -- (dotsl);
  \draw[deg] (dotsr) -- (k2);
  \end{scope}

\end{tikzpicture}
 \caption{Simple MAGs with arbitrarily large districts}
 \label{fig:simple_MAG_large_districts}
\end{figure}
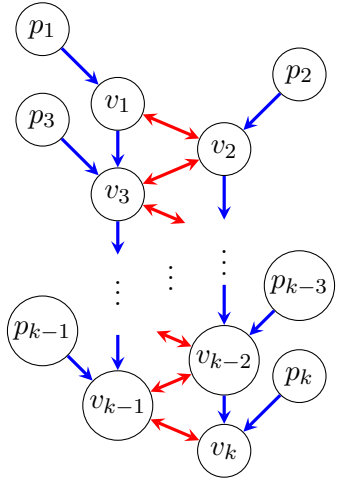

The graph in Figure \ref{fig:simple_MAG_large_districts} has a 
district of size $k$ and that district also has $k$ parents.  However,
note that the only heads are just the bidirected edges, so this is 
a simple MAG.  In addition, the independences for each of the $V_j$ are 
of the form
\begin{align*}
    V_j \indep P_1,V_1,\ldots,V_{j-3},P_{j-2},P_{k-1} \mid P_j,V_{j-2},
\end{align*}
and (unsurprisingly) there is no way to order the variables so that 
the independences are nested within one another as would be required
by a DAG.

\end{example}

We demonstrate how common simple MAGs are by using the following results. The first column in Table \ref{table: number of equiv.class} is the number vertices. Then the second column list the number of equivalence classes of MAGs with the corresponding number of vertices. The next two columns further count how many equivalence classes that contain at least one simple MAG and one DAG, respectively. In particular, the proportion of equivalence classes that contain simple MAGs decreases but not as sharply as that for DAGs.

Then Figure \ref{Figure: random MAGs being simple}, we simulate 1000 random MAGs for number of vertices ranging from five to forty, and plot empirical probabilities that the simulated graphs are Markov equivalent to some simple MAGs. The method we use to simulate MAGs is the same as \citet{claassen2022greedy}. We fix the average and maximal degree of each vertex to three and ten respectively. For each edge, the probability of being bidirected is 0.2. We simulate an ADMG first and project it into a Markov equivalent MAG. Then by converting the MAG to a \emph{partially ancestral graph} (PAG), we finally select a representative MAG from the PAG. The last two steps are from \citet{zhang2007characterization}. If the representative MAG is simple, we consider the original graph as being Markov equivalent to some simple MAGs. 

  \begin{minipage}{\textwidth}
  \begin{minipage}[b]{0.50\textwidth}
    \centering
    \includegraphics[scale=0.45]{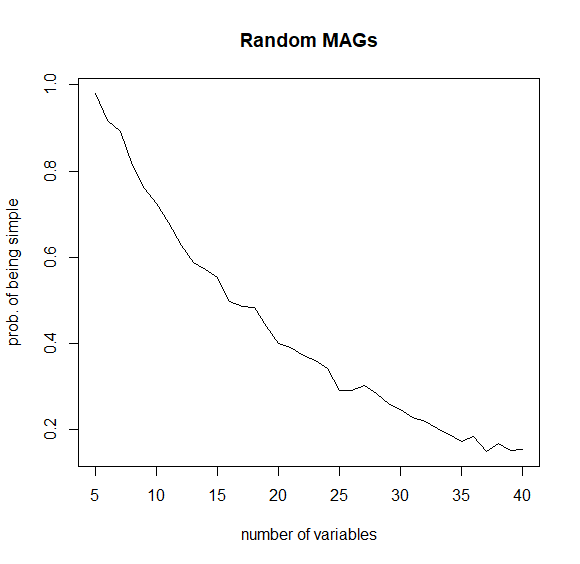}
    \captionof{figure}{A plot for probability of random graphs being Markov equivalent to some simple MAGs}\label{Figure: random MAGs being simple}
  \end{minipage}
  \hfill
  \begin{minipage}[b]{0.50\textwidth}
    
     \begin{tabular}[h]{c|c|c|c}
\toprule
$|V|$ & equiv.~classes & simple MAGs & DAGs\\
\midrule
5 & 285 & 205 & 119\\
6 & 13,303 & 6,278 & 2,025\\
7$^*$ & 1,161,461 & 331,310 & 57,661\\
\bottomrule
\multicolumn{4}{c}{$^*$having at most 13 or at least 18 edges.}\\
\end{tabular}
      \captionof{table}{Number of equivalence classes}\label{table: number of equiv.class}
    \end{minipage}
  \end{minipage}

\subsection{Proof of Theorem \ref{thm: standard imset with head size less 3}}\label{Proof of simple mag Theorem}

\begin{proof}[Proof of Theorem \ref{thm: standard imset with head size less 3}]
Given a topological ordering, we can write the standard imset $u_{\mathcal{G}}$ of a simple MAG $\G$ in the following way:
\begin{align*}
u_{\mathcal{G}}&=  \delta_{\mathcal{V}}-\delta_{\emptyset}-\sum_{i \in \mathcal{V}} \biggl\{ \delta_{\{i\} \cup \pa(i)} - \delta_{\pa(i)} \biggr\} \\
&\phantomrel{=}- \sum_{i \leftrightarrow j} \biggl\{ \delta_{\{i,j\} \cup \tail(i,j)}-\delta_{\{i\} \cup \tail(i,j)} -\delta_{\{j\} \cup \tail(i,j)} +\delta_{\tail(i,j)} \biggl\}\\
&= \sum_{i \in \mathcal{V}} \biggl\{ \delta_{[i]} - \delta_{[i-1]} -\delta_{\{i\} \cup \pa(i)} + \delta_{\pa(i)} \biggr\} \\
&\phantomrel{=}- \sum_{i \leftrightarrow j} \biggl\{  \delta_{\{i,j\} \cup \tail(i,j)}-\delta_{\{i\} \cup \tail(i,j)} -\delta_{\{j\} \cup \tail(i,j)} +\delta_{\tail(i,j)} \biggr\}\\
&= \sum_{i \in \mathcal{V}} \biggl\{ \delta_{[i]} - \delta_{[i-1]} -\delta_{\{i\} \cup \pa(i)} +\delta_{\pa(i)}\\
&\phantomrel{=} + \sum_{i \leftrightarrow j, i > j} - \delta_{\{i,j\} \cup \tail(i,j)}+\delta_{\{i\} \cup \tail(i,j)} +\delta_{\{j\} \cup \tail(i,j)} -\delta_{\tail(i,j)}\biggr\}\\
\end{align*}
For each vertex $i$, consider the topological ordering on all the $j$ such that $j < i$ and $i \leftrightarrow j$, where $j_1 < ... < j_k <i$. Hence we have:
\begin{align*}
u_{\mathcal{G}}&= \sum_{i \in \mathcal{V}} \biggl\{ \delta_{[i]} - \delta_{[i-1]} -\delta_{\{i\} \cup \pa(i)} +\delta_{\pa(i)}\\
&\phantomrel{=} + \sum_{l=1}^k -\delta_{\{i,j_l\} \cup \tail(i,j_l)}+\delta_{\{i\} \cup \tail(i,j_l)} +\delta_{\{j_l\} \cup \tail(i,j_l)} -\delta_{\tail(i,j_l)}\biggr\}\\
&= \sum_{i \in \mathcal{V}} \biggl\{ \delta_{[i]} - \delta_{[i-1]} -\delta_{\{i,j_k\} \cup \tail(i,j_k)} + \delta_{\{j_k\} \cup \tail(i,j_k)} \\
&\phantomrel{=} +\sum_{l=1}^{k-1} -\delta_{\{i,j_{l}\} \cup \tail(i,j_l)} + \delta_{\{i\} \cup \tail(i,j_{l+1})} +\delta_{\{j_{l}\} \cup \tail(i,j_l)}-\delta_{\tail(i,j_{l+1})} \\
&\phantomrel{=} +\delta_{\{i\} \cup \tail(i,j_1)}-\delta_{\tail(i,j_1)}-\delta_{\{i\} \cup \pa(i)} +\delta_{\pa(i)}\biggr\}.
\end{align*}
Now for each vertex $i$, consider the following list of conditional independence, denoted by $\mathbb{L}_i$:
\begin{align*}
i &\indep [i-1] \setminus (\tail(i,j_k) \cup \{j_k\}) \mid \tail(i,j_k)\cup \{j_k\}\\
i &\indep \tail(i,j_k) \setminus (\tail(i,j_{k-1}) \cup \{j_{k-1}\}) \mid \tail(i,j_{k-1})\cup \{j_{k-1}\}\\
&\phantomrel{=}\vdots\\
i &\indep \tail(i,j_2) \setminus (\tail(i,j_1) \cup \{j_{1}\}) \mid \tail(i,j_1) \cup \{j_1\}\\
i &\indep \tail(i,j_1) \setminus \pa(i) \mid \pa(i).
\end{align*}
It is straightforward to check that $u_{\mathcal{G}}$ is a sum of the semi-elementary imsets corresponding to the conditional independence list $\mathbb{L}^{\G} = \bigcup_{i \in \mathcal{V}} \mathbb{L}_i$. Notice that if there is no such $j$ for $i$ ($k=0$) then it reduces to the ordered local Markov property of DAGs: $$i \indep [i-1] \setminus \pa(i) \mid \pa(i).$$ Thus $u_{\mathcal{G}}$ is a combinatorial imset and every independence in $\mathbb{L}$ is in $\I_{u_{\G}}$.

By Theorem \ref{thm:head and tail Markov property}, any ordered local Markov property can be deduced from $\mathbb{L}^{\G} = \bigcup_i \mathbb{L}_i$ using semi-graphoid thus $\I_\G \subseteq \I_{u_\G}$. Now  the faithfulness result in \citet{richardson2002} shows that for every MAG $\G$ there exists a distribution
$P$ such that $\I_P = \I_{\G}$  and Theorem 5.2 in \citet{studeny2006probabilistic} implies the existence of
a structural imset $u$ such that $I_u = I_P$. Now every independence in $\mathbb{L}^{\G}$ is in $\I_{u}$ by Theorem \ref{thm:head and tail Markov property}, and $u_{\G}$ is
sum of the semi-elementary imsets corresponding to the independence in $\mathbb{L}^{\G}$, so by Lemma 
6.1 in \citet{studeny2006probabilistic}, we have $I_{u_{\G}} \subseteq \I_u = \I_{\G}$. Note the idea of this proof is similar to
the proof of Lemma 7.1 in \citet{studeny2006probabilistic}.
\end{proof}

\section{Some useful results on heads}

\subsection{Characterization of the marginalization sets}

We first study some properties of the marginalization sets $K$.

\begin{lemma}\label{the minimal set of vertices from one head to another head}
For a MAG $\G$ and two heads $H,H'$ whose maximal vertex is $i$, if $H \to^{K} H'$ and $H \to^{L} H'$ then $H \to^{K \cap L} H'$. 
\end{lemma}
\begin{proof}
Let $T = K \cap L$ and $\dot{K} = K \setminus T$. Suppose $H \to^{T} H''$. Now after marginalizing $T$, the vertices in $\dot{K}$ are either not in the same district as $i$ or lie in $H''$. If all of them are not in the same district with $i$, then $H' = H''$ because then marginalizing $\dot{K}$ would not change the barren subset of the district containing $i$, which are $H''$ and $H'$ before and after the marginalization, respectively. 

Let $\Tilde{K}$ be those lie in $H_3$ and similarly we define $\Tilde{T}$ for $L$, note that they are disjoint by definition, because $T$ is the intersection of $K$ and $L$, and $\dot{K}, \dot{K}_2$ are the complement of $T$ in $K$ and $L$ respectively. Also $H \to^{T \cup \Tilde{K}} H'$ and $H \to^{T \cup \Tilde{T}} H'$. 

We firstly consider $H \to^{T \cup \Tilde{K}} H'$. This implies that for any bidirected path from any vertex $i \in \Tilde{T}$ to any $w \in H' \cup T'$, there is a vertex from $\Tilde{K}$ on the path, otherwise some vertices of $\Tilde{T}$ would be preserved, which is a contradiction. The equivalent statement by swapping $\Tilde{K}$ and $\Tilde{T}$ also holds. If one considers the first vertex in either $\Tilde{K}$ or $\Tilde{T}$ on any such path, one of the two statements would be false. Hence the lemma is true. \end{proof}

An implication of this result is that for any $H, H' \leq i$ , if $H$ is a parent head of $H'$ on the power DAG for $i$, there exists a unique \emph{minimal marginalization set of vertices} that leads from $H$ to $H'$.
\begin{definition}\label{minimal marginalizing set }
For a MAG $\G$ and two heads $H, H' \leq i$, let $K^{m}$ be the minimal marginalization set of vertices such that $H \to^{K^{m}} H'$.
\end{definition}

\begin{lemma}\label{construct K given minimal K}
For a MAG $\G$ and two heads $H, H' \leq i$, and $H \to^{K^m} H'$ , then $H \to^{K} H'$ if and only if $K = K^m \dot{\cup} B $ where $B \subseteq (H \setminus (H' \cup K^m))$.
\end{lemma}

\begin{proof}
($\Rightarrow$): By definition of $K^m$, $K^m \subseteq K$, and $H \setminus (H' \cup K^m)$ are simply all the vertices we can marginalize if we want to reach the head $H'$.

($\Leftarrow$): This is because after marginalization of some vertices in the barren subset, the remaining vertices are either outside of the Markov blanket of $i$ or it stays in the barren subset of the district, i.e.\ the head. Then by the definition of $K^m$,  we may marginalize any other irrelevant vertices in the barren subset after marginalization of $K^m$ as they would be outside of the Markov blanket.
\end{proof}

We only require one more lemma to validate our decomposition of $u_G$. One key step is to find situations when there is only one marginalization set, which is also the minimal marginalization set. We propose the following definitions.

\begin{lemma}\label{characterization of minimal set}
Consider $H' \to^{K} H $ where $K = H' \setminus H$ is the minimal marginalization set of vertices, i.e.\ $K$ is the only marginalization set to reach $H$ from $H'$. Then this happens if and only if  $H'=\barren(K \cup H)$ where $K \subseteq \ceil(\ham(H))$.
\end{lemma}

\begin{proof}
($\Rightarrow$): After marginalizing vertices in the barren subset, the remaining vertices of $H'$ either stay in the barren subset or outside of the Markov blanket. Since $K$ is the only marginalization set of vertices, we know the remaining vertices stay in the barren subset (also in $H$). Hence $H'=\barren(K \cup H)$. We first show that $H' \setminus H \subseteq \ham(H)$. Suppose, 
for a contradiction, that it is not true; let $K^*  = K \setminus (\sib(\dis_{\an(H)}(i))\setminus \dis_{\an(H)}(i))$. 
Consider any bidirected path between any $i \in K^*$ and any $k \in \dis_{\an(H)}(i)$, the first vertex $x$ from $k$ that is not in $\dis_{\an(H)}(i)$ lies in $\an(K) \setminus \an(H)$. 

If $x$ is in $\an(K) \setminus (\an(H) \cup K)$ we have after marginalizing $K$ then $\barren(\dis_{H \setminus K}(i)) \neq H$, because it would include a descendant of $x$ (possibly $x$ itself) that is \emph{not} an ancestor of $H$. This is a contradiction to our assumption.

Hence $x$ must lie in $K$, so by assumption $x \in K \setminus K^*$. Since the choice of the path and vertices are free, it means that if we marginalize $K \setminus K^*$ then $K^*$ would be outside the district of $i$, so $K$ is not minimal.  Hence we reach another contradiction.

Then $K \subseteq \ceil(\ham(H))$ follows, as if it has some ancestors that are also siblings of $\dis_{\an(H)}(i)$ then after marginalizing $K$, the barren subset of the district would not be $H$.

($\Leftarrow$): Let $K$ be any subset of $\ceil(\ham(H))$ (but not $\emptyset$) and consider $H' = \barren(K \cup H)$. Clearly $H'$ is a head. Moreover $K \subseteq H'$ as $K$ either has no ancestral relation with $H$ or $K$ are descendent of $H$; in particular this means that $K$ is a valid marginalization set for $H'$. We still need to show that $(1)$ after marginalizing $K$, we reach the head $H$  and $(2)$ $K$ is minimal.

For (1), suppose there is a vertex $t \notin K$ stays in the barren subset of the district but $t$ is not in $H$. Consider any bidirected path from $t$ to $\dis_{\an(H)}(i)$ in $\G_{\an(H')}$. By definition, on this path there is a vertex in $K$ (not just $\an(K)$) next to some vertex in $\dis_{\an(H)}(i)$, then this path is removed after marginalizing $K$. Also note that all vertices in $\dis_{\an(H)}(i)$ stay in the graph.

For (2), If we marginalize some subset of $K$, then the remaining vertices of $K$ stay in the Markov blanket and hence in the barren subset of the district, which then would not be $H$.
\end{proof}

\subsection{Existence of maximal parents in the power DAG}

\begin{lemma}\label{lemma:maximal marginalization set}
Suppose that for two heads $ i \geq H, H'$, we have $H \to^{K} H'$. Then $H \to^{L} H'$ for $L = H \setminus H'$.
\end{lemma}

\begin{proof}
Clearly, any marginalization set from $H$ to $H'$ must not contain any element from $H'$. By existence of the minimal marginalization set, we can just adding all irrelevant vertices to the set.
\end{proof}

\begin{proposition}\label{prop: maximal parent in power dag}
For a MAG $\G$, suppose that for three heads $ i \geq H_1,H_2, H$, we have $H_1 \to^{K} H$ and $H_2 \to^{L} H$. Then $H_3 = \barren(H_1 \cup H_2)$ is a head and  $H_3 \to^{K'} H$ for $K' = H_3 \setminus H$. This means that in the power DAG for $i$, if a head has a parent head, then there exists a maximal parent head.
\end{proposition}

\begin{proof}
First of all, we show that $H_3$ is a head. Since $i$ is the maximal vertex, $i \in H_3$, then for any $j \in H_3$, it is either in $H_1$ or $H_2$, which means that $j$ lies in the same district as $i$ in either $\G_{\an(H_1)}$ or $\G_{\an(H_2)}$. This graph is a subgraph of $\G_{\an(H_3)}$, hence $j$ lies in the same district as $i$ in $\G_{\an(H_3)}$.

Now let $K_3 = H_3 \setminus H$. It does not contain $i$ and also it is not empty since $H_3 \neq H$, thus it is a valid marginalization set for $H_3$. Suppose also that $H_3 \to^{K_3} H'$. We know that $\an(H_3) \supseteq \an(H_1) \supseteq \an(H)$, so $\an(H_3) \setminus K_3 = B = \an(H')\supseteq \an(H)$. Now suppose $H' \neq H$, this means that $\dis_{\an(H')}(i)$ contains some vertices that are not in $\dis_{\an(H)}(i)$. Among those vertices, there are the strict siblings of $\dis_{\an(H)}(i)$; there must exist such vertices because $H'$ is bidirected-connected. Now select one of these siblings, say, $j$. WLOG, $j$ belongs to $\an(H_1)$. Then to go from $H_1$ to $H$, $j$ must be marginalized, thus $j \in H_1$. If $j$ is not in $H_3$, then this means that there are some descendants of $j$ that are in $H_2$, but then we cannot go from $H_2$ to $H$ since $j$ stays in the districts, there is some extra vertex in the barren subset of the district of $i$, other than $H$. Hence $j \in H_3$. But then as $j$ is in $\an(H')$, $j$ does not lie in $K_3$, so $j \in H$, which is a contradiction to the definition of $j$.
\end{proof}

\begin{lemma}\label{prop: what is minimal marginalization set}
For a MAG $\G$, if head $H$ is a parent head of $H'$ in the power DAG, 
then the minimal marginalization set $K$ is $H \cap \ceil_\G(\ham_\G(H'))$.
\end{lemma}

\begin{proposition}\label{prop:maximal independence from maximal parent head}
Consider the four heads $H, H_k$, $k = 1,2,3$ with the same setting in Proposition \ref{prop: maximal parent in power dag}, then the minimal marginalization set for $H_3$ (to $H$) is the union of the minimal marginalization sets of $H_1, H_2$ (to $H$).
\end{proposition}

\begin{proof}
Lemma \ref{prop: what is minimal marginalization set} shows that the minimal marginalization set is the intersection between the parent head and $\ceil(\ham(H))$, thus since $\an_\G(H_3)$ is just an union of $\an_\G(H_1)$ and $\an_\G(H_2)$, it will not introduce extra vertices in $\ceil(\ham(H))$. 
\end{proof}

In Lemma \ref{prop: what is minimal marginalization set},  by construction, $H_3 \geq H_1, H_2$. One can show that if $H \geq H'$ then $H \cup T \supseteq H' \cup T'$, therefore by Lemma \ref{prop: what is minimal marginalization set}, the maximal independence from the $H_3$ always implies the maximal independences from $H_1$ and $H_2$.

\begin{proposition}\label{prop: maximal indep from minimal independencs}
Suppose in a power DAG for $i$, $H$ is a parent head of $H'$, then the independence associated with the minimal marginalization set $K^m$, $$i \indep H \cup T \setminus H' \cup T' \cup K^m  \mid H' \cup T' \setminus \{i\},$$ implies all the independences associated with any $H \to^{K} H'$ and any other marginalization set $K$.
\end{proposition}
\begin{proof}
The conditioning set is fixed, and the more elements $K$ has, the fewer $i$ is independent from.
\end{proof}

We call this the \emph{maximal independence} associated with $H \to^{K} H'$




\section{Proofs and supplementary materials for section \ref{power DAGs}}

\subsection{Complete power DAGs}\label{sec:complete power DAG}

\begin{definition}\label{def:power DAG}
Consider a MAG $\G$ with a topological ordering.  Given a set $S \subseteq \mathcal{V}$ 
we say that $s \in S$ is a \emph{marginalization vertex} if it is in $\barren_\G(S)$ and 
is not maximal in $S$.

Define the \emph{complete power DAG} $\mathfrak{I}(\G)$ as a graph with vertices $\mathcal{H}(\G)$.  An edge is added from $H \to H'$ if there is a marginalization 
vertex $k \in H$ such that $H \to^k H'$.  In this case 
we call $H$ a \emph{parent head} of $H'$. 
There is a unique component for each 
vertex $i$, which we denote $\mathfrak{I}_i(\G)$.

\end{definition}

Next we present a few results that justify the nomenclature of power DAGs.

\begin{lemma}\label{lemma: power dags are simple}
For a MAG $\G$ and any $i$, there is at most one edge between any two heads in $\mathfrak{I}_i (\G)$.
\end{lemma}

\begin{proof}
This is a direct consequence of Lemma \ref{the minimal set of vertices from one head to another head}, which proves that for a MAG $\G$ and two heads $H,H'$ whose maximal vertex is $i$, if $H \to^{K} H'$ and $H \to^{L} H'$ then $H \to^{K \cap L} H'$.
\end{proof}

\begin{lemma}\label{lemma: partial orders admits ancestral relations}
For two heads $H, H' \leq i$, we have $H > H'$ if and only if $H$ is an ancestor of $H'$ in  $\mathfrak{I}_i (\G)$
\end{lemma}
\begin{proof}
This can be proved by marginalizing vertices in $\an(H) \setminus \an(H)$ step by step. 
\end{proof}

A useful fact is that for any $i$ and $\mathfrak{I}^\G_i$, there exists a (maximal) head $H$ such that $H \geq H'$ for any $H' \leq i$ and this head is the barren subset of the district of $i$ in $\G_{[i]}$.

\begin{lemma}
For a MAG $\G$ and any $i$, $\mathfrak{I}_i (\G)$ is a DAG
\end{lemma}
\begin{proof}
This is a direct consequence of Lemma \ref{lemma: power dags are simple} and \ref{lemma: partial orders admits ancestral relations}.
\end{proof}

Here we prove that the list of independences associated with edges in $\mathfrak{I}^\G_i$ for every $i$, combined with the independences $i \indep [i-1] \setminus \mb_\G(i,[i]) \mid \mb_\G(i,[i])$, are sufficient to deduce the ordered local Markov property. 

\begin{definition}\label{def: head and tail Markov property}
 For a MAG $\G$ and any $i$, we associate $\mathfrak{I}^\G_i$ with a collection of independences  $\mathbb{L}^\G_i$ that contains:
 \begin{itemize}
     \item[($a$)] $i \indep [i-1] \setminus \mb_\G(i,[i]) \mid \mb_\G(i,[i])$, and
     \item[($b$)] for every head $H$ (except $\{i\}$) whose maximal element is $i$:
\begin{align*}
    &i \indep (H \cup T) \setminus (H' \cup T' \cup k ) \mid H' \cup T' \setminus \{i\}  && \text{for } k \in H \setminus \{i\},
\end{align*}
 where $H \to^{k} H'$, and $T=\tail_\G(H)$ and $T'=\tail_\G(H')$.
 \end{itemize}
\end{definition}

\begin{theorem}\label{thm:head and tail Markov property}
For a MAG $\G$, the collection $\mathbb{L}^\G = \bigcup_i \mathbb{L}^\G_i$ is equivalent to the list of independences implied by the ordered local Markov property for $\G$.
\end{theorem}

\begin{proof}
($\Longleftarrow$): notice that $H'\cup T' \setminus \{i\}$ is the Markov blanket of $i$ in the ancestral set $\an(H \cup T \setminus K)$, thus it follows from the ordered local Markov property by marginalizing irrelevant vertices.

($\Longrightarrow$): let $\mathcal{A}^{x}$  denote the set of all ancestral sets whose maximal elements are $x$. Further let $\mathcal{A}_{r}^{x} = \{A \in  \mathcal{A}^{x} : \abs{A} = r\}$ (note $1 \leq r \leq x$). We will proceed by induction on $x$ from $x = 1$ to $x=n$. To show for every $A \in \mathcal{A}^{x}$, the corresponding independence implied by the ordered local Markov property is implied by $\mathbb{L}_{i}$, we will apply a further induction on $\mathcal{A}_{r}^{x}$ from $r=x$ to $r=1$.

For $A \in \mathcal{A}^x$, $A \subseteq [x]$. The base case $x=1$ is trivial. Suppose the induction hypothesis is true, i.e.\  for any ancestral set $A \in \mathcal{A}^{x}, x \leq i-1$ the corresponding conditional independence implied by the ordered local Markov property is in $\I_{u_{\G}}$. 
 
Now consider $\mathcal{A}^i$, we will then apply induction on $\mathcal{A}_{r}^i$. For the base case $r = i$, i.e.\ $A = [i]$, the independence is in $\mathbb{L}_i$, that is, (a). Now suppose the induction hypothesis is true, that is: for any $A \in \mathcal{A}_{r}^i, s+1 \leq r \leq i$, the corresponding conditional independence from the ordered local Markov property is implied by $\mathbb{L}_i$ and we can use the ordered local Markov property on $\G_{[i-1]}$. 

Now consider any $A \in \mathcal{A}_s^i$. There is at least one vertex $v$ in $[i] \setminus A$ such that $v$ is parentless in $[i] \setminus A$, in other words, $A \cup \{v\} = A'$ is ancestral and $A' \in \mathcal{A}^i_{s+1}$.  Hence by the second induction hypothesis we have 
\begin{align}\label{CI1}
i \indep A' \setminus (\mb(i,A') \cup \{i\}) \mid \mb(i,A').    
\end{align}
If $\mb(i,A) = \mb(i,A')$, then we can get required conditional independence by marginalizing $v \in A' \setminus (\mb(i,A') \cup \{i\})$. 

So now assume $\{v\} \cup \mb(i,A) \subseteq \mb(i,A')$ ($v,i$ are in the same district). Let $A'' = A \cup \{v\} \setminus \{i\} = A' \setminus \{i\}$. We know $A''$ is ancestral and is in $\mathcal{A}^{i-1}$, thus by the firt induction hypothesis, we can apply the ordered local Markov property to any vertex in $\barren(A'')$ (changing the topological order on $[i-1]$). Moreover, $v \in \barren(A'')$ because if $v$ has any child in $A$ then $A$ is not ancestral. So we have:
\begin{align*}
v \indep A'' \setminus (\mb(v,A'') \cup \{v\}) \mid \mb(v,A'').   
\end{align*}
Next notice that by assumption, $\mb(v,A') \cup\{v\} = \mb(i,A') \cup \{i\}$ ($i,v$ are in the same district in $A'$). As $A'' \subset A'$, we have $\mb(v,A'') \subset \mb(v,A') = \mb(i,A') \cup \{i\} \setminus \{v\}$. Then because $\mb(v,A'') \subset \mb(i,A') \cup \{i\} \setminus \{v\}$ and $v \in \mb(i,A')$, by semi-graphoids, moving $\mb(i,A') \setminus (\mb(v,A'') \cup \{v\})$ (which is a subset of $A''$) to the conditioning set, we have: 
 \begin{align}\label{CI3}
 v \indep A'' \setminus \mb(i,A')  \mid \mb(i,A') \setminus \{v\}.  
 \end{align}
Now (\ref{CI1}) is equivalent to:
\begin{align}\label{CI4}
i \indep A'' \setminus \mb(i,A') \mid \mb(i,A'). 
\end{align}
Thus by semi-graphoids, we can deduce the following from (\ref{CI3}) and (\ref{CI4}):
\begin{align*}
\{i,v\} \indep A'' \setminus \mb(i,A')  \mid \mb(i,A') \setminus \{v\}.
\end{align*}
The next step is to marginalize $v$ so that:
\begin{align*}
\{i\} \indep A'' \setminus \mb(i,A')  \mid \mb(i,A') \setminus \{v\}.
\end{align*}
which is equivalent to:
\begin{align}\label{CI7}
\{i\} \indep A \setminus ((\mb(i,A') \setminus \{v\}) \cup \{i\} ) \mid \mb(i,A') \setminus \{v\}.    
\end{align}
Note that $(\mb(i,A') \setminus \{v\}) \subseteq A$. To extract vertices in $(\mb(i,A') \setminus \{v\})  \setminus \mb(i,A)$, the key point is to notice that $\{v,i\} \subseteq \barren(\dis_{A'}(i)) = H,$ where $H$ is a head that contains the maximal vertex $i$. Thus $\{v\}$ can be a marginalizing set $K$ for $H$ with the tail $T$, and $H \cup T = \mb(i,A')$. Now by marginalizing $\{v\}$ we reach $B = A' \setminus \{v\} = A$, thus the following conditional independence is in $\mathbb{L}_i$: 
\begin{align}\label{CI8}
\{i\} \indep \mb(i,A') \setminus ((\mb(i,A) \cup \{v\}) \cup \{i\} ) \mid \mb(i,A).    
\end{align}

From (\ref{CI7}) and (\ref{CI8}), we can deduce that:
\begin{align*}
i \indep A \setminus (\mb(i,A) \cup \{i\}) \mid \mb(i,A).    
\end{align*}
\end{proof}

\subsection{Missing proofs}

\begin{proof}[Proof of Proposition \ref{prop:equiv between refined and local}]
In Theorem \ref{thm:head and tail Markov property}, we prove that the list of independences associated with the complete power DAGs, denoted as $\mathbb{L}^\G$, is equivalent to the ordered local Markov property. Because $\widetilde{\mathbb{L}}^\G \subseteq \mathbb{L}^\G$, by Theorem \ref{thm:head and tail Markov property}, it is sufficient to prove that $\widetilde{\mathbb{L}}^\G$ implies $\mathbb{L}^\G$.

We proceed by three inductions. The first induction is on the topological ordering of vertices and it is sufficient to show that given $\mathbb{L}_{k}^\G$ is true for $1 \leq k \leq i-1$, $\mathbb{L}_{i}^\G$ is implied by $\widetilde{\mathbb{L}}^\G_i$. The base case is trivial.

The second induction is on the topological ordering on heads.  We start from the maximal head in $\mathfrak{I}_i (\G)$ and proceed downwards to show that the independences in $\mathbb{L}_i^\G$ associated with each head and its parent heads are true. Again, the base case is trivial. Now suppose that for a head $H'$, every independence associated with any heads preceding $H'$ is true, in particular, independences in $\mathbb{L}^\G$ associated with any parent head of $H'$ hold. We need to show that independences associated with edges $H \to^{k} H'$ for any parent head $H$ of $H'$ hold. 

Now consider the topological ordering on the vertices in $\ceil_\G(\ham_\G(H'))$, which is made of all the marginalization vertices and we proceed to the third induction on this ordering. The base case is for those parent heads of $H'$ with marginalization vertex $k = \min\ceil_\G(\ham_\G(H'))$.  Clearly the independence from the maximal parent head is in $\widetilde{\mathbb{L}}^\G_i$, and this implies all independences from other parent heads $H$ such that $H \to^{k} H'$. 

For the inductive step: consider some parent head $H_{j}$ of $H'$ with marginalization vertex $l \in \ceil_\G(\ham_\G(H'))$ and $l > k$. Take some parent head $H_{1}$ of $H'$  with marginalization vertex $k$, then clearly $k \notin H_{j}$ and $l \notin H_{1}$. Let $$
H^m = \barren_\G(H_{j} \cup H_{1}), 
$$ then by Proposition \ref{prop:maximal independence from maximal parent head}, $H^m \to^{kl} H'$. Therefore, we have $H^m \to^{l} H^{m'}$ for some parent head $H^{m'}$ of $H'$  and also  we have $H^{m'} \to^{k} H'$. 

By the hypothesis of the second induction, we have $$
i \indep H^m  T^m \setminus H^{m'} T^m_{t} l \mid H^{m'}T^{m'} \setminus \{i\}.
$$
By the hypothesis of the third induction, we have $$
i \indep H^{m'}  T^{m'} \setminus H' T' k \mid H'T' \setminus \{i\}.
$$ Hence, by the contraction semi-graphoid, we have $$
i \indep H^m T^m \setminus H' T' kl \mid H'T' \setminus \{i\}.
$$ Now we have $k \notin H_{j}$ and also, by construction, $H^m  T^m \supseteq H_{j} T_{j}$. Therefore, we can marginalize irrelevant vertices to get 
$$
i \indep H_{j} T_{j} \setminus H'  T' l \mid H'T' \setminus \{i\}.
$$

Next we show it contains fewer and smaller independences than the reduced ordered local Markov property. For each maximal ancestral set 
$A$ and its maximal vertex $i$, there is a corresponding head by taking barren of the district of $i$ and this relation is one-to-one. Hence the number of independences in the refined Markov property is not more than the number of independences in the refined Markov property. Further every independence in the refined Markov property can be deduced by an independence from the reduced ordered local Markov property by simply marginalizing some vertices, hence the statement is true.
\end{proof}

\begin{proof}[Proof for Proposition \ref{prop: algorithm 1 give the refined power dag}]
At line \ref{algl:poheads}, the algorithm proceeds by the partial ordering on heads. Thus for any head $H$, once we arrive at line \ref{algl:poheads}, there will be no edge added that is into it. So it is sufficient to show that for any head $H$, once we arrive it at line \ref{algl:poheads}, there is only one edge into it, which corresponds to the edge in $\widetilde{\mathfrak{I}}^\G_i$. We prove this by induction on the partial ordering on heads. Clearly the algorithm does not add any edge into the maximal head. The base case is for the children of the maximal head and this clearly holds.

For the inductive step, consider all the parent heads of a head $H'$, which exceed $H'$ in the partial ordering and consider $H^*$ and $k$ defined in Definition \ref{def: refined power DAG}. Suppose $H^* \to^k  H'$ does not appear and instead we have $H \to^{k'} H'$ for some other parent head of $H'$. 

Firstly if $k' \neq k$ then by Lemma \ref{characterizing marginalization vertex} and definition of $k$, we must have $k' > k$. Now it is clear that $k \in \ceil_\G(\ham_\G(H))$, so $k$ must be marginalized at some point to reach $H$, therefore $k \leq M(H)$ and then line \ref{algl:smaller} prevents marginalizing of $k'$.

Then we have $k' = k$, so $H \to^{k} H'$. Then by definition of $H^*$ we have $H^* > H$, thus from the maximal head, $H^*$ can be reached within fewer steps than $H$, and line \ref{algl:shorter} prevents adding $H \to^{k} H'$.
\end{proof}

\begin{proof}[Proof of Proposition \ref{prop: complexity of algorithm 1}]

For Algorithm \ref{algo1}: there are at most $\binom{n}{k}$ heads at line \ref{line: for i in [n]} and \ref{algl:poheads}. Then line \ref{algl:smaller} is of order $k$. At line \ref{line: check heads}, it takes $O(n+e)$ to check which new head is reached. Remaining lines all takes constant time or does not contain any loop.

In Algorithm \ref{alg:algo3}, there are at most $\binom{n}{k}$ heads. Then we need to compute the tail for each head and add the corresponding semi-elementary imset. The latter is of constant time and the tail of each head can be computed simultaneously when we visit the head from its parent head, therefore is also of $O(n+e)$ time.
\end{proof}

\begin{algorithm}
    \SetAlgoLined
\SetKw{KwDef}{define}
\KwIn{A simple MAG $\mathcal{G}([n],\mathcal{E})$ ($\mathcal{E}$ stored as adjacencies)}
\KwResult{The refined power DAGs $\widetilde{\mathfrak{I}}^\G_i$ for each $i$}
\For{$i \in [n]$ }{
  
  Let smaller siblings of $i$ be $j_1,\ldots,j_k$\\
  Set $\widetilde{\mathfrak{I}}_i = \{\{j_k,i\} \to \ldots \to \{j_1,i\} \to \{i\}\}$\\
}

\Return{$( \widetilde{\mathfrak{I}}^\G_1, \widetilde{\mathfrak{I}}^\G_2, \cdots, \widetilde{\mathfrak{I}}^\G_n)$}
\caption{Obtain the refined power DAGs for a simple MAG}
\label{alg:algo2}
\end{algorithm}

\subsection{An example for redundant independences in the refined power DAGs}\label{sec:redundant indep for refined power DAGs}

\begin{example}\label{exp: reordering is not enough}
Consider Figure \ref{fig:reordering is not enough}(i). The independences associated with the component of its refined power DAG for $6$ under a numerical topological ordering are:
\begin{align*}
    6 &\indep 3 \mid 1,2,4 &
    6 &\indep 2 \mid 1,3,5 & 
    6 &\indep 2 \mid 1.
\end{align*}
Whatever the topological order over the other vertices, there is always a third independence that is redundant. For example, with this numerical ordering, one can deduce $6 \indep 2 \mid 1$ from $6 \indep 2 \mid 1,3,5$ and $2 \indep 3,5 \mid 1$.
If one adds an edge between $2$ and $3$ as in Figure \ref{fig:reordering is not enough}(ii), then only semi-graphoids are insufficient to deduce $6 \indep 2 \mid 1$ from only first two independences and independences associated with earlier vertices; hence, in this case, we need the third independence.

This example suggests that a minimal list of independences required to define the model sometimes depends on structures not local to the vertex being considered. We are investigating how to obtain such a list, but conjecture that it may be computationally difficult to 
do so in general.

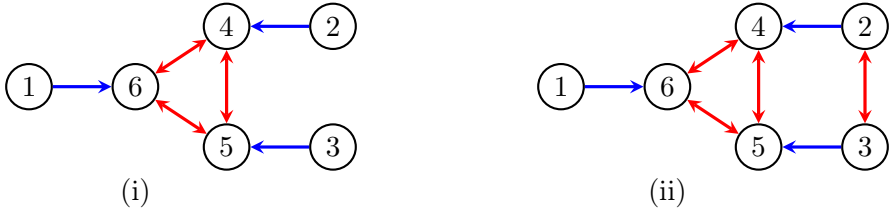
\begin{figure}
\centering
  \begin{tikzpicture}
  [rv/.style={circle, draw, thick, minimum size=6mm, inner sep=0.8mm}, node distance=14mm, >=stealth]
  \pgfsetarrows{latex-latex}
\begin{scope}
  \node[rv]  (6)             {$6$};
  \node[rv,right of=6, xshift=-0.2cm, yshift=-0.8cm] (5)  {$5$};
  \node[rv,right of=6, xshift=-0.2cm, yshift=0.8cm] (4)  {$4$};
  \node[rv, right of=5](3) {$3$};
  \node[rv, right of=4](2) {$2$};
  \node[rv, left of=6](1) {$1$};
  \draw[beg] (6) -- (4);
  \draw[beg] (5) -- (4);
  \draw[beg] (6) -- (5);
  \draw[deg] (3) -- (5);
  \draw[deg] (2) -- (4);
  \draw[deg] (1) -- (6);
  \node[below of=6] {(i)};
  \end{scope}
\begin{scope}[xshift = 7cm]
  \node[rv]  (6)          {$6$};
  \node[rv,right of=6, xshift=-0.2cm, yshift=-0.8cm] (5) {$5$};
  \node[rv,right of=6, xshift=-0.2cm, yshift=0.8cm] (4)  {$4$};
  \node[rv, right of=5](3) {$3$};
  \node[rv, right of=4](2) {$2$};
  \node[rv, left of=6](1) {$1$};
  \draw[beg] (2) -- (3);
  \draw[beg] (6) -- (4);
  \draw[beg] (5) -- (4);
  \draw[beg] (6) -- (5);
  \draw[deg] (3) -- (5);
  \draw[deg] (2) -- (4);
  \draw[deg] (1) -- (6);
  \node[below of=6] {(ii)};
  \end{scope}

\end{tikzpicture}
 \caption{Example where refined Markov property is still redundant}
 \label{fig:reordering is not enough}
\end{figure}

\end{example}
\subsection{Proof of Theorem \ref{thm:imset as sum of elementary imset for general MAGs} in Section \ref{sec:decomp of u_G}} \label{full proof of decomp of u_G}

\begin{proof}[Proof of Theorem \ref{thm:imset as sum of elementary imset for general MAGs}]
Let $T = \tail_\G(H)$. Note that $u_\G$ from Theorem \ref{thm: standard imset} can be rewritten as: $$\sum_{i=1}^{n}\biggl\{ \delta_{[i]}-\delta_{[i-1]}-\sum_{\substack{H \in \mathcal{H}(\G) \\ i \geq H}} \sum_{W \subseteq H} (-1)^{\abs{H \setminus W}} \delta_{W \cup T}\biggr\}.$$

By induction, if we restrict to the final vertex of a topological ordering, say $n$, then all we need to prove is that:
\begin{align*}
\lefteqn{\delta_{[n]} - \delta_{ [n-1]} - \sum_{\substack{H \in \mathcal{H}(\G) \\ n \geq H}} \sum_{W \subseteq H} (-1)^{|H \setminus W|} \delta_{W \cup T} }\\
&= u_{\langle n,[n-1] \setminus \mb(n,[n]) \mid \mb(n,[n]) \rangle } + \sum_{\substack{H \in \mathcal{H}(\G) \setminus \{n\} \\ n \geq H}}  \sum_{ \substack{ \emptyset \subset K \subseteq H \setminus \{i\}:  \\ H \to^{K} H'}} (-1)^{\abs{K}+1} 
    u_{\langle i,H  T \setminus H' T' K  \mid H'T' \setminus i \rangle }.
\end{align*}
Note that $u_{\langle n,[n-1] \setminus \mb(n,[n]) \mid \mb(n,[n]) \rangle} = \delta_{[n]} - \delta_{[n-1]} - \delta_{\{n\} \cup \mb(n)} + \delta_{\mb(n)}$, so we can reduce the equivalence to
\begin{align}
\lefteqn{-\sum_{\substack{H \in \mathcal{H}(\G) \\ n \geq H}} \sum_{W \subseteq H} (-1)^{|H \setminus W|} \delta_{W \cup T}} \label{LHS}\\ 
&= \delta_{\mb(n,[n])} - \delta_{\{n\} \cup \mb(n,[n])} \nonumber \\
&\phantomrel{=}+ \sum_{\substack{H \in \mathcal{H}(\G) \setminus \{n\} \\ n \geq H}} \sum_{ \substack{ \emptyset \subset K \subseteq H \setminus \{i\}:  \\ H \to^{K} H'}} (-1)^{\abs{K}+1} 
(\delta_{H T \setminus K} - \delta_{(HT \setminus Kn)} - \delta_{H'T' } + \delta_{H'T' \setminus n}). \label{RHS}
\end{align}
We can rewrite  (\ref{LHS}) as: 
\begin{eqnarray}
\lefteqn{-\sum_{\substack{H \in \mathcal{H}(\G) \\ n \geq H}}\sum_{W \subseteq H} (-1)^{\abs{H \setminus W}} \delta_{W \cup T} }  \nonumber\\
& &= -\sum_{\substack{H \in \mathcal{H}(\G) \\ n \geq H}}\sum_{K \subseteq H} (-1)^{\abs{K}} \delta_{H T \setminus K}\nonumber\\
& & =\sum_{\substack{H \in \mathcal{H}(\G) \\ n \geq H}}\sum_{K \subseteq H \setminus \{n\}} (-1)^{\abs{K}+1} (\delta_{H T \setminus K} - \delta_{H T \setminus Kn}). \label{LHS2}  
\end{eqnarray}
The repeating part is only $-\delta_{H T}+\delta_{H T \setminus n}$ for each $H$ (when $K= \emptyset$). For each $H$, the two terms do not appear in the summation for $H$ (as we rule out the case when $K=\emptyset$ in Theorem \ref{thm:imset as sum of elementary imset for general MAGs}) but they appear when other heads marginalized to $H$ and the two terms are multiplied by some constants. 

Our objective is to show that (\ref{LHS2}) and (\ref{RHS}) are equivalent, for which it is sufficient to prove that for every head $H$, the coefficient of $-\delta_{H T}+\delta_{H T\setminus i}$ is 1 in (\ref{RHS}). For the largest head, which is $\barren(\dis_{[n]}(n))$, it is clearly true because there is no head that is `larger'  than it and for this head $-\delta_{H T}+\delta_{H T \setminus i}$ is simply $\delta_{\mb(n,[n])} - \delta_{\{n\} \cup \mb(n,[n])}$, 
which never appears in the summation.

For any other head $H$ with maximal vertex $n$, we need to prove: $$\sum_{K, H': H' \to^{K} H}(-1)^{\abs{K}+1} = 1.$$ 
Note that the summation is over both $H'$ and $K$ because for a pair of heads there might exist different marginalization sets that lead one to another.

Now by Lemma \ref{construct K given minimal K}, any head $H'$ that can reach $H$ with multiple marginalizing sets and minimal marginalizing set $K'$ contributes $\sum_{B: B \subseteq (H' \setminus (H \cup K'))} (-1)^{\abs{K'}+\abs{B}+1} = 0$ to the coefficient of $-\delta_{H \cup T}+\delta_{(H \setminus \{n\} \cup T)}$. 

Thus it is sufficient to consider any head $H'$ that can reach $H$ with the only (minimal) marginalizing set $K = H' \setminus H$. By Lemma \ref{characterization of minimal set}, we find all these heads and in total they contribute $$\sum_{K: \emptyset \subset K \subseteq \ceil(\sib(\dis_{\an(H)}(i))\setminus \dis_{\an(H)}(i))} (-1)^{\abs{K}+1} = 1$$ to the coefficient of $-\delta_{H \cup T}+\delta_{(H \setminus \{n\} \cup T)}$.
\end{proof}


\begin{proof}[Proof for Corollary \ref{cor: submodel}]
By Theorem \ref{thm:imset as sum of elementary imset for general MAGs}, the `standard' imset $u_{\G}$ can be expressed as one combinatorial imset $u_p$ subtracted by another combinatorial imset $u_n$ where $u_p$ are obtained by adding all semi-elementary imsets with positive coefficients (where $|K| $ is odd), and similarly for $u_n$ but with negative coefficients (where $|K| $ is even). 

Let $u_p^1$ be sum of semi-elementary imsets corresponding to marginalize only one vertex, which are all in $u_p$. By Theorem \ref{thm:head and tail Markov property}, $I_{\G} \subseteq \I_{u_p^1}$. The faithfulness result in \citet{richardson2002} shows that for every MAG $\G$ there exists a distribution
$P$ such that $\I_P = \I_{\G}$  and Theorem 5.2 in \citet{studeny2006probabilistic} implies the existence of
a structural imset $u$ such that $I_u = I_P$. Now every independence in the list of independences that we used to construct $u_p^1$ is in $\I_{u}$ by Theorem \ref{thm:head and tail Markov property}, so by Lemma 
6.1 in \citet{studeny2006probabilistic}, we have $I_{u_p^1} \subseteq \I_u = \I_{\G}$. Hence $\I_{u_p^1} = \I_{\G}$. Now the remaining semi-elementary imsets that we use to construct $u_p$ are those corresponding to marginalize odd number of vertices (more than one), and the independences they correspond to are all in $\I_{\G} = \I_{u_p^1}$. Therefore $\I_{u_p^1} = \I_{u_p}$. 

Suppose $u_{\G}$ is structural and take any distribution $P$ that is not Markov to it. \citet{studeny2006probabilistic} shows that a distribution is Markov to a structural imset if and only if the inner product between the imset and entropy vector of the distribution is zero. Moreover this inner product is non-negative. Therefore, for $P$ and $u_{\G}$, this inner product is positive. Further, as $u_{\G}$ is also $u_p-u_n$ where $u_p$ and $u_n$ are both combinatorial, this means that the inner product between $P$'s entropy vector and $u_p$ are also positive, and hence $P$ is not Markov to $u_p$. Thus $P$ is not Markov to $I_{\G}$. As a result, $\I_{u_{\G}} \subseteq \I_{\G}$.
\end{proof}
\begin{definition}
    Let $\G$ be a MAG containing a district $D$.  
    Then by $\G^{|D}$ denote the induced subgraph over 
    $D \cup \pa_\G(D)$, but where all vertices within
    $\pa_\G(D) \setminus D$ have been joined by bidirected
    edges.
\end{definition}

\begin{corollary}{\label{cor:tian_factorization}}
Let $\G$ be a MAG. The `standard' imset $u_\G$ is perfectly Markovian
with respect to $\G$ if the `standard' imsets 
$u_{\G^{|D}}$ are perfectly Markovian with respect
to $\G^{|D}$.
\end{corollary}

\begin{proof}
Take the expression for the standard imset given in 
Then note that, for a district $D$, each summand other than the very first term only contains independences between an element of the district, and other elements of the district and its parents, and the collection of all first terms gives a DAG model.  Furthermore, the subimset defined by the second inner summation represents the model in which all the vertices in $\pa_\G(D) \setminus D$ have been joined by an edge.  Since we know that any MAG can be defined by independences only within the district (plus the initial independence in the expression above), this shows that a `standard' imset for a MAG $\G$ will be perfectly Markovian with respect to the graph if and only if the associated standard imsets for each district and its parents are perfectly Markovian with respect to the induced subgraph obtained after filling in any missing edges between parents.
\end{proof}

\begin{example}
Consider Figure \ref{DAG on heads}.  The multiple labels on the edges in (ii) means that there are different sets of vertices that can be marginalized to lead from one head to another.

The edge $\{4,5,6\} \to^{45} \{1,3,6\}$ indicates that when we compute the semi-elementary imset for $\{4,5,6\}$ by marginalizing vertices in $\{4,5\}$, we only reach  the head $\{1,3,6\}$ once. Hence the head $\{4,5,6\}$ contributes $(-1)^{2+1}=-1$ to the coefficient of $-\delta_{1236}+\delta_{123}$ for the head $\{1,3,6\}$. Similarly the edge from $\{3,5,6\}$ to $\{6\}$ means that we may reach to the head $\{6\}$ from the head $\{3,5,6\}$ by marginalizing either $\{3\}$ or $\{3,5\}$. So the head $\{3,5,6\}$ contributes $(-1)^{1+1}+(-1)^{2+1}=0$ to the coefficient of $-\delta_{26}+\delta_{2}$ for the head $\{6\}$.

\begin{figure}
\centering
  \begin{tikzpicture}
  [rv/.style={circle, draw, thick, minimum size=6mm, inner sep=0.8mm}, node distance=14mm, >=stealth]
  \pgfsetarrows{latex-latex}
\begin{scope}
  \node[rv]  (1)            {$1$};
  \node[rv, below of=1, xshift=-6mm,yshift=2mm] (2) {$2$};
  \node[rv, below of=1, xshift=6mm,yshift=2mm] (3) {$3$};
  \node[rv, above of=1,yshift=-2mm] (4) {$4$};
  \node[rv, below of=2, xshift=-10mm,yshift=5mm](5){$5$};
  \node[rv, below of=3, xshift=10mm,yshift=5mm](6){$6$};
  \draw[->, very thick, blue] (1) to [bend right=30] (5);
  \draw[->, very thick, blue] (2)  to [bend right=30] (6);
  \draw[->, very thick, blue] (3)  to [bend right=30] (4);
  \draw[beg] (1) -- (3);
  \draw[beg] (2) -- (3);
  \draw[beg] (1) -- (2);
  \draw[beg] (2) -- (5);
  \draw[beg] (3) -- (6);
  \draw[beg] (1) -- (4);
  \node[below of=1, yshift=-2.1cm] {(i)};
  \end{scope}
\begin{scope}[xshift = 7cm,yshift=3cm]
  \node[]  (1)            {$4,5,6$};
  \node[, below of=1,yshift=-0.5cm] (2) {$3,5,6$};
  \node[, right of=2,xshift=0.5cm] (3) {$1,4,6$};
  \node[, below of=2,yshift=-0.5cm] (4) {$1,3,6$};
  \node[, below of=3,yshift=-0.5cm] (5) {$3,6$};
  \node[, below of=4,yshift=-0.5cm] (6) {$6$};
  \path[->,draw,very thick,blue]
    (1) edge node [xshift=1.5mm]{\scriptsize{$4$}} (2)
    (1) edge node [xshift=1.5mm,yshift=1.5mm]{\scriptsize{$5$}} (3)
    (2) edge node [xshift=1.5mm,yshift=1mm]{\scriptsize{$5$}} (4)
    (3) edge node [xshift=-1.5mm,yshift=1.5mm]{\scriptsize{$4$}} (4)
    (3) edge node [xshift=2mm]{\scriptsize{$\substack{1 \\ 1,4}$}} (5)
    (4) edge node [yshift=2mm]{\scriptsize{$1$}} (5)
    (4) edge node [xshift=2mm,yshift=1mm]{\scriptsize{$\substack{3 \\ 1,3}$}} (6)
    (5) edge node [xshift=1.5mm,yshift=-1.5mm]{\scriptsize{$3$}} (6)
    (2) edge [bend right=30] node [xshift=-3mm]{\scriptsize{$\substack{3 \\ 3,5}$}} (6)
    (1) edge [bend right=30] node [xshift=-3mm]{\scriptsize{$4,5$}} (4);

  \node[below of=6, xshift=7.5mm, yshift=0.5cm] {(ii)};
  \end{scope}

\end{tikzpicture}
 \caption{(i) A MAG $\G$ (ii) A complete power DAG on the heads of $\G$ with maximal vertex 6, under the numerical topological ordering.}
 \label{DAG on heads}
\end{figure}
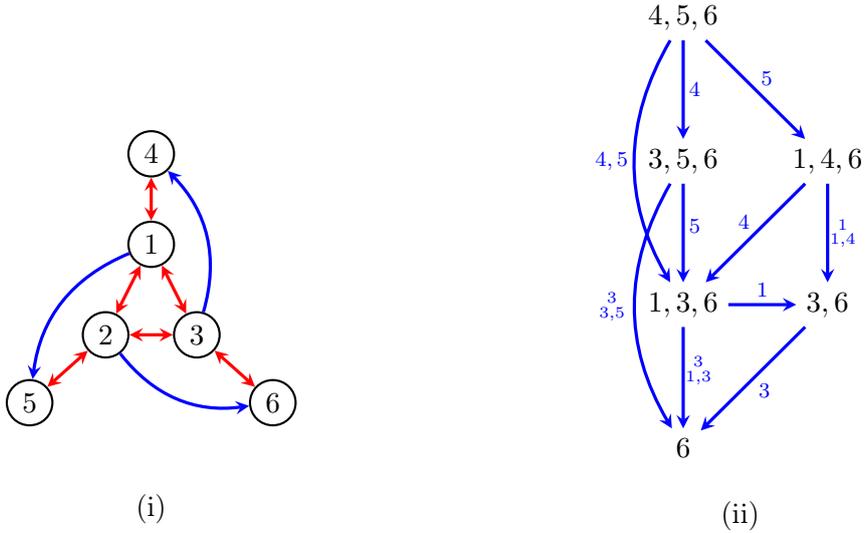

\end{example}

\section{Scoring criteria}\label{sec:scoring_app}
The BIC \citep{schwarz1978estimating} is a consistent (defined below) scoring criterion. For discrete models, \cite{Evans2014,evans2010maximum} provide procedures for fitting ADMGs by maximum likelihood, and thus we can compute the BIC.

Let $\ell$ be the log-likelihood and $q^{\theta}$ denote the family of distributions that are Markov to a fitted graph $\G$, with parameter $\theta$; denote the maximum likelihood estimator for $\theta$ by $\hat{\theta}$. Let $d=\abs{\mathcal{S}(\G)}$ be the dimension of the binary model \citep{Evans2014}, and $N$ and $N(x_{\mathcal{V}})$ be the number of samples and the number of samples such that $X_{\mathcal{V}}=x_{\mathcal{V}}$, respectively. Then the BIC for the model defined by $\G$ is: 
$$
-2 \hat{\ell}+ d \log N,
$$ 
where $\hat{\ell}=\sum_{x_{\mathcal{V}}} N(x_{\mathcal{V}}) \log q_{\hat{\theta}}(x_{\mathcal{V}}).$


However this score is not very suitable for a greedy learning algorithm as we need to re-fit the whole graph when we consider new models. Thus in this chapter we aim to study a scoring criterion proposed by \citet{andrews:phd} that is decomposable with respect to the parametrizing set $\mathcal{S}(\G)$. The motivation is that equivalent MAGs have the same BIC and the parametrizing set, and every time we move between Markov equivalence classes of MAGs, we simply change the 0-1 vector of the characteristic imset, and it would save a lot computations if for those sets remaining in the parametrizing set, we do not need to compute their corresponding scores again.

The proposing scoring criterion is $-2N \langle c_{\G}, \hat{\sf I} \rangle+ d\log N$, where $\hat{{\sf I}}$ is the plug-in estimate of \emph{interaction information}, defined below. Thus we are approximating $\hat{\ell}$ with $\langle c_{\G}, \hat{{\sf I}} \rangle$.

\subsection{Definitions}
\begin{definition}
The \emph{entropy} of a set of discrete variables $X_{\mathcal{V}}$, is defined as:
$$
{\sf H}(X_{\mathcal{V}})=\sum_{x_{\mathcal{V}}} P(X_{\mathcal{V}}=x_{\mathcal{V}}) \log P(X_{\mathcal{V}}=x_{\mathcal{V}}).
$$
\end{definition}

\begin{definition}
The interaction information of a set of discrete variables $X_{\mathcal{V}}$, is defined as: $${\sf I}(X_{\mathcal{V}}) = \sum_{S \subseteq V} (-1)^{\abs{V \setminus S}} {\sf H}(X_S).$$
\end{definition}

We have the following identity: $${\sf H}(X_{\mathcal{V}}) = \sum_{T \subseteq V} \sum_{S \subseteq T} (-1)^{\abs{T \setminus S}} {\sf H}(X_S)=\sum_{T \subseteq V} {\sf I}(X_T).$$

The plug-in estimate of entropy of $X_{\mathcal{V}}$, is $$\hat{{\sf H}}(X_{\mathcal{V}}) = \sum_{x_{\mathcal{V}}} \frac{N(x_{\mathcal{V}})}{N} \log\frac{N(x_{\mathcal{V}})}{N}.$$
The plug-in estimate of interaction information of $X_{\mathcal{V}}$, is $$\hat{{\sf I}}(X_{\mathcal{V}}) = \sum_{S \subseteq V} (-1)^{\abs{V \setminus S}} \hat{{\sf H}}(X_S).$$

Theorem \ref{thm:imset as sum of elementary imset for general MAGs} allows us to decompose the entropy in terms of the interaction information over the parametrizing set. We first need a definition to simplify the expression.

\begin{definition}\label{inner product between imset and functions}
Given a function $f$ which takes $X_A$ for any $A \subseteq V$ as input, and an imset $u$ over $V$, we define $$\langle u, f \rangle = \sum_{A \subseteq V} u(A) f(x_A).$$
\end{definition}
Let ${\sf H}, {\sf I}$ be the entropy function and the interaction information function, respectively.
\begin{proposition}\label{connection between u_G and entropy}
For a MAG $\G$ with vertex set $V$ and a distribution $p$ that is Markov to $\G$, we have 
\begin{align*}
    {\sf H}(X_{\mathcal{V}}) &= \sum_{T \in \mathcal{S}(\G)} {\sf I}(X_T) = \langle c_{\G}, {\sf I} \rangle = \langle \delta_{\mathcal{V}}-u_{\G}, {\sf H} \rangle.
    \end{align*}
\end{proposition}

\begin{proof}
Recall that $\mathcal{P}(S)$ denotes the power-set of $S$. 
By Theorem \ref{thm:imset as sum of elementary imset for general MAGs}, it is sufficient to show that $\langle u_{\G}, {\sf H} \rangle = \sum_{T \notin \mathcal{S}(\G)} {\sf I}(X_T) = 0$.
\begin{align*}
\langle u_{\G}, {\sf H} \rangle &= \sum_{S \in \mathcal{P}(\mathcal{V})} \sum_{\substack{T \in \mathcal{P}(\mathcal{V})\\S \subseteq T \subseteq \mathcal{V}}} (-1)^{\abs{T \setminus S}} (1-c_{\G}(T)) {\sf H}(X_S)\\
&= \sum_{S \in \mathcal{P}(\mathcal{V})} \sum_{\substack{T \in \mathcal{P}(\mathcal{V})\\S \subseteq T \subseteq \mathcal{V}}} (-1)^{\abs{T \setminus S}} \delta_{T \notin \mathcal{S}(\G)} {\sf H}(X_S)\\
&= \sum_{S \in \mathcal{P}(\mathcal{V})} \sum_{\substack{T \notin \mathcal{S}(\G)\\S \subseteq T \subseteq \mathcal{V}}} (-1)^{\abs{T \setminus S}}  {\sf H}(X_S)\\
&= \sum_{T \notin \mathcal{S}(\G)} \sum_{S \subseteq T} (-1)^{\abs{T \setminus S}}  {\sf H}(X_S)\\
&= \sum_{T \notin \mathcal{S}(\G)}  {\sf I}(X_T).
\end{align*}

The last equality in the statement then holds, as $\langle \delta_{\mathcal{V}}, {\sf H} \rangle = {\sf H}(X_{\mathcal{V}})$.
\end{proof}

A consistent scoring criterion requires the following property.

\begin{definition}
Let $\mathcal{D}$ be a set of data that are i.i.d.~samples from a distribution $P$. A scoring criterion $S$ is asymptotically consistent if in the limit when the number of sample grows, the following holds:
\begin{itemize}
    \item[1] If $\G$ contains $P$ and $\G^{'}$ does not, then $S(\G, \mathcal{D}) < S(\G^{'}, \mathcal{D})$.
    \item[2] If both $\G$ and $\G^{'}$ contain $P$ but $\G^{'}$ has higher dimension then $\G$, then $S(\G, \mathcal{D}) < S(\G^{'}, \mathcal{D})$.
\end{itemize}
\end{definition}

\subsection{Bayesian networks}
We firstly show that if we are fitting DAG models, then $\hat{\ell}$ is equal to $N \langle c_{\G}, \hat{{\sf I}} \rangle$.
Recall that for DAGs, the maximum likelihood estimate (MLE) of $P(x_v \mid x_{\pa_{v}})$ is $N(x_v,x_{\pa_{v}})/ N(x_{\pa_{v}}).$ Given $N$ i.i.d.~samples,  the log-likelihood can be expressed as:
\begin{align*}
\ell(P;N) &= \sum_{x_{\mathcal{V}}} N(x_{\mathcal{V}}) \log P(x_{\mathcal{V}})\\
&= \sum_{x_{\mathcal{V}}} N(x_{\mathcal{V}}) \sum_v \log P(x_v \mid x_{\pa(v)})\\
&= \sum_v \sum_{x_v,x_{\pa(v)}} N(x_v, x_{\pa(v)}) \log P(x_v \mid x_{\pa(v)})\\
&= \sum_v \sum_{x_{\pa(v)}} \sum_{x_v} N(x_v, x_{\pa(v)}) \log P(x_v \mid x_{\pa(v)}).
\end{align*}
Replacing $P(x_v \mid x_{\pa(v)})$ by the MLE $N(x_v,x_{\pa_{v}})/ N(x_{\pa_{v}})$, the log-likelihood is then 
\begin{align*}
\hat{\ell} &= \sum_v \sum_{\pa(v)} \sum_{x_v} N(x_v, x_{\pa(v)}) \log \frac{N(x_v, x_{\pa(v)})}{N(x_{\pa(v)})}\\
&= \sum_v \sum_{\pa(v)} \left\{\sum_{x_v} N(x_v, x_{\pa(v)}) \log \frac{N(x_v, x_{\pa(v)})}{N} + \right.\\
&\phantomrel{=} \qquad- \left.  \sum_{x_v} N(x_v, x_{\pa(v)}) \log \frac{N( x_{\pa(v)})}{N}\right\}\\
&= \sum_v \sum_{\pa(v)} \sum_{x_v} N(x_v, x_{\pa(v)}) \log \frac{N(x_v, x_{\pa(v)})}{N}+\\
&\phantomrel{=}\qquad- \sum_v \sum_{\pa(v)}  N( x_{\pa(v)}) \log \frac{N( x_{\pa(v)})}{N}\\
&= N\langle \delta_{\mathcal{V}}-u_{\G}, \hat{{\sf H}} \rangle\\
&= N\langle c_{\G}, \hat{{\sf I}} \rangle.
\end{align*}

The above calculations have appeared in the literature before \citep{studeny2006probabilistic} but we demonstrate it in our notation to make other proofs clear.

\subsection{Consistency for MAGs when imset is perfectly Markovian}

In this section, we show that when fitted MAGs have imsets that are perfectly Markovian w.r.t.\ the graphs, the scoring criterion is consistent. First, we introduce the well-studied Kullback-Leibler divergence.
\begin{definition}
For two distributions $P$ and $Q$ defined over $X_{\mathcal{V}}$, and two disjoint subsets $A,C \subseteq V$, define the \emph{Kullback-Leibler} (KL) divergence between $P_{A|C}$ and $Q_{A|C}$ as
$$
\infdiv{P_{A \mid C}}{Q_{A \mid C}} = \mathbf{E}_{P_{A\mid C}}[\log{\frac{P_{A\mid C}}{Q_{A \mid C}}}].
$$
\end{definition}

Two nice properties of KL divergence are that it is non-negative, and that
it is zero if and only if $P=Q$ almost surely. 

\begin{proposition}\label{KL inner product}
For any two distributions $P$ and $Q$ and any semi-elementary imset $u_{\langle A,B \mid C \rangle}$, provided that $Q$ satisfies $X_A \indep X_B \mid X_C$, we have $$\langle u_{\langle A, B \mid C\rangle}, \infdiv{P}{Q} \rangle \geq 0,$$ with equality holding if and only if $P$ also satisfies $X_A \indep X_B \mid X_C$.
\end{proposition}
\begin{proof}
\begin{align*}
\langle u_{\langle A, B \mid C\rangle}, \infdiv{P}{Q} \rangle &= \mathbf{E}_{P_C}[\infdiv{P_{AB \mid C}}{Q_{AB \mid C}} -\infdiv{P_{A \mid C}}{Q_{A \mid C}}-\infdiv{P_{B \mid C}}{Q_{B \mid C}}]\\
&= \mathbf{E}_{P_C}\left[{\sf I}(A;B \mid C)-\mathbf{E}_{P_{AB \mid C}}\log \frac{Q_{AB\mid C}}{Q_{A\mid C} Q_{B \mid C}}\right]\\
&= \mathbf{E}_{P_C}[{\sf I}(A;B \mid C)] \qquad \text{(provided } Q \text{ satisfies }A \indep B \mid C)\\
&\geq 0, \qquad
\end{align*}
where the last inequality comes from the fact that mutual information is always non-negative.
\end{proof}

Next we show that the score is consistent. 

\begin{proposition}\label{prop:new score is consistent}
The score $-2N \langle c_{\G}, \hat{{\sf I}} \rangle+ d\log N$, is consistent when $\I_{u_\G} = \I_\G$, where $\G$ is the fitting MAG.

\end{proposition}

\begin{proof}
Let us first consider any general MAG $\G$ and $q^{\theta}$ with the parameters $\theta = \{q(x_H \mid x_{\tail(H)}): H \in \mathcal{H}(\G)\}$ \citep{Evans2014}. Note that $q^{\theta}$ can be factorized by Theorem \ref{thm:imset as sum of elementary imset for general MAGs} and the log-likelihood can be rewritten as the inner product between $\delta_{\mathcal{V}}-u_{\G}$ and $\log q^{\theta}$ where the inner product is taken over all subsets of $V$:
\begin{align*}
\ell &= \sum_{x_{\mathcal{V}}} N(x_{\mathcal{V}}) \langle \delta_{\mathcal{V}}-u_{\G}, \log q^{\theta} \rangle \\
&= N \langle \delta_{\mathcal{V}}-u_{\G}, \sum_{x_{\mathcal{V}}} \frac{N(x_{\mathcal{V}})}{N} \log q^{\theta} \rangle.
\end{align*}
The difference between $\hat{\ell}$ and $N \langle c_{\G}, \hat{{\sf I}} \rangle$ is then
\begin{align*}
 \hat{\ell}-N \langle c_{\G}, \hat{{\sf I}} \rangle &= N \langle \delta_{\mathcal{V}}-u_{\G} ,  \sum_{x_{\mathcal{V}}} \frac{N(x_{\mathcal{V}})}{N} \log q^{\hat{\theta}} - \hat{\sf H}\rangle\\
 &= -N \langle \delta_{\mathcal{V}}-u_{\G}, \infdiv{\hat{P}}{q^{\hat{\theta}}} \rangle.
\end{align*}
Hence the difference from the BIC of the fitted model $q^{\hat{\theta}}$ with graph $\G$ and the new score is 
$$
D_1= 2N \times \langle \delta_{\mathcal{V}}-u_{\G},  \infdiv{\hat{P}}{q^{\hat{\theta}}} \rangle. 
$$
Moreover the difference between the fitted model $q^{\hat{\theta}}$ with graph $\G$ and the BIC of the true model $p^{\hat{\theta}^{'}}$ is:
\begin{align*}
D_2 &= -2\sum_{x_{\mathcal{V}}} N(x_{\mathcal{V}}) \log q^{\hat{\theta}}(x_{\mathcal{V}})+d_2 \log N + 2\sum_{x_{\mathcal{V}}} N(x_{\mathcal{V}}) \log p^{\hat{\theta}^{'}}(x_{\mathcal{V}})-d_1 \log N\\
&= 2N \sum_{x_{\mathcal{V}}} \hat{P}(x_{\mathcal{V}}) \log \frac{p^{\hat{\theta}^{'}}(x_{\mathcal{V}})}{q^{\hat{\theta}}(x_{\mathcal{V}})}+(d_2-d_1) \log N\\
&= 2N \sum_{x_{\mathcal{V}}} \hat{P}(x_{\mathcal{V}}) \log\left[\frac{\hat{P}(x_{\mathcal{V}})}{q^{\hat{\theta}}(x_{\mathcal{V}})}\frac{p^{\hat{\theta}^{'}}(x_{\mathcal{V}})}{\hat{P}(x_{\mathcal{V}})}\right] +(d_2-d_1) \log N\\
&= 2N \times \infdiv{\hat{P}}{q^{\hat{\theta}}}-2N \times \infdiv{\hat{P}}{p^{\hat{\theta}^{'}}}+(d_2-d_1) \log N,
\end{align*}
where $d_1,d_2$ denote the dimension of the true model and fitted model, respectively.

To show that the score is consistent, it is sufficient to show that the following conditions are satisfied:
\begin{itemize}
    \item[(i)] when $D_2=0$, i.e.\ the fitted model is the true model, $D_1 = O_p(1)$;
    \item[(ii)] when $D_2 > 0$, i.e.\  the fitted model is not the true model,  we have $D_2-D_1$ diverges at rate at least $O_p(\log N)$.
\end{itemize}

Consider condition (i) when we are fitting the true model. For any subset $A \subseteq V$, the term $$2N \times \infdiv{\hat{P}(X_A)}{q^{\hat{\theta}}(X_A)}$$ has $\chi^2$ distribution with some degree of freedom. Since $\delta_{\mathcal{V}}-u_{\G}$ has fixed number of terms, $D_1=O_p(1)$ on all these as $N$ tends to infinity.

Consider condition (ii), we have:
$$
D_2-D_1 = 2N \times \langle u_{\G}, \infdiv{\hat{P}}{q^{\hat{\theta}}} \rangle-2N \times \infdiv{\hat{P}}{p^{\hat{\theta}^{'}}}+(d_2-d_1) \log N.
$$

The second term $2N \times \infdiv{\hat{P}}{p^{\hat{\theta}^{'}}}$ has $\chi^2_{d_1}$ distribution, so it is $O_p(1)$. Moreover, by Proposition \ref{KL inner product} and given that $u_{\mathcal{}}=\sum u_{\langle A, B \mid C\rangle}$ is combinatorial, we have:
\begin{align*}
2N \times \langle u_{\G}, \infdiv{\hat{P}}{q^{\hat{\theta}}} \rangle &= 2N \sum \mathbb{E}_{\hat{P}_C}[{\sf I}(A;B \mid C)]    
\end{align*}

Now we need to split into two cases. The first case is if the fitted model contains the true model, but has higher dimension and contains more parameters then $d_2 > d_1$ so the third term $(d_2-d_1) \log N$ is at $O(\log N)$. For the first term, because the fitted model contains the true model, it satisfies all the conditional independence corresponding to  the semi-elementary imsets in $u_{G}$. Treating each $2N \times \mathbb{E}_{\hat{P}_C}[{\sf I}(A;B \mid C)]$ as a likelihood ratio test, it has chi-squared distribution, so in total the first two terms are at $O_p(1)$. Thus in this case, $D_2-D_1 = O_p(\log N)$.

For the second case when the fitted model is wrong and does not contain the true model, given that $u_{\G}$ is combinatorial and is perfectly Markovian with respect to the graph, there exists at least one conditional independence in $u_{\mathcal{G}}=\sum u_{\langle A, B \mid C\rangle}$ such that $\mathbb{E}_{P_C}[{\sf I}(A;B \mid C)] = \lambda > 0$. Empirically, when $N$ tends to infinity, $\mathbb{E}_{\hat{P}_C}[{\sf I}(A;B \mid C)]$ will be close to $\lambda$, and hence the first term in $D_2-D_1$ grows at $O_p(N)$. Even if the third term $(d_2-d_1) \log N$ is negative, it grows at $O(\log N)$, which is slower than the first term, hence in the second case, $D_2-D_1=O_p(N).$
\end{proof}

\section{Bidirected graphs}\label{sec:bidirected graphs}

For general MAGs, the problem of characterizing the conditions for when $\I_{u_\G}$ is well defined and $\I_{u_\G} = \I_{\G}$, seems hard in general. To reduce the difficulty, we focus on bidirected graphs in this section. We will give a condition such that our proposed `standard' imsets $u_{\G}$ for \emph{bidirected graphs} are always combinatorial and perfectly Markovian with respect to $\G$. We have computationally verified that this condition is also necessary for $\abs{\mathcal{V}} \leq 7$, and not found any graphs for which it is not. 

\begin{definition}\label{dualgraph}
Let $\G$ be a bidirected graph, and define its undirected 
\emph{dual graph}, $\overline{\G}$, by $i - j \in \overline{\G}$ 
if and only if $i \not\leftrightarrow j$ in $\G$.
\end{definition}

The dual graph is a powerful tool when we analyse the bidirected graphs; see Example \ref{exm:5chain}.

\subsection{How does the characteristic imset help?}

There are many MAGs that are not simple, but still have combinatorial standard imsets that are perfectly Markovian with respect to the graph; one example is the bidirected 4-cycle, which has a standard imset consisting of the elementary imsets for its two marginal independences. 
For bidirected graphs, we found that since the vertices can be given any topological order and there are many heads (any connected subgraph), it is difficult to directly decompose the standard imset or prove the validity of decomposition; however, it turns out to be easier if we work with the characteristic imset. 
In this section, we consider the relationship between the semi-elementary imset decomposition of combinatorial standard imsets and the characteristic imset.  

Recall that two conditional independences $I_1, I_2$ overlap if 
$\bar{\Sset}(I_1) \cap \bar{\Sset}(I_2) \neq \emptyset$.




\begin{proposition}\label{how characteristic imset help}
A `standard' imset $u_{\G}$ is combinatorial if and only if there is a list of non-overlapping conditional independences $\mathbb{L}$ such that $\cup_{I \in \mathbb{L}} \bar{\Sset}(I) = \mathcal{P}(\mathcal{V}) \setminus \mathcal{S}(\G)$.
\end{proposition}

\begin{proof}
This follows from Lemma \ref{lemma:linear relation between u and c}, 
Corollary \ref{cor:imset must agree on S(G)} and the arguments used
in their proofs.
\end{proof}

For example, if $\G$ is the bidirected 4-cycle $1 \leftrightarrow 2\leftrightarrow 3\leftrightarrow 4\leftrightarrow 1$ the independences are $1\indep 3$ and $2 \indep 4$.  The corresponding sets to be constrained are $\{1,3\}$ and $\{2,4\}$, and these are clearly disjoint.  Therefore the standard imset is perfectly Markovian and is simply $u_{\G} = u_{\langle 1,3 \rangle}+u_{\langle 2,4 \rangle}$. 

Proposition \ref{prop:induced_subgraph} allows us to quickly determine if a standard imset is not perfectly Markovian w.r.t.\ the graph by checking if the graph contains certain ancestral structures.

\begin{figure}
    \centering
    \begin{tikzpicture}[rv/.style={circle, draw, thick, minimum size=5mm, inner sep=0.5mm}, node distance=14mm, >=stealth]
\begin{scope}[xshift=4cm, yshift=1.5cm]
    \node[] (0) {};
    \node[rv] (1) at (30:.8cm) {1};
    \node[rv] (2) at (150:.8cm) {2};
    \node[rv] (3) at (270:.8cm) {3};
    \node[rv] (4) at (90:1.6cm) {4};
    \node[rv] (5) at (330:1.6cm) {5};
    \node[rv] (6) at (210:1.6cm) {6};
\draw[very thick] (3) -- (2) -- (1) -- (3);
\draw[very thick] (1) -- (4) -- (2) -- (6) -- (3) -- (5) -- (1); 
\node at (270:1.8cm) {(a)};
    \end{scope}
\begin{scope}[xshift=8cm, yshift=-.4cm]
    \node[] (0) {};
    \node[rv] (1) at (180:.6cm) {1};
    \node[rv] (3) at (60:.6cm) {3};
    \node[rv] (5) at (300:.6cm) {5};
    \node[rv] (4) at (180:1.8cm) {4};
    \node[rv] (6) at (60:1.8cm) {6};
    \node[rv] (2) at (300:1.8cm) {2};
\draw[very thick] (1) -- (3) -- (5) -- (1) -- (4) -- (2) -- (6) -- (4);
\draw[very thick] (3) -- (6); 
\draw[very thick] (2) -- (5);
\node at (270:1.8cm) {(c)};
    \end{scope}
    \begin{scope}[xshift=8cm, yshift=-5cm]
    \node[] (0) {};
    \node[rv] (1) at (180:.6cm) {1};
    \node[rv] (3) at (60:.6cm) {3};
    \node[rv] (5) at (300:.6cm) {5};
    \node[rv] (4) at (180:1.8cm) {4};
    \node[rv] (6) at (60:1.8cm) {6};
    \node[rv] (2) at (300:1.8cm) {2};
\draw[very thick] (1) -- (3) -- (5) -- (1) -- (4) -- (2) -- (6) -- (4) to[bend left=8] (3);
\draw[very thick] (3) -- (6); 
\draw[very thick] (2) -- (5);
\node at (270:1.8cm) {(f)};
    \end{scope}
       \begin{scope}[xshift=0cm, yshift=-.4cm]
    \node[] (0) {};
    \node[rv] (1) at (70:1.5cm) {1};
    \node[rv] (2) at (110:1.5cm) {2};
    \node[rv] (3) at (150:1.5cm) {3};
    \node[] (4) at (190:1.5cm) {};
    \node[rv] (k) at (30:1.5cm) {$k$};
    \node[] (k1) at (-10:1.5cm) {};
\draw[very thick] (3) -- (2) -- (1) -- (k);
\draw[very thick, dashed] (4) -- (3); 
\draw[very thick, dashed] (k) -- (k1);
\node at (270:1cm) {(b)};
    \end{scope}
             \begin{scope}[xshift=37mm, yshift=-25mm]
     \node[rv] (2) {2};
    \node[rv, below of=2] (3) {3};
    \node[rv, right of=2] (5) {5};
    \node[rv, right of=3] (4) {4};
    \node[rv, yshift=-7mm, xshift=-10mm] (1) at (2) {1};
    \node[rv, xshift=-7mm, yshift=10mm] (6) at (5) {6};
\draw[very thick] (3) -- (2) -- (5) -- (4) -- (3) -- (1) -- (2);
\draw[very thick] (2) -- (6) -- (5); 
\node[below of=3, xshift=3mm, yshift=5mm] {(d)};
\end{scope}
         \begin{scope}[xshift=-7mm, yshift=-37mm]
    \node[rv] (2) {2};
    \node[rv, below of=2] (3) {3};
    \node[rv, right of=2] (5) {5};
    \node[rv, right of=3] (4) {4};
    \node[rv, yshift=-7mm, xshift=-10mm] (1) at (2) {1};
    \node[rv, yshift=-7mm, xshift=10mm] (6) at (5) {6};
\draw[very thick] (3) -- (2) -- (5) -- (4) -- (3) -- (1) -- (2);
\draw[very thick] (4) -- (6) -- (5); 
\node[below of=3, xshift=7mm, yshift=5mm] {(e)};
    \end{scope}
            \begin{scope}[xshift=3.7cm, yshift=-7cm, node distance=12mm]
    \node[rv] (2) {2};
    \node[rv, below of=2] (3) {3};
    \node[rv, right of=2] (5) {5};
    \node[rv, right of=3] (4) {4};
    \node[rv, left of=2] (1) {1};
    \node[rv, above of=2] (6) {6};
\draw[very thick] (3) -- (2) -- (5) -- (4) -- (3) -- (1) -- (2);
\draw[very thick] (2) -- (6) -- (5); 
\draw[very thick] (1) -- (6);
\node[below of=3, xshift=3mm, yshift=5mm] {(g)};
    \end{scope}

\end{tikzpicture}
    \caption{ Forbidden dual subgraphs: (a) triangles; (b) a $k$-cycle, for $k \geq 5$; (c) dual to the 6-cycle; (f) dual to the 6-chain; (d), (e), (g) other graphs with at least one chordless 4-cycle.}
    \label{fig:forbidden}
\end{figure}
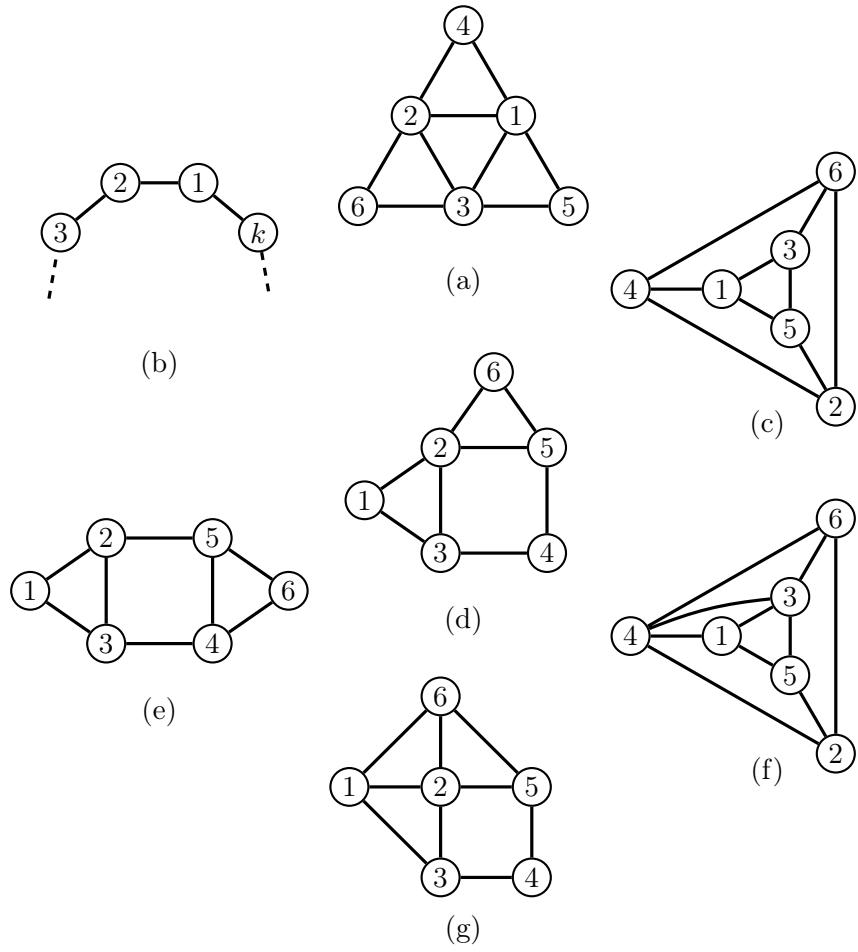

\begin{remark}
In Figure \ref{fig:forbidden}, we display all dual graphs with at most 6 vertices that must not appear as induced subgraphs to the dual of a bidirected graph for the standard imset to be perfectly Markovian with respect to the graph. This is primarily a consequence of Proposition \ref{prop:induced_subgraph}.
\end{remark}

\subsection{For which bidirected graphs do we get standard imsets that are perfectly Markovian?}

We will present a theorem which gives sufficient criteria for bidirected graphs to have combinatorial standard imsets and induce the same model as the graph. For $\abs{\mathcal{V}} \leq 7$ we have empirically verified that this is also necessary. Before we present the theorem, we give a motivating example.

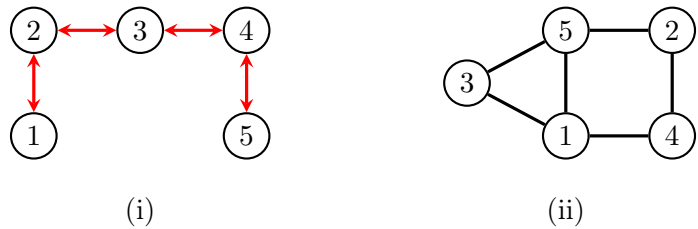
\begin{figure}
    \centering
     \begin{tikzpicture}
  [rv/.style={circle, draw, thick, minimum size=6mm, inner sep=0.8mm}, node distance=14mm, >=stealth]
  \pgfsetarrows{latex-latex}
\begin{scope}
  \node[rv]  (1)            {$1$};
  \node[rv, above of=1] (2) {$2$};
  \node[rv, right of=2] (3) {$3$};
  \node[rv, right of=3] (4) {$4$};
  \node[rv, below of=4] (5) {$5$};

  \draw[beg] (1) -- (2);
  \draw[beg] (2) -- (3);
  \draw[beg] (3) -- (4);
  \draw[beg] (4) -- (5);
  \node[below of=3, yshift=-1cm] {(i)};
  \end{scope}
\begin{scope}[xshift=7cm]
  \node[rv]  (1)            {$1$};
  \node[rv, left of=1, yshift=7mm, xshift=1mm] (3) {$3$};
  \node[rv, right of=1] (4) {$4$};
  \node[rv, above of=1](5){$5$};
  \node[rv, right of=5] (2) {$2$};

  \draw[-,very thick, black] (1) -- (3);
  \draw[-,very thick, black] (1) -- (5);
  \draw[-,very thick, black] (3) -- (5);
  \draw[-,very thick, black] (1) -- (4);
  \draw[-,very thick, black] (2) -- (4);
  \draw[-,very thick, black] (2) -- (5);

  \node[below of=5, yshift=-1cm] {(ii)};
  \end{scope}

\end{tikzpicture}
    \caption{(i) The bidirected 5-chain and (ii) its dual graph.}
    \label{fig:5-chain and its dual}
\end{figure}

\begin{example} \label{exm:5chain}
Consider the bidirected 5-chain and its dual graph in Figure \ref{fig:5-chain and its dual}. Consider the following list of conditional independences: 
\begin{align*}
    & \left\{ \begin{array}{c}
    4 \indep 2 \cmid 1,5\\
    4 \indep 1 \cmid 3,5 
    \end{array} \right\} &
    2 &\indep 5 \cmid 1,3 &
    5 &\indep 1,3 &
    1 &\indep 3. 
\end{align*}
\end{example}
One can check that the sum of semi-elementary imsets corresponding to the above list of independences is the same as the standard imset for the bidirected 5-chain.  In addition the
standard imset is perfectly Markovian with respect to the graph. 

This decomposition starts with 4, and by symmetry it could also start with 2; however, none of the other vertices will work, and we cannot obtain a decomposition in this manner if we try to do so.  In particular, if we take the imset for the graph where all edges for 1, 3, 5 are added and subtract from it the standard imset for the 5-chain, what remains is not a structural imset.  This shows that for bidirected graphs, even though the topological order can be arbitrary, in order to properly decompose the standard imset further restrictions are required.  

Here are some observations. In the dual graph, 4 has neighbours 1 and 2, these vertices share one common neighbour 5, and 1 has one more neighbour 3. However for 3, two of its neighbours have distinct neighbours. This suggests that perhaps we need the neighbours of neighbours to be nested within one another. 

For a vertex $v$, we may also want to treat its neighbours which have the same neighbours as a block, since any path to any one of them also lead to any other vertex in the block. Then any edge within one block does not have any effect on any path to $v$.

Suppose we have a vertex $v$ and that it has neighbours $A$ in the dual graph.  Now partition its neighbours $A^{j}, 1 \leq j \leq m$ such that the neighbours of each $w \in A^j$ (outside $\{v\} \cup A$) are precisely the set $N^{j}$, and such that $N^{1} \subseteq N^{2} \subseteq \dots \subseteq N^{m}$. 

Now the question is how the edges connects the blocks of neighbours of $v$ so that $v$ is a valid ending vertex. One observation from the 5-chain is that for two blocks $A_i$ and $A_j$, if there is one edge between $A_i$ and $A_j$, then we want $A_i$ and $A_j$ to be fully connected. For example, consider vertex 1 which has neighbours $\{3,4,5\}$. Vertices 4 and 5 share the same neighbours outside $A$, so they are in the same block, while 3 is connected only to 5.

Another question is how $A_i$ and $A_j$ are connected, i.e.~the cross-block edges. Suppose $v$ has three blocks of neighbours, $A_1, A_2$ and $A_3$. There are 8 possible ways to connect the three blocks. We cannot write out a proper decomposition for three of them. They are: $A_1 - A_3$, $A_2-A_3$ and $A_1 - A_2 - A_3$. Inspired by the 5-chain example, we have the following definition. For convenience, we introduce the notation $A^{[i,j]}$ to denote $\bigcup_{k=i}^j \{A^k\}$ and each $A^k$ is referred as a block (of vertices). 

\begin{definition}\label{rootedneigihbours}
For a vertex $v$ and an ascending partition of its neighbours $A^{[j_l, j_h]}$,
we say it is \emph{rooted} if:
\begin{itemize}
    \item[(i)] it is empty; or
    \item[(ii)] there exists a block $A^j$ (the root) for some $j_l \leq j \leq j_h$ and 
    $H^j := A^{[j+1, j_h]}$
    and 
    $L^j := A^{[j_l, j-1]}$
    (respectively empty if $j=j_h$ or $j_l$) such that $A^j$ and each set in $H^j$ are fully connected to each set in $L^j$, and for every set in $H^j$ it is entirely disconnected from $A^j$.  In addition, $H^j$ and $L^j$ must themselves be rooted.
\end{itemize}
\end{definition}

One can check that for three blocks and 8 possible ways to connect them, the three ways that do not lead to a proper decomposition also do not satisfy the above definition, but that the other five ways do.

\begin{definition}
For a sequence of blocks $A^{[j_l,j_h]}$ rooted at 
$A^j$,
we define $T^j$ to be the subset of $A^{[1,j_l-1]}$
such that $A^i$ is in $T^j$ if and only if it is connected to $A^j$.
\end{definition}

\begin{lemma}\label{root of subset}
Let $A^{[j_l,j_h]}$ be a collection of neighbours of $v$ 
that is rooted.  Then any subset of it is also rooted.
\end{lemma}

\begin{lemma}\label{uniqueroot}
For a vertex $v$ and a collection of its consecutive neighbours, if it is rooted then the root is unique.
\end{lemma}

Note that $H^j$ or $L^j$ and $A^{[i,j]}$ are written as sets of sets, but when we use them in certain set expressions or in any independence, we just think of it as the union of the blocks contained in that set, i.e.~sets of variables.  The context should prevent any confusion.

\begin{lemma}\label{L^jrelation}
Suppose a block ${A}^{[j_l,j_h]} = H^j \cup L^j \cup \{A^j\}$ is rooted at $A^j$, where $H^j$ and $L^j$ are rooted at $A^k$ and $A^t$ respectively (non-empty), then $T^k = T^j \cup L^j$ and $T^t = T^j$. Moreover $k_l = j+1$ and $t_l = j_l$.
\end{lemma}

\begin{theorem}\label{thm: rooted bidirected graph}
For a bidirected graph, if there exists an ordering of the vertices such that for each $i$ in $\overline{\G}_{[i]}$, its neighbours $A$ can be partitioned into blocks $\{A^{j}, 1 \leq j \leq m\}$ such that
\begin{itemize}
    \item[(i)] the neighbours of each $w \in A^j$ (outside $\{v\} \cup A$) in $\overline{\G}$ are the same set $N^{j}$; 
    \item[(ii)] $N^{1} \subseteq N^{2} \subseteq \dots \subseteq N^{m}$;
    \item[(iii)] for any two blocks, either there is no edge between them or they are fully connected;
    \item[(iv)] $\{A^{j}, 1 \leq j \leq m\}$ are rooted.
\end{itemize}
Then $u_\G$ is combinatorial and $\I_{u_{\G}} = \I_{\G}$.
\end{theorem}

The key to the proof of Theorem \ref{thm: rooted bidirected graph} is to construct a proper independence decomposition of $u_{\G}$ with the aid of Proposition \ref{how characteristic imset help}, and show that it is equivalent to the global Markov property. The following example will give some intuition on how the independences are constructed.
\begin{example}
Consider a representation of the dual graph $\overline{\G}$ in Figure \ref{bidirected graph theorem example}. The vertex $i$ has neighbours partitioned into $A^1, A^2, A^3$ which have neighbours $N^1 \subseteq N^2 \subseteq N^3$ respectively. 

To construct a proper decomposition of $u_{\G}$, we first look at the following conditional independences: 
\begin{align}
    i \indep \widetilde{A} \cmid (\cap_{w \in \widetilde{A}} \nb(w))\setminus \{i\}, \label{nb independence}
\end{align}
 for every $\widetilde{A} \subseteq A = \nb(i)$. These independences are definitely implied by the graph as none of the conditioning variables are siblings of $\widetilde{A}$, so no paths are open. Moreover, one can see that each disconnected set containing $i$ is associated with one of the independences. However, the independences also overlap. Proposition \ref{how characteristic imset help} suggests that we do not need all of them, but just a subset of independences (that may require further modification), such that:
\begin{itemize}
    \item[(i)] the associated constrained sets do not overlap; and  
    \item[(ii)] the independences associated with any disconnected set can be deduced from them using the semi-graphoid axioms.
\end{itemize}
Let's consider the following independences from (\ref{nb independence}):
\begin{align}
    i &\indep A^3, A^2, A^1 \cmid N^1\label{K1}\\
    i &\indep A^3, A^2 \cmid N^2, A^1\label{K2}\\
    i &\indep A^3, A^1 \cmid N^1\label{K3}\\
    i &\indep A^3 \cmid N^3, A^1\label{K4}\\
    i &\indep A^2 \cmid N^2, A^1\label{K5}\\
    i &\indep A^1 \cmid N^1, A^2, A^3.\label{K6}
\end{align}
(\ref{K3}), (\ref{K5}) and (\ref{K6}) are each implied by (\ref{K1}), (\ref{K2}) and (\ref{K1}) respectively, so we can ignore them. Then looking at (\ref{K1}), (\ref{K2}) and (\ref{K4}) we see that there are some overlaps. The set $\{i\} \cup A^3 \cup A^1$ is associated with both (\ref{K1}) and (\ref{K4}), but the latter cannot be further simplified so we marginalize $A^3$ in (\ref{K1}). 

The set $\{i\} \cup A^3 \cup A^2 \cup A^1$ only appears in (\ref{K2}) so we want to keep that, but $\{i\} \cup A^3 \cup A^1$ is associated with both (\ref{K2}) and (\ref{K4}). Move $A^3$ into the conditioning set for (\ref{K2}) to resolve this. 
Also $\{i\} \cup A^2 \cup A^1$ is associated with both (\ref{K1}) and (\ref{K2}), so we marginalize $A^2$ in (\ref{K1}).
This gives the following:
\begin{align*}
    (\ref{K1}) \quad &\rightarrow \quad i \indep A^1 \cmid N^1\\
    (\ref{K2}) \quad &\rightarrow \quad i \indep A^2 \cmid A^3, A^1,N^2\\
    (\ref{K4}) \quad &\rightarrow \quad \text{keep } (\ref{K4}).
\end{align*}
One can check that any independence involving $i$ can be deduced from the above three independences, together with other independences that do not involve $i$ from the induction hypothesis. In addition, these cover all the disconnected sets containing $i$.

\begin{figure}
    \centering
     \begin{tikzpicture}
  [rv/.style={circle, draw, thick, minimum size=6mm, inner sep=0.8mm},
  bl/.style={rectangle, draw, thick, minimum size=6mm, inner sep=0.8mm}, node distance=14mm, >=stealth,
  node distance=14mm, >=stealth]
  \pgfsetarrows{latex-latex}
\begin{scope}
  \node[rv]  (1)            {$i$};
  \node[bl, below of=1,  xshift=-1.5cm] (2) {$A^1$};
  \node[bl, below of=1] (3) {$A^2$};
  \node[bl, below of=1,xshift=1.5cm](4){$A^3$};
  \node[bl, below of=2, xshift=-.2cm] (5) {$N^1$};
  \node[bl, below of=3] (6) {$N^2 \setminus N^1$};
  \node[bl, below of=4, xshift=.5cm] (7) {$N^3 \setminus N^2$};

  \draw[-,very thick, black] (1) -- (2);
  \draw[-,very thick, black] (1) -- (3);
  \draw[-,very thick, black] (1) -- (4);
  \draw[-,very thick, black] (2) -- (5);

  \draw[-,very thick, black] (3) -- (6);

\draw[-,very thick, black] (3) -- (5);
\draw[-,very thick, black] (4) -- (5);
\draw[-,very thick, black] (4) -- (6);

  \draw[-,very thick, black] (4) -- (7);
  \draw[-,very thick, black] (2) -- (3);
  \draw[-,very thick, black] (2) to[bend left=40] (4);

  \end{scope}

\end{tikzpicture}
    \caption{Example for Theorem \ref{thm: rooted bidirected graph}}
    \label{bidirected graph theorem example}
\end{figure}
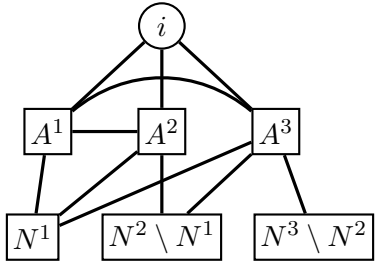
\end{example}
\begin{proof}[Proof of Theorem \ref{thm: rooted bidirected graph}]
The main idea of this proof is to construct a list of independences $\mathbb{L}$, where the sum of semi-elementary imsets corresponding to the independences in $\mathbb{L}$ is $u_{\G}$, then using the similar ideas in the proof of Theorem \ref{thm: standard imset with head size less 3}, we just need to show that $\mathbb{L}$ holds in $\G$ so $I_{u_{\G}} \subseteq I_{\G}$, and $\mathbb{L}$ implies the global Markov property of $\G$ so $I_{\G} \subseteq I_{u_{\G}}$ or the disconnected Markov property from \citet{drton2008binary}.

We will proceed by induction on the ordering of vertices. For a vertex $v$, suppose its neighbours in $\overline{\G}$ are partitioned into $A^{j}, 1 \leq j \leq m$, each of which have neighbours $N^{j}$ (outside of $\{v\} \cup \nb_{\overline{\G}}(v)$) such that $N^{1} \subset N^{2} \subset \dots \subset N^{m}$. 

For each $A^{j}$, it will be a root for exactly one collection of consecutive sets $\hat{A}^{j} = A^{[j_l, j_h]}$, and we consider the following list of independences $\mathbb{L}$:
\begin{align*}
    v &\indep A^{j} \mid H^{j}, L^j, N^j, T^j, \qquad j \in 1,\ldots,m. 
\end{align*}
We will prove that for each $j$, the following independence holds for each $j$ from the above list of independences:
$$v \indep \hat{A}^j \mid N^{j_l}, T^j.$$
We will proceed by two inductions, an outer induction on the number of vertices in $\G$, and an inner induction on the lengths of $\hat{A}^j$.  The base case for the outer induction is trivial, since a graph with one vertex has no independences.  For the inner induction, the base case is $\abs{\hat{A}^j} = 1$, so $\hat{A}^j = \{A^j\}$, $j_l = j$ and $H^j = L^j = \emptyset$, hence it is in the given list of independences.

For the inner induction step, suppose $\hat{A}^j = H^j \cup L^j \cup A^j$, where $H^j$ and $L^j$ are rooted at $k$ and $t$ (non-empty), respectively. Then by the induction hypothesis, we have:
\begin{align*}
    v &\indep H^j \mid N^{k_l}, T^k &
    v &\indep L^j \mid N^{t_l}, T^t,
\end{align*}
but by Lemma \ref{L^jrelation}, this is equivalent to:
\begin{align}
    v &\indep H^j \mid N^{j+1}, L^j, T^j \label{H*}\\
    v &\indep L^j \mid N^{j_l}, T^j. \label{S*}
\end{align}
Since $v$ is the last vertex, any independence not involving $v$ will hold by the induction hypothesis, hence we have:
\begin{align}
    H^j \indep L^j, N^{j+1} \mid T^j \label{I1}\\
    H^j, A^j \indep L^j, N^j \mid T^j. \label{I2}
\end{align}
Combining (\ref{I1}) and (\ref{H*}), we have:
\begin{align}
    H^j \indep \{v\}, L^j, N^{j+1} \mid T^j. \label{S1}
\end{align}
Then marginalizing $N^{j+1}$ to $N^j$ and moving $L^j \cup N^j$ to conditioning variables, we have:
\begin{align}
    v \indep H^j \mid L^j, N^j, T^j. \label{S2}
\end{align}
Now from the given list of independences we have: $$v \indep A^j \mid H^j, L^j, N^j, T^j.$$ Putting this with (\ref{S2}) we obtain:
\begin{align}
    v \indep H^j, A^j \mid L^j, N^j, T^j. \label{S3}
\end{align}
Using (\ref{I2}), (\ref{S3}) then changes to:
\begin{align}
    \{v\}, L^j, N^j \indep H^j, A^j \mid T^j. \label{S4}
\end{align}
marginalizing $N^{j_l+1}$ to $N^{j}$ and moving $L^j \cup N^{j_l}$ to the conditioning variables, we have:
\begin{align}
    v \indep H^j, A^j \mid L^j, N^{j_l}, T^j. \label{S5}
\end{align}
Finally, combining (\ref{S5}) and (\ref{S*}), we obtained the required independence:
\begin{align}
    v \indep H^j, A^j, L^j \mid N^{j_l}, T^j.
\end{align}

Next we show that for any disconnected sets involving $v$, we can deduce the associated independence, that is, the global Markov property. We use the independence (\ref{S4}) that appears in the previous deduction and combine it with an independence from the induction hypothesis, $T^j \indep H^j, A^j$, to get:
\begin{align}
   \{v\}, L^j, T^j, N^j \indep H^j, A^j. \label{S5a} 
\end{align}
Based on this, we prove by induction that for every $j$, we have:
\begin{align}
    \{v\}, L^j, T^j, N^j \indep A^{[j,m]}. \label{S6}
\end{align}

Before the induction, one should notice that for any $i > j$, $L^i \cup T^i \supseteq L^j \cup T^j$.

An order can be given to blocks based on the collection of blocks in which they are rooted. If $H^j, L^j$ are rooted at $A^k, A^t$ respectively, then we say that $j$ is precedes both $k$ and $t$, and apply the induction on this order. The base case is for the root of $A^{[1,m]}$, say $A^j$. In this case, $H^j \cup A^j = A^{[j,m]}$, and it is true.

Suppose for a block $A^j$, any other block $A^i$ such that $i$ precedes $j$ in the above order so we have $\{v\} \cup L^i \cup T^i \cup N^i \indep A^{[i,m]}$. If $H^j \cup A^j \neq A^{[j,m]}$, then $A^{j_h+1}$ must be a block that also precedes $A^j$ in the order, and so we also have $\{v\} \cup L^{j_h+1} \cup T^{j_h+1} \cup N^{j_h+1} \indep A^{[j_h+1,m]}$ and (more importantly) $L^{j_h+1} \cup T^{j_h+1} \supseteq A^{[j, j_h]} \cup L^j \cup T^j$ (it is possible these sets are equal). now put $A^{[j, j_h]} = H^j \cup A^j$ into the conditioning variables and marginalize $N^{j_h+1}$ to $N^{j-1}$ and other irrelevant variables, to obtain:
\begin{align}
    \{v\} \cup L^j \cup T^j \cup N^j \indep A^{[j_h, m]} \mid H^j \cup A^j.
\end{align}
Then combining this with (\ref{S4}), the result holds for all $j$. 

Now we prove that for every disconnected set $C$ involving $v$, the associated independence $\{v\} \cup D \indep A$ is true, where $\{v\} \cup D$ is the district for $\{v\}$ in the subgraph induced by $C$. We only need to consider vertices in $\{v\} \cup A^{[1,m]} \cup N^m \cup \dots \cup N^1$, as other vertices will definitely link $v$ and any of its neighbours. This $A$ must be a subset of $A^{[1,m]}$. Suppose $A^j$ is the least block that have non-empty intersection with $A$, so $A \subseteq A^{[j,m]}$.

Consider the vertices in $D$. Firstly notice that it cannot contain any vertices in $N^i$ or $A^i$ where $i > j$. The first one is clear, suppose the second is not true, consider the bidirected path from $A^i$ to $v$, the last vertex before $v$ is a sibling of $v$ and it must belong to some of the $N^k$. Now if $k > j$, then as $A^j$ is connected to $N^k$, this contradicts to the assumption of $j$. so $k < j$ but then as $A^i$ is a neighbour of $N^k$ in the dual graph, they are disconnected, so there must be other blocks lower than $A^j$ involved. however as we have $L^i \cup T^i \supseteq L^j \cup T^j$, this means that any block connect to $A^i$ must also connect to $A^j$ in the bidirected graphs and then $A^j$ lies in the same districts as $v$, contradiction. Thus $D$ can only contain vertices in $L^j \cup T^j \cup N^j$, and we have the required independence.

We are left to prove that the list of independence $v \indep A^j \mid H^j \cup L^j \cup T^j \cup N^j$ are non-overlapping and associated with every disconnected set. For the disconnected set $M$, suppose $D \cup \{v\} = \dis_{M}(i)$ and we have this independence $\{v\} \cup D \indep A$ with $A^j$ the least block. Then consider the root of $A^m, \dots, A^j$ (by Lemma \ref{root of subset}), say $A^i$, and also consider $\hat{A}^j$ then it is clear that the sets associates with the independence $v \indep A^i \mid H^i \cup L^i \cup T^i \cup N^i $ contains this disconnected set.

To show there is no overlap, assume that $i < j$ (in the numerical sense).  Then $A^i$ will 
appear in the conditioning set for $j$ only if it is in $L^j \cup T^j$,
which by definition implies that $A^j$ and $A^i$ are fully connected.
However $A^j$ will appear in the conditioning set for $A^i$ only if
$A^j \subseteq H^i$, which implies that they are completely
disconnected.  Since at most one of these conditions can be true, there is 
no overlap in the sets generated by these independences.
%
%
\end{proof}

\end{appendix}

\end{document}